\title{Competitive Online Virtual Cluster Embedding Algorithms}
\author{Feras Iwan Fattohi}
\date{\today}

\documentclass[a4paper,titlepage,twoside,12pt]{article}
\usepackage[a4paper, total={6in, 9in}]{geometry}
\usepackage{setspace}   
\usepackage{graphicx}
\usepackage{hyperref}
\usepackage[english, ngerman]{babel} 
\usepackage{amssymb}
\usepackage{amsmath}
\usepackage{pgfplots}
\usepackage{mdframed}
\usepackage{amsthm}
\usepackage{verbatim}
\usepackage{enumitem}
\usepackage{algorithm2e}
\theoremstyle{plain}
\newtheorem*{duality*}{Duality theorem of linear programming}
\newtheorem{theorem}{Theorem}
\newcommand\myeq{\mathrel{\stackrel{\makebox[0pt]{\mbox{\normalfont\tiny weak duality}}}{\leq}}}
\author{\Large%
	Iwan Feras Fattohi\\{\normalsize Matriculation Number %
		349828} \\}
\title{
	\vspace{-15em}
	\centering
	\begin{figure*}[!htb]
		\begin{minipage}{.17\linewidth}
			\includegraphics[width=1.0\linewidth]{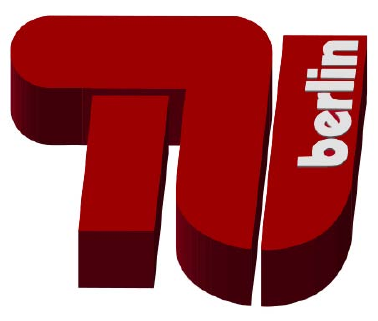}
		\end{minipage}
		\begin{minipage}{.55\linewidth} 
			\centering
			Internet Network Architectures\\
			Faculty IV Electrical Engineering\\ and Computer Science\\
			Technical University of Berlin
		\end{minipage}
		\begin{minipage}{.25\linewidth}
			\includegraphics[width=1.0\linewidth]{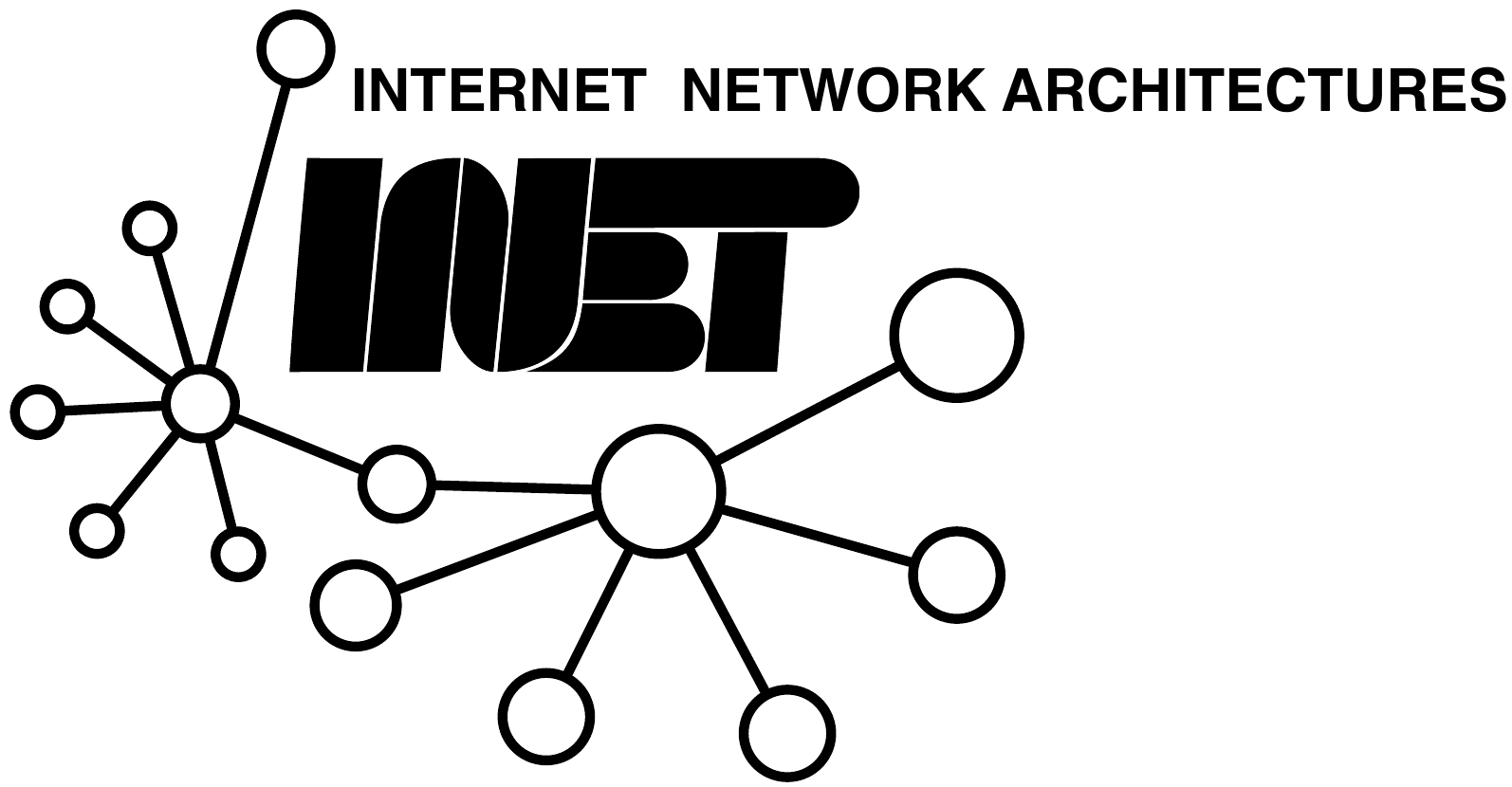}
		\end{minipage}
	\end{figure*}
	\vspace{8em}
	{\huge\textbf{%
		Competitive Online Virtual Cluster Embedding Algorithms
		}}\\ \vspace{1.7cm}{\LARGE{\bfseries Master Thesis}\\\vspace{0.6cm}
	}}
	\date{
		16. February 2018\\
		\vspace{9em}
		\large
		\textbf{Advisor:} \\
		\vspace{1em}
		Prof. Dr. Anja Feldmann\\
		Univ.-Prof. Dr.sc. Stefan Schmid\\[0.5cm]
		Matthias Rost, M.Sc.
	}
\begin{document}
\thispagestyle{empty}
\enlargethispage{20\baselineskip}
\oddsidemargin3.3mm
\addtolength{\textwidth}{3em}
\pagenumbering{gobble}
\maketitle
\addtolength{\textwidth}{-3em}
\oddsidemargin-3.3mm
%
\cleardoublepage
\thispagestyle{empty}\section*{Eidesstattliche Erkl\"arung}
\vspace{2.5cm}Hiermit erkl\"are ich, dass ich die vorliegende Arbeit selbstst\"andig und \\eigenh\"andig sowie ohne
unerlaubte fremde Hilfe und ausschlie"slich unter \\Verwendung der aufgef\"uhrten Quellen und
Hilfsmittel angefertigt habe. \\
\\ \vspace{1.2cm}\\
\line(1,0){250} \\
\\
$~{}~{}~{}~{}~{}~{}$Berlin, den 16. Februar 2018
%
\thispagestyle{empty}
\cleardoublepage
\thispagestyle{plain}
\selectlanguage{english}
\oddsidemargin3.3mm
\begin{abstract}
In the conventional cloud service model, computing resources are allocated for tenants on a pay-per-use basis. However, the performance of applications that communicate inside this network is unpredictable because network resources are not guaranteed. To mitigate this issue, the virtual cluster (VC) model has been developed in which network and compute units are guaranteed. Thereon, many algorithms have been developed that are based on novel extensions of the VC model in order to solve the online virtual cluster embedding problem (VCE) with additional parameters. In the online VCE, the resource footprint is greedily minimized per request which is connected with maximizing the profit for the provider per request. However, this does not imply that a global maximization of the profit over the whole sequence of requests is guaranteed. In fact, these algorithms do not even provide a worst case guarantee on a fraction of the maximum achievable profit of a certain sequence of requests. Thus, these online algorithms do not provide a competitive ratio on the profit. 

In this thesis, two competitive online VCE algorithms and two heuristic algorithms are presented. The competitive online VCE algorithms have different competitive ratios on the objective function and the capacity constraints whereas the heuristic algorithms do not violate the capacity constraints. The worst case competitive ratios are analyzed. After that, the evaluation shows the advantages and disadvantages of these algorithms in several scenarios with different request patterns and profit metrics on the fat-tree and MDCube datacenter topologies. The results show that for different scenarios, different algorithms have the best performance with respect to certain metrics.
\end{abstract}
\selectlanguage{ngerman}
\begin{abstract}
	Im herk\"ommlichen Cloud Service Modell werden IT-Infrastrukturen wie Rechenleistung, Speicherplatz und Anwedungssoftwares auf Basis des Pay-per-Use Modells bereitgestellt. Jedoch ist die Kommunikation zwischen Anwendungen innerhalb des Could-Datacenters nicht zuverl\"assig, weil dessen Netzwerk Ressourcen nicht gew\"ahr-leistet sind. Um diesem Problem entgegenzukommen, wurde das Virtual Cluster (VC) Modell entwickelt, das garantierte Rechenleistung und Netzwerk Ressourcen betrachtet. Auf dessen Basis wurden viele Algorithmen entwickelt, die dieses Modell um neuartige Parameter erweitern. Dabei l\"osen diese das online Virtual Cluster Embedding Problem (VCE). In der online Variante wird der Ressourcenverbrauch des Service Providers pro Anfrage minimiert. Dies steht in Korrelation mit einer Profitmaximierung pro Anfrage. Jedoch folgt daraus nicht, dass eine garantierte Profitmaximierung \"uber die gesamte Anfragesequenz erreicht wird. Dazu kommt, dass diese online Algorithmen keine Worst-Case Garantie auf den Profit haben. Daher sind diese Algorithmen nicht kompetitiv in Bezug auf den Profit.
	
	In dieser Masterarbeit werden zwei kompetitive online VCE Algorithmen und zwei Heuristiken  vorgestellt. Die kompetitiven online VCE Algorithmen haben verschiedene ``competitive ratios'' bez\"uglich der Zielfunktion und den Kapazit\"atsbeschr\"ankungen, wohingegen die Heuristiken keine Kapazitätengrenzen verletzen. Unter anderem werden die Worst-Case ``competitive ratios'' hergeleitet. Danach werden anhand der Evaluation die Vor- und Nachteile der Algorithmen in mehreren Szenarien mit verschiedenen Profit Typen und Metriken untersucht. Das Ergebnis der Evaluation zeigt, dass f\"ur die verschiedenen Szenarien, verschiedene Algorithmen ihre beste Leistung aufweisen bez\"uglich bestimmter Metriken.
\end{abstract}
\selectlanguage{english}
\pdfbookmark[1]{Inhaltsverzeichnis}{Inhalt}
\pagenumbering{roman}
\tableofcontents\thispagestyle{plain}
\cleardoublepage
\pagenumbering{arabic}
\oddsidemargin-3.3mm
\oddsidemargin9.5mm
\evensidemargin-3mm
\hypersetup{breaklinks, linkcolor=darkred}
\section{Introduction}\label{intro}
Cloud computing deals with the allocation of remote services over the Internet in data centers. This term includes three different service models, each of them being applied for certain scenarios. These are Infrastructure-, Platform- and Software-as-a-Service (IaaS, PaaS and SaaS) \cite{vaquero2008break}. A tenant can lease a software (SaaS), an environment for developing software (PaaS) or computing power and memory (IaaS) as a remote service on demand.
These resources are offered by a service provider which is maintaining them in her substrate network. The advantage of such a network is that the resources can be provisioned over the Internet and can be freed again if they are not needed anymore. 

Considering a tenant who does not have enough local resources for executing a certain software, she can allocate remote resource of the cloud if the cloud has enough available resources to serve the request. After the execution of the software, the resources are freed again and the provider earns money based on a pay-per-use business model \cite{vaquero2008break}. This is beneficial for the tenant because instead of buying new hardware for her local network, the tenant could execute her software without additional local resources by using remote resources on demand.
 
The compute units (including RAM and HDD) that are served by the provider are guaranteed. However, there are in general no such guarantees on the network connecting the allocated elements. This becomes an issue as soon as the allocated virtual machines of a certain tenant need to communicate to each other since the performance of the network becomes unpredictable as soon as its links are congested. Thus, network resources inside the cloud network are not guaranteed. 

This problem was a hot topic in the past and therefore models incorporating guaranteed bandwidth demands were studied. One of the most prominent is the \emph{virtual cluster} model. The virtual cluster (VC) was introduced in 2012 by \cite{vc} and describes a simple way to model a tenant's request and is depicted in Figure \ref{fig58}.
\begin{figure}[ht]
	\centering
	\includegraphics[width=0.3\textwidth]{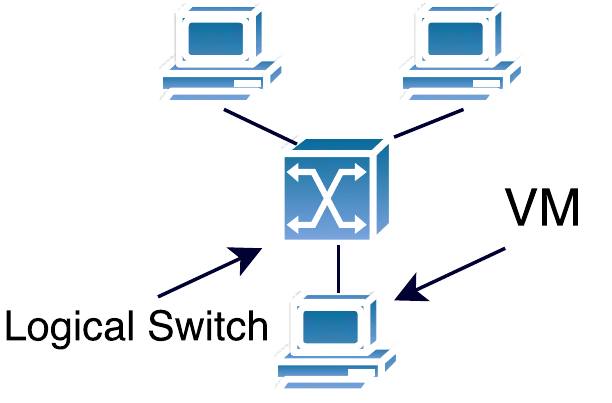}
	\caption{Virtual Cluster with 3 VMs}
	\label{fig58}
\end{figure}
This paper also introduced the virtual cluster embedding (VCE) problem in which the request in terms of a constructed virtual cluster has to be embedded into the substrate network of the provider ensuring guaranteed compute and network resources. The VCE problem is an optimization problem with the objective to minimize the resource footprint when embedding a VC request. 

Based on this model, Rost et al. developed in \cite{vcace} the VC-ACE algorithm which finds an optimal solution for the VCE problem in polynomial time. Many works have extended the initial model by additional parameters like survivability \cite{survivability}, a heterogeneous bandwidth demand model \cite{heterogeneous} and an elastic resource reservation property \cite{kraken}. All these extensions have in common that their algorithms are solving the virtual cluster embedding problem with respect to some additional parameters.

As already indicated before, a tenant's request arrives in real-time in an online fashion and shall be served as fast as possible. This means that based on the current state of the substrate network, the provider has to find an allocation for the requested amount VMs if enough resources are available since previous requests have already reserved resources. Furthermore, preemption is generally not considered, meaning  the embedding algorithm cannot interrupt a current job, such that previously embedded requests cannot be undone. Hence, the virtual cluster embedding problem is an online problem. Such an online scenario is shown in the following example, depicted in the Figures \ref{fig34} and \ref{fig35}. Figure \ref{fig36} shows the corresponding offline version.

As seen in Figure \ref{fig58}, a virtual cluster has a star-shaped structure, there is a logical switch $LS$ in the center serving for connecting the virtual machines (VMs) which are attached to it. Each of the $\mathcal{N}_i$ VMs of the virtual cluster request $r_i \in R$ ($R$ being the sequence of requests) asks for a certain amount of compute units $\mathcal{C}_i$. The logical switch itself does not consume any resources, it only connects the VMs. Each virtual link connecting a VM to the logical switch demands a bandwidth $\mathcal{B}_i$. These links are undirected. Lastly, a request $r_i$ pays a benefit $b_i$ to the provider when the provider embeds the request. The profit is then the sum of all benefits over the sequence of requests.
\begin{figure}
	\centering
	\includegraphics[width=1.0\textwidth]{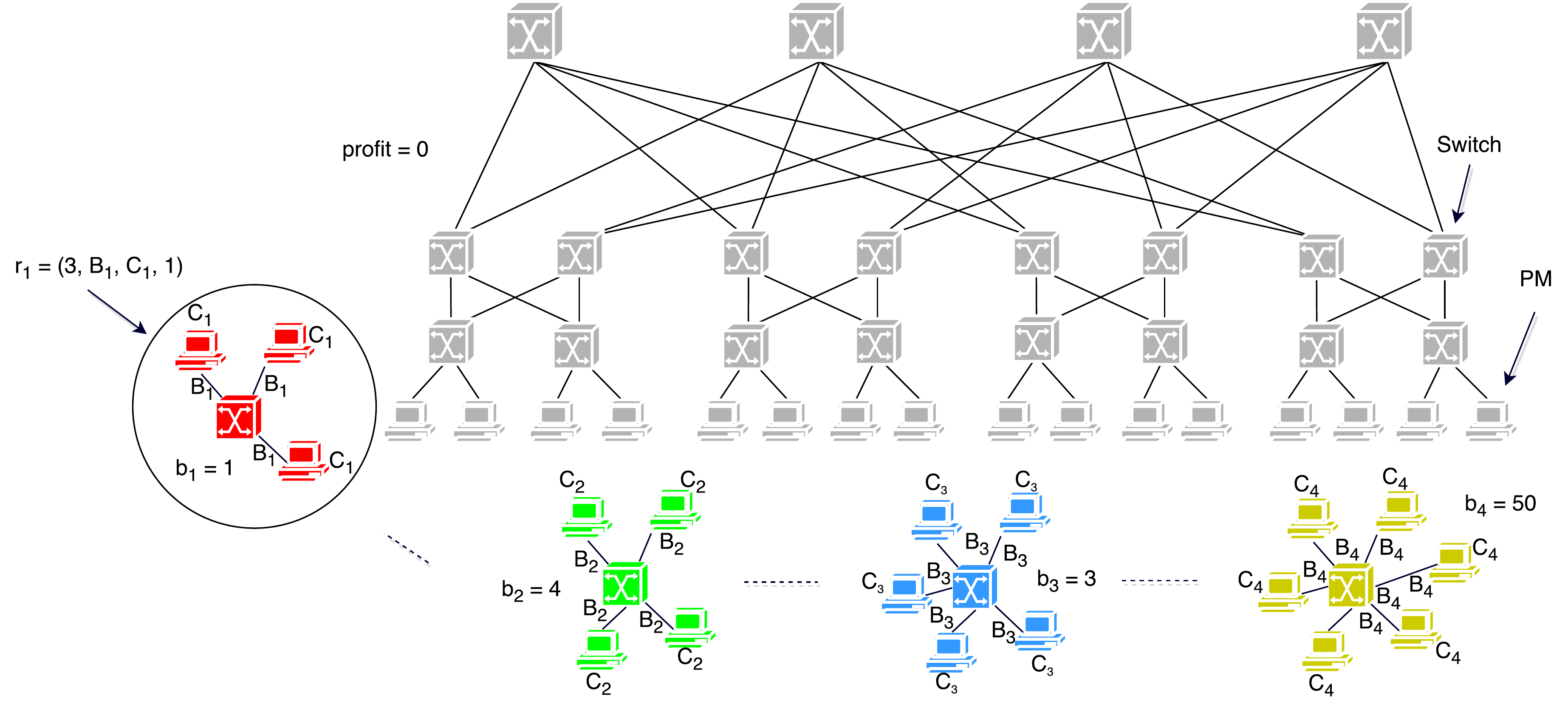}
	\caption{Sequence of requests 1-4 arriving online in this order}
	\label{fig34}
\end{figure}
Figure \ref{fig34} shows the initial situation of the online algorithm. The substrate network, here a fat-tree \cite{fattree}, consisting of switches and physical machines (PMs), is empty. No reservations are performed yet and the sequence of requests $\langle r_1,r_2,r_3,r_4 \rangle$, with $r_i = (\mathcal{N}_i, \mathcal{B}_i, \mathcal{C}_i, b_i)$, arrives in an online fashion, each paying the provider a certain benefit $b_i$. In this example, we assume that one VM of a VC request consumes a whole PM of the substrate network. The goal of the online algorithm is to minimize the amount of resources (PMs and physical links) consumed by a virtual cluster request, if it can be embedded. By minimizing the resource footprint per request, more resources are left to accept even more requests. The latter results in a higher profit, so that there is a correlation between minimizing the resource footprint and maximizing the profit.

Upon the arrival of request 1 (red-marked), the online algorithm decides that there is enough space for the VC in the substrate network and therefore embeds the request. The online algorithm processed the request in real-time, directly after its arrival. Since request 1 is accepted, the profit of the provider increases by the benefit of this request. 
\begin{figure}
	\centering
	\includegraphics[width=1.0\textwidth]{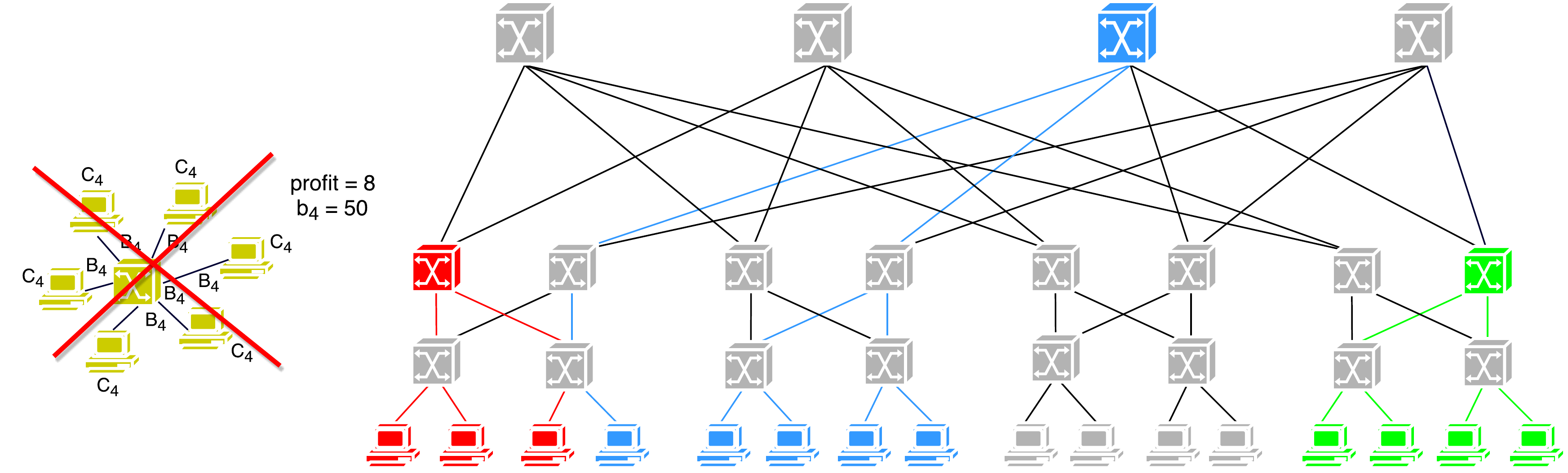}
	\caption{4th request is rejected because of missing available resources}
	\label{fig35}
\end{figure}
The same procedure is executed for the next two requests. However, at the arrival of request 4 in Figure \ref{fig35}, the algorithm decides to reject the request because there are not enough PM resources left. The VC request demands five VMs but the substrate network has only four free PMs left. Though, it would have been beneficial for the provider to embed this request since its benefit of 50 is higher than the benefit of all previous requests. Thus, accepting \emph{any} request, that can be embedded, does not provide the highest overall profit for the provider in this example. 

In contrast to this, if the provider would have known the sequence of requests over the whole time axes, which is the case in the offline setting of this problem, she could have chosen a subset of requests which maximizes her profit. An example embedding in which requests are also rejected is depicted in Figure \ref{fig36}. The optimal offline algorithm now chooses the best subset of that sequence that maximizes the profit for the provider and embeds it, resulting in a profit of 57 in comparison to a profit of 8 for the online algorithm. 
\begin{figure}
	\centering
	\includegraphics[width=1.0\textwidth]{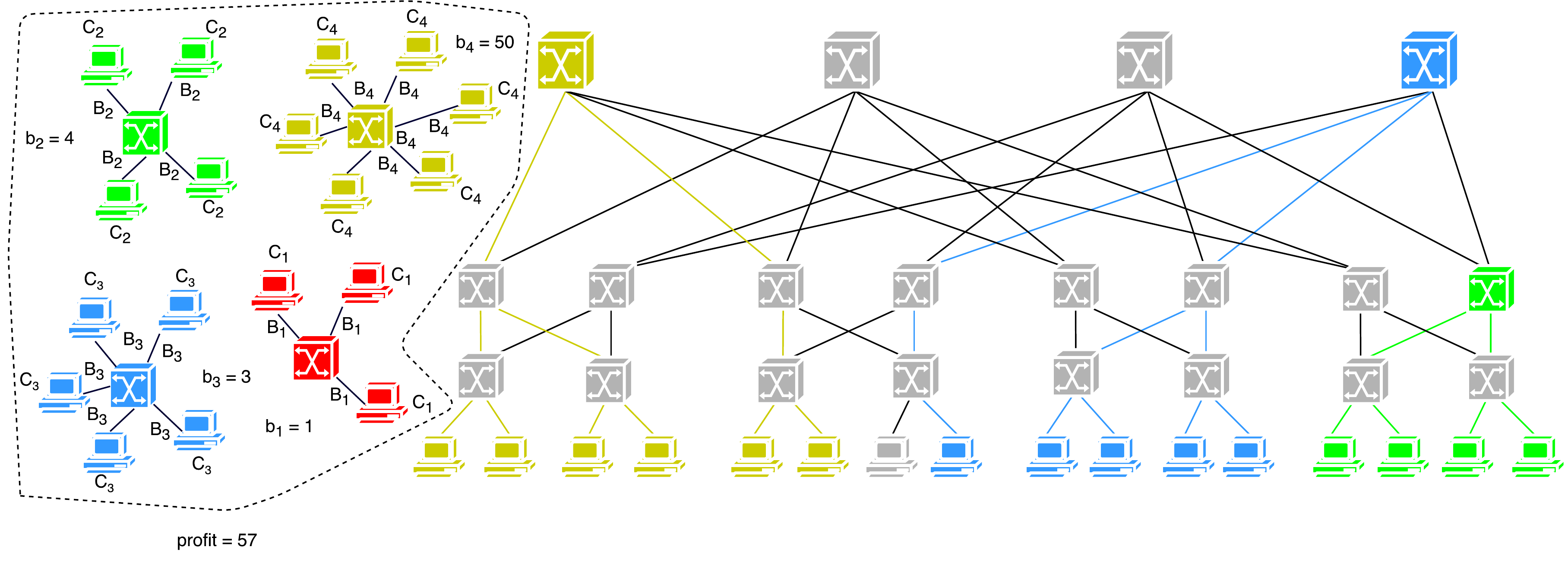}
	\caption{Optimal offline algorithm chooses the benefit-maximizing subset}
	\label{fig36}
\end{figure}
Generally speaking, the optimal offline algorithm finds the best solution for the provider. 

In the online version of this problem, the online algorithm cannot select the best subset of the sequence of requests, because it does not know the whole sequence of requests in advance, but has to decide in real-time upon the arrival of a request whether to embed it or not. Thus, it is not guaranteed that the online algorithm's solution is equal to the optimal offline solution. Since all the previously referenced embedding algorithms are online VCE algorithms they do not guarantee to reach the optimal offline solution. 

On the one side, these online algorithms do not ensure that the resource footprint over the whole sequence of requests is minimized, but only the resource footprint per request. They accept \emph{any} request that arrives online and can be feasibly embedded, which means there are enough resources to serve the request. These algorithms do not include a \emph{rejection} functionality.

On the other side, minimizing the resource footprint per request also does not yield the maximum profit as seen in the example. These algorithm do not guarantee that a certain fraction of the optimal offline profit is achieved since any feasibly embedable request is accepted. An accepted request may have the worst possible benefit, as seen with $r_1$. More formally, let $\mathtt{profit}_{\mathrm{online}}$ be the profit earned by the online algorithm and ${\mathtt{profit}_{\text{opt,offline}}}$ be the profit earned by the optimal offline algorithm. Then, the online embedding algorithm of the example does not guarantee that $\frac{\mathtt{profit}_{\text{opt,offline}}}{c} \leq \mathtt{profit}_{\mathrm{online}}$ holds for any request sequence, which means that the online algorithm does not ensure in general a $c$-fraction of the optimal offline profit. An online algorithm which fulfills this inequality for any input sequence and a fixed $c$ is called $c$-\emph{competitive}, i.e. it provides a \emph{competitive ratio} on the profit. Thus, the online embedding algorithm of the example is not competitive. In addition, the referenced online virtual cluster embedding algorithms also do not provide a competitive ratio since they accept any feasibly embedable request. However, a competitive online VCE algorithm may be beneficial for practical applications since it provides a guarantee on the profit although future requests are unknown. 

To develop such an algorithm, we apply the primal-dual approach for online problems developed in \cite{survey}. They present a framework for developing competitive online algorithms via the primal-dual approach for linear programs. They show for certain problems that there is a trade-off between the competitiveness on the objective function and finding a feasible solution. In the case of the VCE problem, the validity of the capacity constraints of the nodes and edges in the provider's networks is affected. So, a competitive online VCE algorithm also involves a competitiveness on the violation of the capacity constraints. Hence, they introduce $(c_1,c_2)$-competitiveness where $c_1$ is the competitive ratio on the objective function and $c_2$ the one on the capacity constraints.

The contribution of this thesis are two competitive online virtual cluster embedding algorithms and two heuristic algorithms. The first competitive online algorithm is the \textbf{G}eneral all-or-nothing \textbf{V}Net Packing \textbf{O}nline \textbf{P}acking (GVOP) algorithm based on \cite{gvop} and the second one is the \textbf{C}ompetitive \textbf{O}nline \textbf{V}irtual \textbf{C}luster \textbf{E}mbedding (COVCE) algorithm which is a novel contribution of this thesis. Both algorithms provide a competitive ratio on the profit and the capacity constraints. The difference between them is that each focuses on one of the parameters. The GVOP algorithm is (2, $\mathrm{log}_2 (1+ 3\cdot (\mathrm{max}_{i,k}w(i,k)\cdot \mathrm{max}_{i}b_i))$)-competitive, where $w(i,k)$ is the sum of the allocations made by embedding $k$ for request $i$. The COVCE algorithm is ($4 \cdot \mathrm{ln}(1+|G_S|\cdot C_{max}\cdot \alpha) + 1, 2)$-competitive, where $G_S$ is the size of the substrate graph, $C_{\mathrm{max}}$ is the maximum capacity of any node and edge and $\alpha$ being the maximum benefit over the sequence of requests. Additionally, the Greedy VC-ACE algorithm is implemented which receives a sequence of requests in an online fashion and greedily accepts any requests applying the VC-ACE algorithm with respect to residual capacities. Furthermore, the COVCEload algorithm is presented which is a combination of the Greedy VC-ACE algorithm and the COVCE algorithm.

Besides the worst-case guarantees of the competitive online algorithms, we study their performance in several experiments in order to analyze which competitive ratio-tuple is more applicable in realistic scenarios. In addition, the heuristic algorithms are also compared to the competitive online algorithms in order to find out whether a non-competitive online algorithm can compete with the competitive online algorithms. The main questions of this thesis are the following:

\begin{enumerate}
	\item How big are the capacity violations of the GVOP and the COVCE algorithm?
	\item How big are the rejection ratios of the GVOP and the COVCE algorithm?
	\item Which approach yields a higher profit per resource violation ratio?
	\item Is there a sweet spot between competitive and non-competitive algorithms?
	\item When are (non)-competitive VCE algorithms better than (non)-competitive VCE algorithms?
\end{enumerate}

The thesis is divided into 7 sections. At first some background knowledge is discussed in order to become familiar with the VC model and the VCE in Section \ref{model} and the primal-dual approach in Section \ref{background}. Afterwards related work is discussed in Section \ref{relatedwork} to get an overview of the state of the art. Next, the theoretical contribution of this thesis is described in Section \ref{theory}. After that, in the evaluation in Section \ref{evaluation}, the competitive  and non-competitive online algorithms are tested and compared to each other in several scenarios in order to find out, which algorithm shows advantages or disadvantages in certain scenarios. Lastly, some ideas for future work are presented in Section \ref{futurework} and the results of this thesis are summarized in the conclusion Section \ref{conclusion}.

\section{Model}\label{model}
In the following, the virtual cluster model together with the virtual cluster embedding problem are described in order to introduce the setting and the notation. The linear program, the GVOP algorithm \cite{gvop} and the Competitive Online Virtual Cluster Embedding (COVCE) algorithm are working based on this model.
\subsection{Virtual Cluster}
The sequence $R = \langle r_1,...,r_l\rangle$ describes the VC requests. The request $r_i \in R$ contains information on the $i$-th VC  that shall be embedded, $r_i = \{\mathcal{N}_i, \mathcal{B}_i,\mathcal{C}_i,b_i\}$. A virtual cluster $VC_i$($\mathcal{N}_i, \mathcal{B}_i,\mathcal{C}_i$) \cite{vc} has a star-shaped topology and is formally described as an undirected graph $VC_i = (V_{VC_i}, E_{VC_i})$  with $\mathcal{N}_i \in \mathbb{N}$ virtual machines being connected to a logical switch $LS$ which form the set of nodes $V_{VC_i} = \{1,2,...,\mathcal{N}_i,LS\}$ of $VC_i$. Each node in $V_{VC_i}\setminus\{LS\}$ has a capacity of $\mathcal{C}_i$ which can be interpreted as compute units, i.e. $\mathcal{C}_i$ CPU or memory units of a computer/server. Lastly, the parameter $b_i$ describes the benefit obtained by the provider for embedding $r_i$. The logical switch only serves as a connector between the VMs to build the virtual cluster and does not consume any resources. The nodes are connected over the undirected virtual links $E_{VC_i}$ = $\{\{j, LS\}|  j \in \{1,2,..,\mathcal{N}_i\}\}$, where each link has the bandwidth $\mathcal{B}_i \in \mathbb{N}$. The virtual machines, except the $LS$, of the virtual cluster will be referred to as virtual nodes or VMs and its links as virtual links. 

The virtual cluster model is simple because the tenant only has to determine how many VMs she needs, how many compute units to assign to each VM and how much bandwidth is needed between the VMs on this star-shaped structure. This is enough information to construct the virtual cluster and to start the virtual cluster embedding algorithm.

The disadvantages of this model are that the VC is a star-shaped topology, so that the tenant cannot choose an arbitrary topology when applying VCE algorithms. In addition, the compute units are the same for each VM and the bandwidth is the same on each link, too. In a real world scenario, the bandwidth requirements between each pair of VMs may be different, so that the fixed bandwidth $\mathcal{B}_i$ for the VC model can be derived by maximizing over the bandwidths of all virtual links. This can be wasteful because some links receive more resources than they actually need. Thus, there is a trade-off between the simplicity of the model and therefore easily constructing requests and optimally embedding them, and the fact that this model might consume a redundant amount of resources.
\subsection{Substrate Network}
The substrate network is the physical network of a service provider who offers its available resources to the tenants. It is the undirected graph $S = (V_S, E_S, \mathtt{cap}, \mathtt{cost})$ with node and link capacities $\mathtt{cap} : V_S \cup E_S \rightarrow \mathbb{N}$. Additionally, a cost parameter is considered which describes how much the costs for the provider are for using a unit of a certain resource, formally $\mathtt{cost} : V_S \cup E_S \rightarrow \mathbb{R}_{\geq 0}$. More precisely, this cost parameter can be interpreted in different ways.
For instance, costs can be seen as a load balancing parameter. To this end, the link costs are set inversely proportional to the available (residual) capacity causing a load distribution onto several links. In addition, high costs on a link mitigate link over-subscription because the algorithm prefers links with lower costs. 
This is used in the COVCE and the GVOP algorithm to limit the over-subscription of nodes and edges in the substrate network. In general, high costs show that not much available capacity is left or that the node or the edge is already over-subscribed.
\subsection{Virtual Cluster Embedding}
After having introduced the virtual cluster and the substrate network, the virtual cluster embedding (VCE) now connects both these components. It describes a mapping of the virtual cluster's VMs and the logical switch onto the nodes of the substrate network and a mapping of the corresponding virtual links onto paths of physical links within the substrate network. The virtual link $\{k, LS\}, k \in V_{VC_i} \setminus \{LS\}$ of the VC request $r_i$ is embedded as a path leading from the mapping location of $k$ to the mapping location of $LS$. The path consists of either a set of physical edges connecting these two physical nodes or is empty, in case that $LS$ and $k$ are mapped onto the same physical node. A \emph{valid} VCE does not violate any capacities of substrate nodes and links. In the general case, there are several ways to embed the VC of a given request $r_i$ into substrate graph $S$. Formally speaking, $M(r_i) = \{m_i^1, m_i^2,...,m_i^j\}$ denotes the set of all valid VCEs for request $r_i$. Let $\mathcal{P}(E_S)$ be the power set of the set of physical edges. An embedding $m_i^k \in M(r_i)$ is the tuple $m_i^k = (m_{V,i}^k, m_{E,i}^k)$ containing the mappings of the virtual nodes and links onto the physical nodes and links, $m_{V,i}^k : V_{VC_i} \rightarrow V_S$ and $m_{E,i}^k : E_{VC_i} \rightarrow \mathcal{P}(E_S)$. As already indicated, the virtual cluster embedding problem deals with finding a valid VCE $m_{V,i}^k$ with minimum cost for a given request $r_i$. Formally:
\begin{alignat}{7}
\mathcal{C}_i \cdot \sum_{v \in V_{VC_i}\setminus \{LS\}} \mathtt{cost}(m_{V,i}^k(v)) +  \mathcal{B}_i \cdot \sum_{e' \in E_{VC_i}, e\in m_{E,i}^k(e')} \mathtt{cost}(e).  
\end{alignat}
Thus, solving the VCE problem outputs an optimal embedding that minimizes the summed costs of the allocated node and link resources. To ensure no violations of the resource capacities, the following constraints need to hold:
\begin{alignat}{7}
\sum_{v'\in V_{VC_i}\setminus\{LS\}, v = m_{V,i}^k(v')} \mathcal{C}_i \leq \mathtt{cap}(v) \quad\mathrm{and} \sum_{e' \in E_{VC_i}, e \in m_{E,i}^k(e')} \mathcal{B}_i \leq \mathtt{cap}(e)
\end{alignat}
Before dealing with approaches to solve this problem, some fundamental and essential background knowledge to linear programming, which is a method to deal with certain types of optimization problems, is briefly explained.
\section{Background}\label{background}
In this section, the necessary background knowledge will be discussed. Our competitive algorithms, presented in Section \ref{theory}, are inspired by the survey \cite{survey}. For this reason, their approach of how to design competitive online algorithms is explained in the following. At first linear programming is briefly introduced because it is the basis for the duality theorem and the consequential primal-dual approach for online problems of \cite{survey}.

\subsection{Linear Programming and Duality}
To begin with, the most fundamental term \emph{linear program} (LP) is briefly explained. Linear programming deals with optimization problems in which a linear objective function shall be either minimized or maximized with respect to a set of linear constraints over a set of continuous variables. On the one hand, if the objective function shall be maximized, the final value of the objective function can be interpreted as a maximized profit or benefit. On the other hand, if it shall be minimized, one can interpret it as a minimized cost. The difference between an \emph{integer linear program} (ILP) and an LP is that the domain of the variables is restricted to integral numbers, so that the variables are discrete and not continuous anymore. 

To bring the terms linear programming and duality together, consider the following Example \ref{fig22} inspired by \cite{linearprogramming}.
\begin{figure}[ht]
\begin{alignat}{7}
&\text{maximize the value} &\quad 4x_1 + 2x_2 \nonumber\\
&\text{among all vectors } (x_1,x_2) \in \mathbb{R}^2& \nonumber\\
&\text{satisfying the constraints } & x_1 + x_2 &\leq 2\\
&									& 2x_1 + 3x_2  &\leq 9\\
&									& x_1 + 2x_2 &\leq 3\\	
& 									& x_1 &\geq 0\nonumber\\ 
&									& x_2 &\geq 0\nonumber						
\end{alignat}
\caption{Linear Program $\tilde{P}$}
\label{fig22}
\end{figure}
A vector $(x_1,x_2) \in \mathbb{R}^2$ that satisfies the constraints (3)-(5) is a \emph{feasible} solution to the linear program $\tilde{P}$. $x_1 \geq 0$ and $x_2 \geq 0$ define the domain of the $x_i$ variables. To illustrate the previously mentioned definitions, consider Figure \ref{fig21}.
\begin{figure}[ht]
\centering
\begin{tikzpicture}

\draw[gray!50, thin, step=0.5] (-1,-3) grid (5,4);
\draw[very thick,->] (-1,0) -- (5.2,0) node[right] {$x_1$};
\draw[very thick,->] (0,-3) -- (0,4.2) node[above] {$x_2$};

\foreach \x in {-1,...,5} \draw (\x,0.05) -- (\x,-0.05) node[below] {\tiny\x};
\foreach \y in {-3,...,4} \draw (-0.05,\y) -- (0.05,\y) node[right] {\tiny\y};

\fill[blue!50!cyan,opacity=0.3] (0,0) -- (0,3/2) -- (1,1) -- (2,0) --cycle;

\draw (-1,3) -- (2,0) -- node[below,sloped] {\tiny$x_1+x_2\leq2$} (5,-3);
\draw (-1,11/3) -- node[above,sloped] {\tiny$2x_1+3x_2\leq9$} (5, -1/3);
\draw (-1,2) -- node[above,sloped] {\tiny$x_1+2x_2\leq3$} (5,-1);
\draw (-1,0) -- (4,0) -- node[above,sloped] {\tiny$x_1\geq0$} (5,0);
\draw (0,-3)  -- node[above,sloped] {\tiny$x_2\geq0$} (0,-2) -- (0,4);

\end{tikzpicture}
\caption{Solution space of the linear program (blue marked)}
\label{fig21}
\end{figure}
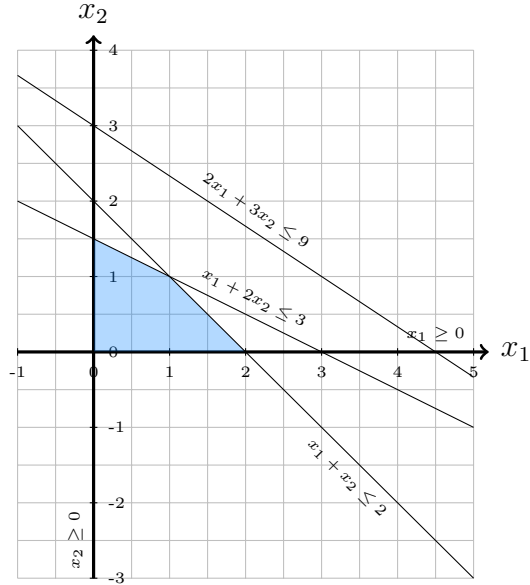
It visualizes the polyhedron that is created by intersecting the five inequalities which are drawn and the intersection of the half planes is marked blue. In other words, all vectors in the marked area are an element of all five half planes and therefore satisfy all inequalities/constraints. Thus, each of these vectors is a feasible solution to the linear program $\tilde{P}$. This marked area is called the \emph{feasible region}. In general, the latter can be either the empty set, bounded or unbounded. The first case occurs when there exists no assignment to the linear variables that satisfies all linear constraints. In other words, no feasible solution exists, the LP is then said to be \emph{infeasible}. The second case is depicted by Figure \ref{fig21} showing that there are feasible solutions in the feasible region and that it is closed, which means that the domains of the linear variables are restricted by the linear constraints. It is though sufficient that the linear variables are restricted in direction of the objective function to state that the linear program is bounded. This implies that there is a vector in this feasible region which maximizes the value of the objective function. This vector is called the \emph{optimal solution} for this linear program. In the last case, the linear program is said to be unbounded if the domain of some linear variable is not restricted in the direction of the objective function. For example if constraints (3)-(5) in $\tilde{P}$ do not exist, there does not exist an optimal solution because the value of the objective function is unbounded in the direction of the objective. Thus, the LP is unbounded.
After seeing that the feasible region for $\tilde{P}$ is closed, the task is to find a certain vector $(x_1,x_2)$ in the feasible region that maximizes the function $4x_1 + 2x_2$. 

However, this is not the only way to calculate the maximum value of the objective function. There exists the \emph{dual} linear program $\tilde{D}$ to the linear program $\tilde{P}$ which describes a minimization problem and its optimal solution results in the same value as the one of $\tilde{P}$, if an optimal solution for $\tilde{P}$ exists. Accordingly, the linear program $\tilde{P}$ is called the \emph{primal} linear program. 

\begin{table}
	\centering
	\begin{tabular}{|c | c  c|c c|}
		\hline			
		&   Primal linear program 		&& Dual linear program 			&\\ \hline
		Variables 			&	 $x_1,x_2,...,x_n$			&& $y_1,y_2,...,y_m$			&	\\
		Matrix	  			& $A$	 				  		&&	$A^T$ 						&\\
		Right-hand side		& \textbf{b} 					&& \textbf{c} 					&\\
		Objective Function	& max $\textbf{c}^T\textbf{x}$	&& min $\textbf{b}^T\textbf{y}$& 	\\
		Constraints			& $i$th constraint has &$\leq$	& $y_i \geq 0$			&		\\
		&	 				   &$\geq$  & 	 $y_i \leq 0$		&		\\
		&					   & = 		& $y_i \in \mathbb{R}$	&		\\
		& $x_j \geq 0$					&&$j$th constraints has &$\geq$ \\
		& $x_j \leq 0$					&&						&$\leq$\\
		& $x_j \in \mathbb{R}$			&&						& = 	\\
		\hline  
	\end{tabular}
	\caption{Dualization Recipe, from \cite{linearprogramming}}
	\label{table1}
\end{table}
In the following, the dualization recipe is shown in Table \ref{table1} taken from \cite{linearprogramming}. This recipe summarizes all the transformations that are needed to construct the dual linear program to a certain primal linear program. If the reader is not familiar with the several steps of this transformation, the reader is recommended to read Chapter 6 of \cite{linearprogramming}. 

After discussing a certain linear program $\tilde{P}$, linear programs are now reviewed in a more general form according to the dualzation recipe. The primal linear program is now expressed by
\begin{alignat*}{7}
	\text{maximize } \textbf{c}^T\textbf{x}\text{ considering } A\textbf{x} \leq \textbf{b} \text{ and } \textbf{x} \geq 0. \quad\quad\quad\quad (\text{P})
\end{alignat*}
The bold variables stand for vectors. \textbf{x} is the vector of the $x_1,x_2,...,x_n$ continuous variables of the primal linear program and $\textbf{c}^T$ is the vector containing the constants by which the $x_i$ variables are multiplied in the objective function. $A$ represents the left-hand side of the constraint matrix, each row pertaining to a constraint. The values in the $m$ rows and $n$ columns are the constants by which the $x_i$ are multiplied in the constraints. \textbf{b} is the vector with the RHS of the $m$ constraints, restricting the domain of the linear variables and therefore limiting the maximum value of the objective function. As already mentioned, there exists a dual linear program $\tilde{D}$ to $\tilde{P}$ describing a minimization problem which yields the same value of the objective function, if an optimal solution exists. 

On the one side, in the maximization LP the linear variables on the LHS are increased with the aim to maximize the objective function, being restricted by the linear constraints. Accordingly, some linear constraints may become tight during the execution when increasing the linear variables. When a linear constraints becomes tight, the variables occurring in this constraints cannot be further increased. So, when no further increasing of any variable is possible anymore, the optimal solution is found. In this way, the optimal solution is approached from below. On the other side, in the minimization problem the RHS of the constraints is considered. Instead of increasing the linear variables, the vector $\mathbf{b}$ shall be minimized in order to restrict the domain of the linear variables as much as possible. In this way, the optimal solution is approached from above. So, minimizing $\mathbf{b}$ is the objective of the dual linear program. Accordingly, the dual linear program looks like:
\begin{alignat*}{7}
\text{minimize } \textbf{b}^T\textbf{y}\text{ considering } A^T\textbf{y} \geq \textbf{c} \text{ and } \textbf{y} \geq 0.\quad\quad\quad\quad (\text{D})
\end{alignat*}
Besides, the minimization of $\mathbf{b}$ is restricted by $\mathbf{c}$. In other words, the LHS of the dual minimization linear program's constraints have to be greater than $\mathbf{c}$ in order to keep the inequalities valid. The mapping of the operators in Table \ref{table1} shows how to dualize the linear program in both directions and that the minimization/maximization problem can be expressed as the primal/dual linear program.

After transforming $\tilde{P}$ to $\tilde{D}$, the result is that the value of the objective function of $\tilde{D}$, corresponding to the found dual solution, is an upper bound for maximum value of the objective function of $\tilde{P}$. This bound can be generalized over all feasible primal and dual solutions as the \emph{weak duality theorem} \cite{linearprogramming}. Simply said, for each dual feasible solution \textbf{y} of D, the value $\textbf{b}^T\textbf{y}$ provides an upper bound on the maximum of the objective function for any feasible solution of the primal linear program P. Formally:
\begin{alignat}{7}
	\textbf{c}^T\textbf{x} \leq \textbf{b}^T\textbf{y}, \quad\textbf{x},\textbf{y} \geq \vec{0}, A\textbf{x}\leq \textbf{b}, A^T \textbf{y} \geq c.
\end{alignat}
It is called the \emph{weak} duality theorem because it only provides the upper bound of the minimization to the maximization problem. Its proof can be read in \cite{survey}.  The difference between the values of the primal and dual solution is called the \emph{duality gap}. It is always greater or equal to zero due to the weak duality theorem. The closer it gets to zero, the closer the primal and dual feasible solutions are to the optimal solution. Other important consequences of this theorem are that if D is bounded, P cannot be unbounded because D has an optimal solution which is then an upper bound for P. Furthermore, if P is bounded, D cannot be unbounded because P then has an optimal solution which lower bounds D. On top of that, P and D cannot be simultaneously unbounded because one of them always bounds the other with a feasible solution. 
Concluding, the weak duality theorem rules out three out of nine cases between P and D being bounded, unbounded or infeasible.
The (strong) duality theorem \cite{linearprogramming} provides statements to the other six cases.
 
To sum up, the duality recipe is an algorithmic way to transform a primal to a dual linear program and vise versa.  Note that there is no rule that the minimization linear program is always called the dual and the maximization being called the primal, both can be the primal/dual linear program. 

\subsection{Duality and Online Problems}
When analyzing an algorithm, it is the conventional way to ask how high the complexity is in terms of running time. Upon determining the worst case complexity, one can predict how much time is needed to execute the algorithm given the technical details of the computer and the input size for the algorithm. 

However, when dealing with algorithms whose input is not completely known before executing it, the focus of interests changes. In an \emph{online} problem the input is partly unknown to the algorithm and provided in parts during its execution. Upon the arrival of a new input, the \emph{online} algorithm has to process this input in real-time. Furthermore, in some scenarios (as in this thesis) the online algorithm is not allowed to undo previous decisions. Hence, the challenge is to accomplish as much ``profit'' as possible under the condition of uncertainty about the input. In other words, the future is unknown to the online algorithm and the algorithm still has to achieve good results. In contrast to this, the \emph{offline} algorithm knows the sequence of arriving inputs in advance and can therefore select the subset of the best fitting inputs for optimizing the result.

The \emph{competitive ratio} can be used to measure the performance of an online algorithm. More precisely, it compares the solution of the online algorithm to the solution of the corresponding offline algorithm which optimally solves the problem. Thus, one achieves a relative comparison between an online solution and the best offline solution. Therefore, the competitive ratio works as a metric for measuring how well an online algorithm relatively performs to the optimal offline algorithm. So, the competitive ratio is of most interest besides the run time complexity when dealing with online problems.

More formally, the definition of the competitive ratio is as follows. Consider a maximization optimization problem $Q$ and a certain instance $I \in \mathbb{I}(Q)$ of this problem. Instance $I$ includes a sequence of requests $\sigma$ that arrive in an online fashion as described above, such that the online algorithm has to process these requests in real-time. Furthermore, the outcome of such an algorithm solving this problem is interpreted as a \emph{profit}. Therefore, there is a set of feasible solutions $S(I)$ for $I \in \mathbb{I}(Q)$ and each $s \in S(I)$ results in a certain profit value $\mathtt{profit}(s)$. Let $s_{\mathrm{OPT}} \in S(I)$ be the solution with the maximum profit upon all solutions in $S(I)$. In other words, $s_{\mathrm{OPT}}$ is the optimal solution for $I$ which is computed by an optimal offline algorithm knowing $\sigma$ in advance. An online algorithm is called \emph{c-competitive} if it achieves, under any instance $I \in \mathbb{I}(Q)$ , at least a profit of $\frac{\mathtt{profit}(s_{\mathrm{OPT}})}{c}$ of the maximum profit solution $s_{\mathrm{OPT}}$. The definition of competitiveness of minimization optimization problems is analogous. Accordingly, the results of the algorithm solving the minimization problem are interpreted as a \emph{cost}. Therefore, a $c$-competitive online algorithm will produce a solution of at most $c \cdot \mathtt{cost}(s_{\mathrm{OPT}})$ for any instance $I \in \mathbb{I}(Q)$ , $\mathtt{cost}(s_{\mathrm{OPT}})$ being now the optimal solution minimizing the cost.

Having such a guarantee on the performance can be significant in practical applications, especially when dealing with a certain level of uncertainty in online problems. In order to obtain a competitive online algorithm one can apply the primal-dual approach for online problems, developed in \cite{survey}. The primal-dual approach is an algorithmic technique to design approximation algorithms for NP-hard problems and is introduced in \cite{survey} to develop approximation algorithms for online problems (competitive online algorithms). In the following, an abstract explanation on how to construct such an algorithm with the primal-dual approach is explained followed by a sample algorithm in the next section.

To begin with, the precondition is that the online problem can be transformed to a linear optimization problem because it is essential that the offline version of this problem is formulated as a LP and the survey \cite{survey} only considers the duality theorem for linear programs. Afterwards, the LP is dualized. Then, an online approximation algorithm is constructed based on the primal and dual LP of the offline problem. More precisely, the algorithm aims to simultaneously approximate the maximum profit and the minimum cost with respect to the linear constraints. However, some competitive online algorithms may only approximate feasible solutions, meaning that the constraints are violated up to a certain factor. 

As the algorithm runs in an online scenario, it has to \emph{decide} whether the current input (request) is accepted or rejected because not every request provides a high profit together with a low cost. This decision is influenced by the current state of the constraints of the corresponding primal and dual LPs. Accordingly, it may happen that the online algorithm accepts a request that the optimal offline algorithm would not have selected. Thus, if the online algorithm accepts a request which is not part of the optimal offline solution, the cost and profit of the online algorithm may diverge from each other, so that the duality gap is not zero anymore.

As mentioned before, it may be the case that the competitive online algorithm does not only approximate the value of the objective function but also approximates the feasibility, meaning it violates the linear constraints. In \cite{survey}, Buchbinder et al. show for certain problems that when trying to optimize the competitive ratio on the objective function, there is always a violation by a logarithmic factor of the constraints of the maximization LP. On the opposite, they provide algorithms that achieve primal and dual feasible solutions but a logarithmic competitive ratio on the objective function. Hence, they introduce $(c_1,c_2)$-competitiveness where $c_1$ is the competitive ratio on the objective function and $c_2$ is the competitive ratio on the feasibility. More precisely, a solution may violate the constraints by a factor of at most $c_2$.

The competitive ratio is determined with the help of the weak duality theorem. The latter is only applicable if the primal and dual solutions are feasible. Assuming that the primal and dual solutions are feasible, let $\mathtt{cost}(s_\mathrm{D})$ be the cost of a dual feasible solution and $\mathtt{profit}(s_\mathrm{P})$ be the profit of a primal feasible solution. The weak duality theorem now states that the value of the dual minimization solution $\mathtt{cost}(s_\mathrm{D})$ upper bounds the value of the primal maximization solution $\mathtt{profit}(s_\mathrm{P})$. In fact, both solutions bound each other by that inequality, so that the optimal solution lies in the duality gap. Therefore, both solutions are at most by factor $c_1 = \frac{\mathtt{cost}(s_\mathrm{D})}{\mathtt{profit}(s_\mathrm{P})}$ away from the optimal solution. As a result, the online algorithm is $(c_1, 1)$-competitive.

If the primal solution is not feasible, another approach can be applied. The competitive ratio describes the factor by which the profit of the optimal offline solution $\mathtt{profit}(s_{\mathrm{OPT}})$ is greater than the one of the online solution $\mathtt{profit}(s_{\mathrm{P}})$ in the worst case. The optimal offline solution is unknown but it is known that it is feasible. Thus, the value $\mathtt{cost}(s_\mathrm{D})$ of the dual feasible solution is an upper bound for $\mathtt{profit}(s_{\mathrm{OPT}})$ due to weak duality. Formally:
\begin{align}\label{eq50}
	\frac{\mathtt{profit}(s_{\mathrm{OPT}})}{\mathtt{profit}(s_{\mathrm{P}})} \quad\quad \myeq \quad	\quad	\frac{\mathtt{cost}(s_{\mathrm{D}})}{\mathtt{profit}(s_{\mathrm{P}})}
\end{align}
Thus, by Equation \ref{eq50}, the competitive ratio is upper bounded by the quotient $\frac{\mathtt{cost}(s_{\mathrm{D}})}{\mathtt{profit}(s_{\mathrm{P}})}$.
 
In the survey \cite{survey}, the authors have designed a framework for designing competitive online algorithms using the primal-dual approach. Based on this framework, they present several generic primal-dual algorithms, which are based on the primal and its dual linear program, for solving the online problem. In addition, they derive for each of these generic algorithms how to systematically proof their competitive ratios with the help of the weak duality theorem as shown previously. This opens a wide field of opportunities for designing competitive online algorithms. They showed the application of the primal-dual approach for many online problems as the weighted caching problem, the set-cover problem, several graph optimization problems, routing and load balancing.

The idea is that one of these generic algorithms should be applicable to the specific online problem of the user, so that she can design her specific competitive online algorithm based on the primal-dual framework. One of these generic algorithms is also applied in this thesis and therefore a similar online problem, which is discussed in the survey \cite{survey}, is presented in the following as an introduction into the primal-dual approach for online problems. Note that in the following, the minimization linear program is referred to as the primal and the maximization linear program as the dual linear program. 

\subsection{The Online Routing Problem}
In this section, an example of the design of an competitive online algorithm via the primal-dual approach is given. This example is given in \cite{survey}. It deals with an online routing problem. The setting consists of the substrate network $G = (V,E)$ with edge capacities $u(e)$, for $e \in E$ and a sequence of requests $R$ with requests $r_i = (s_i,t_i) \in R$ that arrive in an online fashion. When processing $r_i$, the algorithm has to find a path from $s_i$ to $t_i$ with an available bandwidth of 1. If there exists no path that can serve $r_i$, $r_i$ is rejected. Furthermore, as soon as a request is accepted and bandwidth has been allocated for this request, this decision cannot be reversed. In addition, an accepted request never leaves the substrate network, so that completed allocations are maintained forever. 

The corresponding optimization problem asks to maximize the throughput in the substrate network by accepting as many requests as possible which is restricted by the capacities of the edges. The throughput and the utilization of the edges' capacities increase linear with the addition of a request. Furthermore, the bandwidth being allocated for a request is either 0 or 1, hence it is a discrete integer value. So, this optimization problem can be formulated as an ILP, as shown in Figure \ref{fig57}. 
\begin{figure}[ht]
	\begin{alignat}{7}
	&\text{maximize the value} \sum_{r_i \in R} \sum_{p \in \mathcal{P}} f(i,p) \nonumber\\
	&\text{subject to}\nonumber\\
	&\forall r_i\in R: \quad\sum_{p \in \mathcal{P}(r_i)} \quad\quad\quad f(r_i, p) \leq 1 \\
	&\forall e \in E: \sum_{r_i \in R, p \in \mathcal{P}(r_i) | e \in p} f(r_i,p) \leq u(e)\\
	&\forall r_i \in R,p \in \mathcal{P}(r_i): \quad \quad f(r_i,p) \in \{0,1\}\nonumber	
 	\end{alignat}
	\caption{ILP of the described routing problem}
	\label{fig57}
\end{figure}
There exists a discrete variable $f(i,p)$ for every request $r_i \in R$ and possible path $p \in \mathcal{P}$ that can serve $r_i$ with a bandwidth of 1 from $s_i$ to $t_i$. In the objective function, the $f(r_i,p)$ are summed up over all requests and their corresponding possible paths. In other words, the objective is to maximize the amount of requests being served. This maximization is restricted by the constraints in (8) and (9). In (8), the sum over all possible paths $p \in \mathcal{P}$ for a certain $r_i$ shall be less or equal to 1. Since $f(r_i,p)$ is defined to be either 0 or 1, only one $f(r_i,p)$ can be set 1 for a certain $r_i$. Thus, per request $r_i$, only one path $p \in \mathcal{P}(r_i)$ serves a bandwidth of 1 from $s_i$ to $t_i$. The second set of constraints in (9) sums up the $f(r_i,p)$ over all requests $r_i$ and paths $p \in \mathcal{P}(r_i)$ for which a certain edge $e$ occurs in $p$. So, this sum represents the sum of allocations done on $e$ which shall be smaller than its capacity. Thus, this set of constraints prevents the edges from being oversubscribed. 

This ILP is now relaxed in order to apply the primal-dual approach on its corresponding LP. To this end, all discrete variables (in this case ($f(r_i,p)$) are changed to continuous variables. So, the only change in the ILP is the domain of $f(r_i,p)$. Next, this LP is dualized. The primal and dual LP are shown in Table \ref{table3}. Consider at first the dual LP. The domain of $f(i,p)$ changed to be between 0 and 1. Due to the first set of constraints no $f(i,p)$ can be greater than 1 because the sum of all $f(i,p)$ for a certain $r_i$ is at most 1. Thus, it is sufficient to define the domain of $f(i,p)$ to be $f(i,p) \geq 0$. Since $f(i,p)$ is between 0 and 1, multiple paths $p \in \mathcal{P}(r_i)$ can serve $r_i$ by distributing the bandwidth over several paths. In addition, a request $r_i$ does not need to be fully served with a bandwidth of 1 because the first set of constraints only states that the sum of the $f(i,p)$ has to be smaller than 1. In other words, \emph{fractional} solutions are allowed.  
\begin{table}[htpb]
	\resizebox{\textwidth}{!}{%
	\begin{tabular}{l l | l l}
		&   Primal (P) 		& \multicolumn{1}{|c}{Dual (D)} &			\\ \hline
		Minimize: 			&	$\sum_{e \in E}u(e)x(e) + \sum_{r_i}z(r_i)$		& Maximize: $\sum_{r_i}$$\sum_{p \in P(r_i)}f(r_i,p)$	&\\
		subject to 			&	 				  		    & subject to	&					\\
		$\forall r_i\in R, p \in P(r_i)$:& $\sum_{e \in p} x(e) + z(r_i) \geq 1$& $\forall r_i\in R: \sum_{p \in P(r_i)} f(r_i, p)$&$\leq 1$  \\
		&&	$\forall e \in E: \sum_{r_i \in R, p \in P(r_i) | e \in p} f(r_i,p) $&$\leq u(e)$				\\
	&	$\forall r_i,z(r_i) \geq 0, \forall e, x(e) \geq 0$ & $\forall r_i,p: f(r_i,p) \geq 0$ &\\
		
	\end{tabular}}
	\caption{Primal and dual problems}
	\label{table3}
\end{table}

On the other side, in the primal LP there exists a set of primal variables for each set of constraints in the dual LP. Thus, there are two sets of variables $x(e)$ and $z(r_i)$ corresponding to the first and second set of dual constraints. These primal variables are multiplied to the vector on the RHS of the dual constraints which is recognizable in the  primal objective function. If a request $r_i$ is served, the corresponding $z(r_i)$ is increased and for all edges $e$ which are involved in the path $p$, the corresponding $x(e)$ variables are increased too. The cost of the primal objective function then increases too. Thus, one can interpret $x(e)$ as the cost for utilizing $e$'s capacity.
By asking to minimize the costs for the provider, the resource footprint is reduced to a minimum per request.

The problem that is induced with this LP formulation is that there exists an exponential amount of possible paths to serve the request because a trivial online algorithm would try all possible paths in order to find the one minimizing the costs. This leads to an exponential amount of dual variables $f(r_i,p)$ that are created for every request which is inefficient in memory consumption and running time. As shown in the following Algorithm \ref{alg3}, this problem is easily mitigated. Let $n = |E|$ be the amount of edges in the graph.
\RestyleAlgo{boxruled}
\LinesNumbered
\begin{algorithm}[ht]
	\caption{Competitive Online Routing Algorithm\label{alg3} \cite{survey}}
	Initially $x(e) \leftarrow 0$.\\
	When a new request $r_i = (s_i, t_i, P(r_i))$ arrives:\\
	\If{there exists a path $p \in P(r_i)$ of length $<$ 1 with respect to $x(e)$}
	{Route the request on ``any'' path $p \in P(r_i)$ with length $<$ 1.\\
		$f(r_i,p) \leftarrow$1.\\
		$z(r_i) \leftarrow 1$.\\
		\For{$e \in p$}
		{ $x(e) \leftarrow x(e) \mathrm{exp}\bigg(\frac{\mathrm{ln}(1+n)}{u(e)}\bigg) + \frac{1}{n} \bigg[\mathrm{exp}\bigg(\frac{\mathrm{ln}(1+n)}{u(e)}\bigg)-1\bigg]$.
		} 		
	}
\end{algorithm}
By only asking for the existence of a path $p$, it is enough to calculate one path. This can be for example the shortest path between $s_i$ and $t_i$. If the shortest path does not have a length smaller than 1, then no other path will have a length smaller than 1, so that $r_i$ will be rejected. 
Line 4 in Algorithm \ref{alg3} says that ``any'' path $p \in P(r_i)$ with length smaller than 1 can be taken because the competitive ratio does not change if the minimum-cost path or any other path of length smaller than 1 is taken. This keeps the path selection generalized and does not restrict the algorithm to use a shortest path algorithm.
If a path exists, a bandwidth of value 1 is allocated on each edge $e \in p$ for $r_i$. Since $\sum_{e \in p} x(e) < 1$, the primal constraint, that corresponds to $r_i$, is not feasible yet but has to become feasible in order to not violate the competitiveness of the algorithm. Therefore, $z(r_i)$ is set to 1 which will set the corresponding primal constraint of $r_i$ satisfied since the inequality demands the LHS to be greater than 1. Next, the $x(e)$ values of all involved edges are increased according to the shown function in line 8 of Algorithm \ref{alg3}. This function is of most interest in this algorithm because it directly influences its competitive ratio. To become more familiar with the structure of this function, consider the following Theorem \ref{theorem}. The corresponding proof in \cite{survey} is outlined.
\begin{theorem}\label{theorem}
	The algorithm is $O(u_{\mathrm{min}}[\mathrm{exp}(\mathrm{ln}(1+n)/u_{\mathrm{min}})-1])$-competitive and produces a primal and dual feasible solution \cite{survey}.
\end{theorem}
\begin{proof}
	At first, observe that the function in line 8 of Algorithm \ref{alg3} is monotonically increasing. Consequently, the $x(e)$ value is never decreased on further request iterations in order to hold a primal feasible solution. Accordingly to the previous section, the competitive ratio is computed by looking at the quotient between the primal and dual solution in the worst case since both are feasible. It is equivalent to compute the quotient between the changes of the values of the objective functions in the worst case because both objective function are increased by their worst case changes for every accepted request and therefore add up to their final value. When a request is accepted, the change in the dual profit is exactly 1 because the additional $f(r_i, p)$ variable is added to the objective function. The latter describes the reserved bandwidths for $r_i$ which is exactly 1 on each edge of the path $p$. The change in the primal cost is at most
	\allowdisplaybreaks
	\begin{alignat}{7}
		&\sum_{e \in p} u(e)(x(e)-x_{\mathrm{prev}}(e)) + z(r_i)\\
		&\resizebox{.9 \textwidth}{!} 
		{$
		= 1+ \sum_{e \in p} u(e)\bigg(x_{\mathrm{prev}}(e) \mathrm{exp}\bigg(\frac{\mathrm{ln}(1+n)}{u(e)}\bigg) + \frac{1}{n} \bigg[\mathrm{exp}\bigg(\frac{\mathrm{ln}(1+n)}{u(e)}\bigg)-1\bigg]-x_{\mathrm{prev}}(e)\bigg)
		$}\\
		&\resizebox{.9 \textwidth}{!} 
		{$
		=1+ \sum_{e \in p} u(e)\bigg(x_{\mathrm{prev}}(e)\bigg[ \mathrm{exp}\bigg(\frac{\mathrm{ln}(1+n)}{u(e)}\bigg)-1\bigg] + \frac{1}{n} \bigg[\mathrm{exp}\bigg(\frac{\mathrm{ln}(1+n)}{u(e)}\bigg)-1\bigg]\bigg)
		$} \\ 
		&\leq 1+ u(e)\bigg[\mathrm{exp}\bigg(\frac{\mathrm{ln}(1+n)}{u(e)}\bigg)-1\bigg] + u(e)\bigg[\mathrm{exp}\bigg(\frac{\mathrm{ln}(1+n)}{u(e)}\bigg)-1\bigg] \\ 
		&= 1+  2\cdot u(e)\bigg[\mathrm{exp}\bigg(\frac{\mathrm{ln}(1+n)}{u(e)}\bigg)-1\bigg]\\
		&\leq 1+  2\cdot u_{\mathrm{min}} \bigg[\mathrm{exp}\bigg(\frac{\mathrm{ln}(1+n)}{u_{\mathrm{min}}}\bigg)-1\bigg]
	\end{alignat}
	The change of the primal variable $z(r_i)$ is at most its value since its set from 0 to 1 upon the acceptance of $r_i$. The change of the $x(e)$ variables that are involved in $p$ is described by the difference between the new and old value $x(e)-x_{\mathrm{prev}}(e)$. This difference is shown in (10) where the primal objective function is given with the difference between the new and old $x(e)$ value. In (11) $x(e)$ is substituted by its definition. In (12), $x_{\mathrm{prev}}(e)$ is excluded. Inequality (13) upper bounds the change by using that $\sum_{e \in p} x(e) \leq 1$ and $\sum_{e \in p} \frac{1}{n} \leq 1$. Inequality (15) uses the fact that the function $u(e)\bigg[\mathrm{exp}\bigg(\frac{\mathrm{ln}(1+n)}{u(e)}\bigg)-1\bigg]$ is monotonically decreasing with respect to $u(e)$. Thus, the change in the primal cost is upper bounded by $1+  2\cdot u_{\mathrm{min}} \bigg[\mathrm{exp}\bigg(\frac{\mathrm{ln}(1+n)}{u_{\mathrm{min}}}\bigg)-1\bigg]$. Now, the weak duality theorem is applied stating that the value of any feasible solution of the minimization problem upper bounds the value of any feasible dual solution of the maximization problem. Since the theorem states that both solutions are feasible, consider the weak duality theorem for the maximum changes in the primal cost and dual profit:
	\begin{alignat}{7}
		1 \leq 1+ & 2\cdot u_{\mathrm{min}} \bigg[\mathrm{exp}\bigg(\frac{\mathrm{ln}(1+n)}{u_{\mathrm{min}}}\bigg)-1\bigg].
	\end{alignat}
	According to the weak duality theorem, the value of the dual solution is in every iteration at most by the factor $1+  2\cdot u_{\mathrm{min}} \bigg[\mathrm{exp}\bigg(\frac{\mathrm{ln}(1+n)}{u_{\mathrm{min}}}\bigg)-1\bigg]$ smaller from the maximum possible change in the profit. In other words, the algorithm is \\$1+  2\cdot u_{\mathrm{min}} \bigg[\mathrm{exp}\bigg(\frac{\mathrm{ln}(1+n)}{u_{\mathrm{min}}}\bigg)-1\bigg]$-competitive.
	
	Next, the theorem states that the capacity constraints are not violated. 
	Observe that the $x(e)$-function in line 8 of Algorithm \ref{alg3} is a geometric sequence with initial value $1/n[\mathrm{exp}(\mathrm{ln}(1+n)/u(e)-1]$ and the multiplier $q= \mathrm{exp}(\mathrm{ln}(1+n)/u(e))$. Since each request that involves edge $e$ in its path allocates a bandwidth of value 1 on $e$, maximally $u(e)$ allocations can be done on $e$. So, after $u(e)$ allocations on $e$, the value of $x(e)$ is 
	\begin{alignat}{7}
		 x(e) &= \frac{1}{n}\bigg( \mathrm{exp}\bigg(\frac{\mathrm{ln}(1+n)}{u(e)}\bigg)  -1\bigg) 
		 \frac{\mathrm{exp}\bigg(\frac{u(e)\mathrm{ln}(1+n)}{u(e)}\bigg)-1}{\mathrm{exp}\bigg(\frac{\mathrm{ln}(1+n)}{u(e)}\bigg)-1}\\
		 &= \frac{1}{n} (1+n-1) = 1.
	\end{alignat} 
	Note that the algorithm only accepts a request if its length is smaller than 1. So, if $x(e) = 1$ the length is already 1. Thus, $e$ with $x(e) = 1$ will never occur in any further path, such that it will never be oversubscribed.
\end{proof}
One can see that the cost-function $x(e)$ determines the competitive ratio on the objective function and on the constraints. Thus, it is of most interest to develop such a function according to the desired competitive ratios. 

The competitive ratio on the objective function is a logarithmic factor, but Algorithm \ref{alg3} is therefore 1-competitive on the constraints. The survey shows throughout several examples that there is the previously mentioned trade-off between the competitiveness on the objective function and the capacity constraints for the online routing problem, as they developed a (1, log $n$)-competitive online algorithm. More precisely, this algorithm is 1-competitive on the objective function and log $n$ competitive on the capacity constraints. In the contribution of this thesis, one competitive online algorithm is developed and another is designed, based on \cite{gvop}, for the virtual cluster embedding problem where one of them focuses on the competitiveness of the objective function and the other on the competitiveness of the capacity constraints.

After introducing the linear programming duality and its application, in the following section the related work is presented.
\section{Related Work}\label{relatedwork}
There has been a lot of work done in the past years based on the virtual cluster model \cite{vc}. An example for an extension of the VC model is presented in \cite{survivability} in which a survavibility parameter is added to the VCE problem. A system that recovers a virtual cluster embedding in case of single-PM failure was investigated by deploying so called ``primary'' and ``backup'' embeddings.

Another researched parameter was the running time of VCE embedding algorithms. As in the previously mentioned paper \cite{survivability}, they used a heuristic to compute VCE embeddings since solving the VCE problem was thought to be NP-hard. So, most papers related to VCE embedding problems suggested heuristic algorithms to find VCE solutions. However, Rost et al. developed in \cite{vcace} an algorithm that finds an optimal solution to the VCE problem in polynomial time, showing that the problem is not NP-hard. 

The results in \cite{vcace} have been used in \cite{kraken} to develop an algorithm that allows an elastic resource reservation. This means that the tenant is allowed to update its demands during the running time of its request. 

However, all the mentioned algorithms	 have in common that they are online algorithms. The execution of these non-competitive online algorithms can lead to a poor profit for the provider. Thus, the objective of this thesis is develop a competitive online VCE algorithm in order to guarantee a fraction of the maximum possible profit. 

In \cite{allocator}, the authors present competitive online algorithms, based on \cite{survey}, for the joint allocation of different resources (e.g., CPU, memory, disk, etc.) in a geo-distributed cloud. Particularly, they do not apply the VC model but propose a generalized resource placement methodology that considers the parameters: resource type, VM class, location of the resource, capacity of a resource dependent on the resource type, the location and timestamp. 
In contrast to the VC model, they do not model guaranteed bandwidths. Nevertheless, weighted edges can be incorporated by replacing the weighted edge with a weighted node and two incident edges, so that weighted edges are modeled as weighted nodes. Thus, edges with a bandwidth capacity can be represented by a resource type. 
However, in one of their proofs they assume that the minimum consumption of any resource type is at least 1 per request, which prohibits virtual cluster requests from being embedded onto a single PM since at least one unit of bandwidth has to be used according to their definitions.

In \cite{gvop}, the authors present a framework for general network embedding problems. Their framework is based on the primal-dual approach for online problems \cite{survey}. 
It is applicable to the VCE problem because it also considers a substrate graph with edge capacities and a sequence of requests arriving in an online fashion. A request has to be either rejected or accepted upon its arrival dependent on whether an acceptable embedding exists. They developed a generic competitive online algorithm whose competitive ratio is constant on the objective function but logarithmic on the capacity constraints. Based on this algorithm, the GVOP algorithm of this thesis is designed. In \cite{gvop}, the authors mention in Corollary 2 that the algorithm's competitive ratio can be transformed into a constant competitiveness on the capacity constraints by downscaling the capacities. However, with this transformation a restriction is induced, namely that no request is allowed to allocate more bandwidth on an edge than the minimum edge capacity in the substrate network divided by the logarithmic term of the competitive ratio. This is a hard restriction when considering that the minimum capacity can be 1. In contrast to this, the COVCE algorithm of this thesis does not have such a restriction.

\section{Theory}\label{theory}
This section consists of a short introduction into the online VCE problem by revisiting the example of the introduction Section \ref{intro} and then we present two competitive online algorithms, where each of them focuses on a different competitive ratio of the $(c_1,c_2)$-competitiveness. The first one is the \textbf{G}eneral all-or-nothing \textbf{V}Net Packing \textbf{O}nline \textbf{P}acking (GVOP) algorithm based on \cite{gvop} and the second one is the \textbf{C}ompetitive \textbf{O}nline \textbf{V}irtual \textbf{C}luster \textbf{E}mbedding (COVCE) algorithm which is a novel contribution of this thesis. 

To begin with, in the last paragraph of Section \ref{model}, in which the VC model and the VCE problem are explained, an approach for an optimization problem has been formulated in which the costs of serving several requests shall be minimized. This approach is the basis for designing the VC-ACE algorithm \cite{vcace} which finds a minimal-cost embedding in polynomial time solving an integral minimum-cost single-commodity flow problem. It minimizes the resource consumption per request considering the initial capacities for every request it processes.

Consider the VC-ACE algorithm in which residual capacities are used for embedding a request and therefore respecting the \emph{load} on the substrate nodes and edges. The load on a node or an edge is the sum over all allocations on this node/edge done by the previous requests. Thus, this version is called the \emph{Greedy} VC-ACE algorithm and will embed \emph{any} request with minimal resource consumption that has a valid embedding. However, by applying this procedure, it is not ensured that the profit for the provider is maximized. Consider the Figures \ref{fig34} - \ref{fig36} (page 2-4) again in which multiple arriving requests are embedded with the least amount of needed resources. After filling the substrate, there is request 4 with a higher benefit than any other previous request. Since the network is fully utilized, request 4 has to be rejected. 

More formally, these requests arrive over time and the algorithm only knows the current state of the network at the arrival of a request. So, the algorithm has to make a decision upon its current situation. Hence, this variant of the virtual cluster embedding problem is an online problem.

By this observation, one can follow that accepting any request, under the assumption that there exists a valid embedding for this request, does not ensure that the profit for the provider is maximized. Greedily minimizing the resource consumption per request does not guarantee a \emph{globally} maximization of the profit. To achieve a guarantee on the profit, one has at first to consider the maximum profit which is possible for a certain sequence of requests and a certain substrate network. The maximum profit $\mathtt{profit}_{\mathrm{MAX,ILP}}$ which is achievable is computed by the optimal offline algorithm. 
By solving the ILP of that problem, assuming the profit increases linear and the capacity constraints are linear too, the optimal offline solution with its profit $\mathtt{profit}_{\mathrm{MAX,ILP}}$ is calculated.
The latter is also a solution of the corresponding LP that is the relaxed version of this ILP. Furthermore, the profit $\mathtt{profit}_{\mathrm{MAX,LP}}$ of this LP is greater or equals than $\mathtt{profit}_{\mathrm{MAX,ILP}}$. 
Now, the primal-dual approach for online problems of \cite{survey} can be applied in order to upper bound $\mathtt{profit}_{\mathrm{MAX,LP}}$ with the cost of the minimization LP due to weak duality. In this way, a competitive online algorithm is developed that guarantees a $c$-fraction of $\mathtt{profit}_{\mathrm{MAX,LP}}$ for any request sequence. In other words, the competitiveness of the online algorithms, that are designed upon this primal-dual approach, is with respect to an optimal fractional solution. Thus, the optimal offline solution of the ILP is included is in this solution space, so that this competitive ratio is also valid for the optimal offline solution of the ILP (describing the original problem).

However, the Greedy VC-ACE algorithm, which solves the problem in an online fashion, lacks competitiveness and therefore has no worst case guarantee of an $c$-fraction of the optimal offline solution. 

To define a competitive online algorithm, the offline dual and primal linear programs of this optimization problem are shown in the next section and then the main contribution of this thesis, the COVCE algorithm will be explained. Before showing the corresponding LP, some additional variables are defined which are used for describing the primal and dual LPs. 

The \emph{load} $l(i,k,y): V_S \cup E_S\rightarrow \mathbb{N}$ on a node or edge $y$ describes how many resources of that element are consumed by embedding $m_i^k$ of request $r_i$. Note that a physical edge or node $y$ can be used multiple times in an embedding if enough resources are available. In this case, $l(i,k,y)$ describes the sum of the allocations on node/edge $y$ by embedding $m_i^k$ of request $r_i$. In a valid embedding $m_i^k$ of request $r_i$, the inequality $l(i,k,y) + \sum_{r_j \in R\setminus\{r_i\}}\sum_{m_j^q \in M(r_j)} l(j,q,y) \leq \mathtt{cap}(y)$ holds $\forall y \in E_S \cup V_S$. 

Next, a dual variable is described. $f(i,k) \in [0,1]$ is assigned with 1 if the embedding $m_i^k$ of request $r_i$ shall be fully served by the algorithm. Accordingly, $f(i,k)$ is between 0 and 1 if $r_i$ is partly served. In other words, fractional solution are allowed in the offline linear program since the competitive ratio of the online algorithm is with respect to the optimal fractional solution.  

The primal variables $x(e), e \in E_S$ and $x(v), v \in V_S$ describe the cost for using edge $e$ and node $v$ in an embedding. Formally, $x(e),x(v): V_S\cup E_S \to \mathbb{R}^+$.

\subsection{Linear Program}
In this section, the offline LP and its dual LP are discussed. At first the dual LP is shown. As already mentioned, the objective is to maximize the profit, so that the objective function of the LP is a sum of benefits over all requests $r_i$. The benefit $b(r_i)$ is multiplied with the dual variable $f(i,k)$ working as a decision variable. $f(i,k)$ has two arguments: $i$ is the argument iterating over all requests and $k$ over all valid embedding of a certain request $r_i$. So, for each request $r_i$ and each valid embedding $m_i^k \in M(r_i)$ there exists a variable $f(i,k)$.

The first constraints in Equation \ref{eq1} describe a sum of $f(i,k)$ over all valid embeddings for a certain request $r_i$. This sum shall be less or equal than 1. Since $f(i,k) \in [0,1]$, fractional solutions are allowed as described before in the background Section \ref{background}. The consequence for the online routing problem was that multiple paths were allowed to serve the request. In this case, multiple embeddings are allowed to serve the request with the restriction that the sum of all $f(i,k)$, that represents all used valid embeddings for $r_i$, has to be smaller than 1. Next, the second constraints in Equation \ref{eq2} describes the capacity constraints for each edge and each node. For each $x \in E_S \cup V_S$, there is a sum of $f(i,k) \cdot l(i,k,x)$ over all requests and all embeddings for a request $r_i$. By iterating over all valid embeddings and evaluating $f(i,k)$,  only the summands for which $f(i,k)$ is not 0 are considered, meaning the loads of the chosen embedding $m_i^k$ inducing an allocation on $x$ are summed up. So, when $f(i,k) > 0$, the load $l(i,k,x)$ will be added to the sum. Thus, this sum describes the sum of the loads of all accepted embeddings that use $x$. By forcing this sum to be smaller or equal to $x$'s capacity, resource violation is prohibited. Hence, this set of constraints describes the capacity constraints of the edges and nodes in the substrate network.
\paragraph{Dual}

\begin{alignat*}{7}
max \sum_{r_i \in R}\sum_{m^k_i \in M(r_i)} f(i, k) \cdot b(r_i)
\end{alignat*}
\begin{alignat}{7}\label{eq1}
\sum_{m^k_i \in M(r_i)} f(i, k) &\leq 1	\hspace{4em} &\forall r_i \in R
\\\label{eq2}
\sum_{r_i \in R}^{} \sum_{m^k_i \in M(r_i)} f(i, k) \cdot l(i,k,x) &\leq \mathtt{cap}(x) \hspace{4em} &\forall x \in V_S \cup E_S, 
\\\label{eq3}
f(i,k) &\geq 0 \hspace{4em} &\forall r_i \in R, m^k_i \in M(r_i)
\end{alignat}

\paragraph{Primal}

\begin{alignat*}{7}
	min \sum_{e \in E_S} \mathtt{cap}(e)\cdot x(e) + \sum_{v \in V_S} \mathtt{cap}(v)\cdot x(v) + \sum_{r_i \in R} z(r_i)
\end{alignat*}
\begin{alignat}{7}\label{eq5}
	\resizebox{.9 \textwidth}{!} 
	{$
	\sum_{e \in \mathcal{M}_{E,i}^k} l(i,k,e) \cdot x(e) + \sum_{v \in \mathcal{M}_{V,i}^k} l(i,k,v) \cdot x(v)+ z(r_i) \geq b(r_i) \forall r_i \in R, m^k_i \in M(r_i)
	$} \\\label{eq6}
	\resizebox{.46 \textwidth}{!} 
	{$
	x(e), x(v), z(r_i) \geq 0 ~\forall r_i \in R,e \in E_S, v \in V_S
	$}
\end{alignat}

Consider that the load $l(i,k,x)$ is not a dual variable but has to be computed during the execution of the algorithm. Finally, the last line of the dual LP describes the domain for the dual variable $f(i,k)$.
After describing the dual packing LP, the primal covering LP is designed with the dualization recipe in Table \ref{table1}. As described in the background Section \ref{background}, one can alternately solve the primal problem to calculate the value of the optimal solution of the dual LP.
For each constraint in the dual LP, a primal variable is created and for each dual variable a constraint is created. There are three sets of new variables, the $x(e), x(v)$ and $z(r_i)$ variables. Each set of variables is respectively multiplied with the right side of their corresponding dual constraints as shown in the objective function of the primal LP. It consists of a sum of three sums where the sums $\sum_{e \in E_S} \mathtt{cap}(e)\cdot x(e)$ and  $\sum_{v \in V_S} \mathtt{cap}(v)\cdot x(v)$ summarize the analogous dual constraints in Equation \ref{eq2} and lastly $\sum_{r_i \in R} z(r_i)$ is the analog to the set of constraints in Equation \ref{eq1}. 

There is only one set of constraints in this primal LP because there is only one set of dual variables. This set is shown in Equation \ref{eq5}. It connects the resource consumption of embedding $m_i^k$ for request $r_i$ with its benefit $b(r_i)$ and does not have a deeper interpretation. As in the duality recipe, the right side of the constraints contains the factor of the dual objective function $b(r_i)$ and the left side of the constraints includes the left side of the dual constraints.

Considering the domain of the primal variables in Equation \ref{eq6}, their functionality describes some kind of costs for embedding requests because the more requests are being embedded, the more resources are utilized, and accordingly, the corresponding primal variables are increased in order to fulfill the primal constraint. Furthermore, the products between the primal variables and the capacities increase too, reflecting the increasing usage of the resources. Thus, this increasing resource consumption can be interpreted as a cost for the provider because she has less resources left for further requests.
\subsection{Competitive Online VCE Algorithms}
After describing the primal-dual LP for the offline case, the competitive online algorithms are now presented. Both of them are based on the online setting of the previously shown offline problem. In the online problem, previously set dual variables are not allowed to be changed, referencing to the characteristic of online algorithms that previously made decisions cannot be reversed. In addition, it is assumed that an accepted request stays forever in the substrate network. Thus, primal variables are only allowed to increase during the execution of the algorithm. Furthermore, the dual variables are only allowed to be set during the iteration $i$ in which the request $r_i$ is being served because a request shall be either fully served or rejected. In other words, the algorithm shall be an all-or-nothing algorithm. Moreover, a single embedding $m_i^k \in M(r_i)$ shall be selected because the original VCE problem, which is non-relaxed, asks to find a single VCE for fully embedding $r_i$'s virtual cluster. Therefore, no fractional solutions are allowed. The last two conditions (all-or-nothing and non-fractional) induce that there is only one summand $f(i,k)$ per request $r_i$ contributing to the load per node and edge in the set of constraints in Equation \ref{eq2} and that $f(i,k) \in \{0,1\}$. 

At first, the main contribution of this thesis, the COVCE algorithm is explained and afterwards the GVOP algorithm of \cite{gvop}. Both are competitive online VCE algorithms where the COVCE algorithm has a constant competitiveness on the capacity constraints but a logarithmic competitiveness on the objective function and the GVOP algorithm vise versa.
\subsubsection{COVCE Algorithm}
In this section, the main contribution of this thesis is presented. For simplicity, it is called the \textbf{C}ompetitive \textbf{O}nline \textbf{V}irtual \textbf{C}luster \textbf{E}mbedding algorithm (COVCE). The ideas of it are based on the competitive online Algorithm 2 in chapter 4 of \cite{survey}. The algorithm is an \emph{all-or-nothing} algorithm, meaning that it either accepts or rejects a requests completely. In addition, no fractional solutions are allowed, meaning that a single embedding serves the demanded virtual cluster.

Before going through the algorithm, consider the following definitions.
Let $\mathcal{M}_{E,i}^k$ be the single set containing the image of an embedding's $m_i^k$ edge mapping and analogously  $\mathcal{M}_{V,i}^k$ be the single set containing the image of $m_i^k$'s node mapping. Moreover, let 
\begin{equation}\label{eq27}
\mathrm{cost}(m_i^k) = \sum_{e \in \mathcal{M}_{E,i}^k} l(i,k,e) \cdot x(e) + \sum_{v \in \mathcal{M}^k_{V,i}} l(i,k,v) \cdot x(v)
\end{equation} 
be the cost of the embedding $m_i^k$. In this formula, the cost gets higher with higher loads on the edges or nodes and when the primal variables $x(e)$ and $x(v)$ get higher too, which is the case when there is already a high load on the edge or node. Consequently, a high cost($m_i^k$) value reflects an embedding which consumes a lot of resources and/or a substrate network which is close to be fully utilized. Note that $x(e)$ and $x(v)$ are now the cost functions for the nodes and edges being set in the competitive algorithm and are not being set by the LP.

Let $m_i^O$ denote the optimal embedding for request $r_i$ at the time of the
arrival, where optimal means that it provides the least amount of costs due to Equation \ref{eq27}.
While the dual program only considers the capacity constraints, it is also important to respect the costs induced for the provider, where costs means how big the resource footprint of the embedding is. To this end, the minimum-cost embedding $m_i^O$ shall be determined to serve $r_i$. In addition, cost($m_i^O$) shall help with the decision whether it is worth to embed $r_i$. Thus, the primal-dual combination is important to aim a profit maximization and simultaneously a cost minimization while respecting whether the embedding costs are worth the benefit $b(r_i)$ of $r_i$. This achieves an admission control behavior, so that not any request is accepted since embedding costs might get too high. How this characteristic is included into the algorithm, is shown later.

Let $C_E$ be the maximum edge capacity and analogously $C_V$ be the maximum node capacity in the substrate network. Additionally, let $|G_S| = |V_S| + |E_S|$ be the size of the substrate graph. Lastly, $\alpha = \mathrm{max}_i\{b(r_1),...,b(r_i)\}$ be the maximum benefit that has been seen yet by the algorithm until request $r_i$. As one can see, the value of $\alpha$ increases during the running time of the algorithm.

In the following, the Algorithm \ref{alg1} is explained. 
When a new request $r_i$ arrives, the first task is to check whether the benefit $b(r_i)$ is greater than the current $\alpha$. In that case $\alpha$ is updated and all $x(e)$ and $x(v)$ too. This if-clause is further explained in details when considering the cost functions.

After that, the algorithm receives an embedding from an oracle that provides a minimal-cost embedding $m_i^k$ for $r_i$ which is assigned to $m_i^O$. For the competitive online algorithm, it is not relevant how this minimal-cost embedding is found but the decision whether to accept and embed the minimal embedding is important because this "admission control" functionality contributes to the competitiveness of the algorithm. For this reason, the problem of finding a minimal embedding is given over to an oracle procedure. For the implementation of the competitive algorithm, the VC-ACE algorithm \cite{vcace} is applied as the oracle procedure for finding a minimal-cost embedding because it optimally solves the VCE problem in polynomial time. 

The VC-ACE algorithm takes as input the virtual cluster request $VC_i = (\mathcal{N}_i,\mathcal{B}_i, \mathcal{C}_i)$, the substrate network with the original node and edge capacities, meaning it assumes that no requests are embedded, and the $x(e)$ and $x(v)$ as cost functions for the nodes and edges. It outputs a cost-minimal embedding applying the successive shortest path algorithm for solving a minimum cost-flow problem. To recapitulate, the $x(e)$ and $x(v)$ primal variables increase upon including edge $e$ and node $v$ in an accepted embedding. So, the more accepted embeddings are including $e$ and $v$, the higher their $x(e)$ and $x(v)$ increase. This behavior can be interpreted as a higher node or edge capacity consumption. If the $x(e)$ and $x(v)$ variables are applied as costs on nodes and edges for the minimum-cost flow problem, then a higher cost will provide a lesser probability of including this node or edge in an embedding again. Thus, the application of the primal variables as cost functions for the VC-ACE algorithm can protect the node or edge from being oversubscribed and creates load balancing between the nodes and edges of the substrate network. Note that all primal and dual variables are assumed to be initially zero (using lazy initialization).

\makebox[\textwidth][c]{%
\begin{minipage}{1.2\textwidth}
\begin{algorithm}[H]
	\caption{COVCE Algorithm\label{alg1}}
	Initially all $x(e)$ and $x(v)$ are set 0\;
	Upon the arrival of a new request $r_i$, the corresponding primal constraint in Equation \ref{eq5} and the dual variable $f(i,k)$ are introduced:\\
	\Indp \If{$b(r_i) > \alpha$}{ 
		$\alpha \gets b(r_i)$\; 
		$\forall e \in E_S, v \in V_S:$ re-calculate the costs $x(e)$ and $x(v)$ with respect to the new $\alpha$\;
	}
	\If{there exists an embedding $m_i^O \in M(r_i)$ with $\mathrm{cost}(m_i^O) < b(r_i)$}{\label{eq33}

		// \textit{\textbf{accept} $r_i$}\\
		$f(i,\mathrm{O}) \gets 1$\;
		$f(i,k) \gets 0$ for all other valid embeddings $m_i^k \in M(r_i)\setminus \{m_i^O\}$\;   	
		$z(r_i) \gets b(r_i)$ - cost($m_i^O)$\;
		\For{$e \in \mathcal{M}_{E,i}^O$}{
			\setlength{\abovedisplayskip}{-8pt}
			\setlength{\abovedisplayshortskip}{0pt}
			\setlength{\belowdisplayshortskip}{0pt}
			\begin{align}
			\resizebox{.75 \textwidth}{!} 
				{$x(e) \gets \frac{1}{|G_S|\cdot C_E} \cdot \left[\mathrm{\mathrm{exp}}\left(\frac{\mathrm{ln}(1+|G_S|\cdot C_E\cdot \alpha)}{\mathtt{cap}(e)} \cdot (\sum_{r_i \in R} \sum_{m^k_i \in M(r_i)}f(i,k) \cdot l(i,k,e))\right)-1\right]$ \label{eq52}}
			\end{align}
		}
		\For{$v \in \mathcal{M}_{V,i}^O$}{
			\setlength{\abovedisplayskip}{-8pt}
			\setlength{\abovedisplayshortskip}{0pt}
			\setlength{\belowdisplayshortskip}{0pt}
			\begin{align}
			\resizebox{.75 \textwidth}{!} 
				{$x(v) \gets \frac{1}{|G_S|\cdot C_V} \cdot \left[\mathrm{\mathrm{exp}}\left(\frac{\mathrm{ln}(1+|G_S|\cdot C_V\cdot \alpha)}{\mathtt{cap}(v)} \cdot (\sum_{r_i \in R} \sum_{m^k_i \in M(r_i)}f(i,k) \cdot l(i,k,v))\right)-1\right]$ \label{eq53}}
			\end{align}
		}
	}
	\Else{
		// \textit{\textbf{reject} $r_i$}\\
		$z(r_i) \gets 0$\;
		$f(i,k) \gets 0$ for all $k$\;   
		// \textit{all $x(e), e \in E_S$ and $x(v) , v \in V_S$ are unchanged}\\
		// \textit{$\rightarrow$ change in primal and dual objective functions is 0}
	}
\end{algorithm}
\end{minipage}}

A cost-minimal embedding can only be found if the request's demand of VMs, compute units and bandwidth would not violate the initial capacities. In other words, a cost-minimal embedding can only exist if there exists at least one valid embedding.

In the next line 7, the algorithm \emph{decides} in the if-condition whether to accept the request or not by a comparison of its incurred costs to its benefit. If the costs are greater than its benefit, the request is rejected. By this property of not accepting \emph{any} request, the competitiveness is created. The rejection of a request can be influenced by different factors. The request might not provide as much benefit relatively to its embedding costs as previous requests did. Additionally, the request might need to use resources which already have high costs, meaning the algorithm will reject any request whose minimal-cost embedding includes this certain node or edge. On top of that, the algorithm's limits might be already tight for all nodes and edges and therefore it will reject any request because it maintains its competitiveness.

If $r_i$ is accepted, the corresponding dual variable $f(i,k)$ is set to 1 (line 9) to indicate in the objective function of the dual LP that $r_i$ is accepted and to add $b_i$ to the result of the objective function.

In line 11 the primal variable $z(r_i)$ is set to the difference between $b(r_i)$ and cost($m_i^O)$ in order to satisfy the corresponding primal constraint with the least needed value. Furthermore, $z(r_i) \geq 0$ holds because  $b(r_i) >$ cost($m_i^O)$ holds since line 11 is only reached if cost($m_i^O) < b(r_i)$ in line 7 holds. For these reasons, the primal solution is feasible.

Finally, the $x(e)$ and $x(v)$ of the nodes and edges which are involved in $m_i^O$ are increased by the shown formulas in the for-loops. As already mentioned in the background Section \ref{background}, the cost function of the competitive online algorithm determines its competitive ratios of the objective function and the capacity constraints. With this cost function, the COVCE algorithm is ($4 \cdot \mathrm{ln}(1+|G_S|\cdot C_{max}\cdot \alpha) + 1, 2)$-competitive, meaning a constant factor competitiveness on the capacity constraints and a logarithmic factor competitiveness on the objective function. In the next sections, the cost functions are analyzed in more details and the competitive ratios are proven.

\paragraph{Cost Functions}
In this section, the functions that describe the increase of the $x(e)$ and $x(v)$ primal variables are deeper explained and motivated. Since, they are used as costs for nodes and edges in the minimum-cost flow problem, these functions are called \emph{cost functions}. The most important property of a cost function $f(x)$ is that it is monotonically increasing for $x \geq 0$ because the primal variables are not allowed to decrease during the execution of the algorithm since in an online problem, previous decision cannot be reverted. In this case that means that a previous allocation cannot be undone. This characteristic is captured by non-decreasing primal variables and by the dual variables which can only be increased during the processing of their certain request. The cost of a node or edge represents how many allocations are already done on this node or edge. Since allocations are maintained forever, a monotonically increasing cost function is necessary to represent this online behavior because the allocation on a node or edge is also only monotonically increasing over time.

Consider now the structure of the cost function $x(e)$ in Equation \ref{eq52} to further analyze the monotonicity. It is mainly built of constant factors as $|G_S|$, $C_E$, $\mathtt{cap}(e)$ and an exponential function. As these factors are constant, they cannot cause a decrease of $x(e)$ during the execution of the algorithm. The exponential function only increases for positive exponents. Since the constant factors describe sizes and capacities, they are all non-negative. The same holds for $\alpha$ because it describes the maximum benefit and the benefit is always non-negative. Moreover, the dual variable $f(i,k)$ and the load are defined to be non-negative too. To sum up, the exponent is non-negative. Before concluding that the exponential function is monotonically increasing, consider the terms inside the exponent.

At first, the term $\alpha$ only occurs in the numerator of the exponential function and only increases during the execution since its value increases upon a maximization of its current value and the benefit of the current request. Hence, the variable $\alpha$ cannot cause a decrease of $x(e)$. Next, consider the term 
\begin{equation*}
	\sum_{r_i \in R} \sum_{m^k_i \in M(r_i)}f(i,k) \cdot l(i,k,e)
\end{equation*}
for a certain edge $e$. The first sum iterates over all requests $r_i$ that have been observed yet by the algorithm. The second sum iterates over all valid embeddings $m_i^k$ for a certain $r_i$. In addition, the product $f(i,k) \cdot l(i,k,e)$ is greater than 0 if the $k$-th embedding of the  $i$-request has been accepted and if $e$ is included in $m_i^k$ (because then $l(i,k,e) > 0$ would hold). Since $f(i,k) > 0$ only holds for exactly one embedding per request, the second sum only consists of one summand. In other words, $f(i,k)$ \emph{filters} the embedding that has been accepted by the algorithm for $r_i$. So, the first sum describes the load on $e$ by $r_i$. Summing over all the requests that have been observed yet, the whole term describes the sum of the loads on $e$ of all requests, giving the total load on edge $e$. As the total load on $e$ cannot decrease during the execution (because an accepted request stays forever in the substrate) this term cannot cause a decrease of $x(e)$ because it is also contained in the numerator of the exponential function. In conclusion, the cost function is growing \emph{monotonously} and is therefore valid to be used for changing the primal variable $x(e)$ during the execution. Analogously, the $x(v)$ cost function is built where $C_E$ and $\mathtt{cap}(e)$ are substituted by $C_V$ and $\mathtt{cap}(v)$.

After showing the monotonicity of the cost function, its structure is further explained. It looks similar to the cost function of the generic competitive online  algorithm in \cite{survey}, section 4.2 "Three Simple Algorithms", Algorithm 2. 
The next formula from \cite{survey} illustrates the cost function:
\begin{equation}
	x_i = \frac{1}{d} \bigg[exp\bigg(\frac{\mathrm{ln}(1+d)}{c_i} \sum_{j|i\in S(j)}y_j\bigg )-1\bigg ]
\end{equation}
The variable $j$ iterates over all requests and $i$ over all edges of the substrate graph. The set $S(j)$ contains all $x_i$ for which $a(i,j) = 1$ and $a(i,j) = 1$ if edge $i$ is included in the solution ($a(i,j)$ being the matrix consisting of the LHS of the constraints). To capture the biggest set found yet over all requests, $d$ is defined as $d =\mathrm{max}_j|S(j)|\leq n$. Analogously, for the embedding problem, $d$ would be the size of the biggest embedding found yet because a primal constraint in which all $a(i,j)$ = 1 would mean that the whole substrate is included in the embedding. For that reason the variable $d$ of Algorithm 2 is the analogon to the product $|G_S| \cdot C_E$ for edges and $|G_S| \cdot C_V$ for nodes in the node and edge cost functions.

As Algorithm 2 uses an exponential function for achieving a monotonically increasing function, the COVCE algorithm uses one too because the exponential function is monotonously increasing for positive exponents. The more detailed explanation why this function is chosen is discussed in \cite{survey}. The authors present a derivation upon a linear differential equation approach. 

The logarithmic terms in the exponents of both functions have a similar functionality since $d$ describes the size of the biggest solution and analogously $|G_S| \cdot C_E$ describes the maximum possible allocation. In addition, the logarithmic term contains the variable $\alpha$ which is needed in that term to achieve the mentioned competitive ratio. The sum in the exponent of Algorithm 2 sums the dual variables for which edge $i$ has been part of the solution. This is again analog to the summation of the loads on edge $e$ of the accepted embeddings which include edge $e$.

Lastly, the parameter $\alpha$ is explained. One can ask oneself why $b(r_i)$ instead of $\alpha$
is not sufficient in the exponent of the cost function. The benefit $b(r_i)$ for serving request $r_i$ has to be upper bounded by a parameter $\alpha \geq 1$ because the benefit is part of the cost function and the cost function is not allowed to decrease the value of $x(e)$ nor $x(v)$. This could be possible if $b(r_i)$ of the current request is smaller than the one from the previous request. Thus, the algorithm has to set this parameter $\alpha$ during its execution to upper bound the benefit, so that the value of the function is not anymore dependent on the variable benefit. Initially, $\alpha$ is set to 1 and then always increases as soon as the current benefit $b(r_i)$ is greater than the current $\alpha$. In this way, it is ensured that $\alpha$ never decreases. As shown in Algorithm \ref{alg1}, $\alpha$ is increased before the costs of the request's embedding are checked against its benefit. If this step would be omitted, then any request $r_i$ with $b(r_i) \gg \alpha$ has a chance to be accepted, even though the computed embedding contains a node or edge whose capacity is already violated by the factor in the competitive ratio. Consequently, the competitive ratio is violated. Hence, $\alpha$ is needed in the cost function in order to adjust the costs to the upcoming scenario in case a higher benefit than the current $\alpha$ occurs. For this reason, the primal variables except $z(r_i)$ are adjusted after updating $\alpha$ and before deciding the acceptance of $r_i$. With the following sample calculations, the importance of the $\alpha$-update before the if-condition will be shown.

After motivating the cost functions, its co-domain is briefly analyzed by looking at border values. The arguments of the function are the variables that are not constant, namely $\alpha$ and $l(i,k,e)$. Consider the case in which the load on $e$ is zero and the value $x(e)$ is being calculated. Since there is no load yet on $e$, $l(i,k,e) = 0, \forall r_i \in R, m_i^k \in M(r_i)$ holds.
\begin{equation*}
	\sum_{r_i \in R} \sum_{m^k_i \in M(r_i)}f(i,k) \cdot l(i,k,e) = 0
\end{equation*}

Thus, the double sum in the exponent is equal to 0. Consequently, the whole exponent is 0. Therefore $e^0 -1 = 0$ and hence $x(e) = 0$. Another interesting border case is the case when the total load on $e$ equals its capacity. Formally:

\begin{equation*}
	\sum_{r_i \in R} \sum_{m^k_i \in M(r_i)}f(i,k) \cdot l(i,k,e) = \mathtt{cap}(e)
\end{equation*}
So, by substituting the double sum with $\mathtt{cap}(e)$, the following term results:
\begin{alignat*}{7}
x(e) = \frac{1}{|G_S|\cdot C_E} \cdot \left[\mathrm{\mathrm{exp}}\left(\frac{\mathrm{ln}(1+|G_S|\cdot C_E\cdot \alpha)}{\mathtt{cap}(e)} \cdot \mathtt{cap}(e)\right)-1\right]
\end{alignat*}
Next, the fraction ${\mathtt{cap}(e)}/{\mathtt{cap}(e)}$ is reduced to 1, so that there is only the logarithmic term left in the exponent. Thus, we can reduce the term again by using $e^{ln(x)} = x$.
\begin{alignat*}{7}
x(e) \quad= & \quad\frac{1}{|G_S|\cdot C_E} \cdot [(1+|G_S|\cdot C_E\cdot \alpha)-1]
\\ 
= &\quad \frac{|G_S|\cdot C_E\cdot \alpha}{|G_S|\cdot C_E}
\\
= &\quad \alpha
\end{alignat*}
As a result, if the total load on an edge $e$ or node $v$ is equal to its capacity, its $x(e)$ or $x(v)$ value is equal to $\alpha$. Since the cost function is monotonically increasing, $x(e) \geq \alpha$ holds for any total load on $e$ being greater than its capacity $\mathtt{cap}(e)$. Observe that if the total load is equal to the capacity, then this node or edge is fully satisfied and further allocations on this node or edge will oversubscribe it. For this case, consider the if-condition of the algorithm again. To accept a request $r_i$, the following inequality has to hold:
\begin{equation}
	\mathrm{cost}(m_i^O) < b(r_i) \label{eq60}
\end{equation}
If one of the nodes or edges is already fully satisfied and is included in $m_i^O$, the following inequality holds:
\begin{equation*}
\mathrm{cost}(m_i^O) \geq \alpha
\end{equation*}
Inserting this into the if-condition Equation \ref{eq60} results in:
\begin{equation}
\alpha \leq \mathrm{cost}(m_i^O) < b(r_i)
\end{equation}
This inequality is false because $\alpha \geq b(r_i)$ holds. Thus, the if-condition is not satisfied and consequently the request will be rejected by the algorithm. In conclusion, if a fully satisfied node or edge might occur in the optimal embedding of the VC-ACE algorithm, the algorithm will reject the request. This property is used in the next section to proof the competitive ratios of the algorithm.
\paragraph{Theorems and Proofs}\label{proofs}

Now, the before mentioned competitive ratios of the COVCE algorithm are proven. At first, consider the following theorem.
\begin{theorem}
	Let $\alpha \geq b(r_i)$ be an upper bound on the benefit of any request, $|G_S|$ the size of the substrate graph and $C_{max} = \mathrm{max}\{C_E,C_V\}$ be highest capacity of the union of nodes and edges. Then the proposed algorithm produces a primal feasible solution and is $(4 \cdot \mathrm{ln}(1+|G_S|\cdot C_{max}\cdot \alpha) + 1, 2)$-competitive. The first argument of the tuple describes the competitive ratio of the objective function and the second argument describes the competitive ratio of the violation of resource constraints. This ratio is 2, meaning that each dual constraint is violated by at most factor 2.
\end{theorem}

\paragraph{Proof of primal feasible solution:}
Upon the arrival of a new request $r_i$, if $b(r_i) > \alpha$, then $\alpha$ is set to $b(r_i)$ and all $x(e)$ and $x(v)$ variables are updated with this new $\alpha$. If a valid mapping $m^k_i$ exists and $m^k_i$ is accepted, the algorithm updates the corresponding $x(e)$ and $x(v)$ primal variables with the additional load on the corresponding nodes and edges. In order to achieve feasibility, namely to satisfy the primal constraint in Equation \ref{eq5}, $z(r_i)$ is set to the difference which is missing to satisfy this constraint. Furthermore, $z(r_i)$ is never less than 0 because if the if-condition of the algorithm is not satisfied by the cost optimal embedding, then the request will be rejected. The primal variables are never decreased because the cost functions of $x(e)$ and $x(v)$ are monotonically increasing for non-negative arguments and $z(r_i)$ is non-negative too. The RHS of a primal constraint $i$ is the constant benefit $b(r_i)$. Furthermore, when accepting a request, the constraint $i$ is always fulfilled since $z(r_i)$ adds the missing difference to satisfy $i$ and is only set in the $i$-th round and is therefore constant. Due to the monotonicity of the LHS of the primal constraints, they can never be violated. Concluding, the algorithm produces a primal feasible solution. $\square$

\paragraph{Proof of 2-competitiveness on the capacity violation:}

To recapitulate, when the load on edge $e \in E_S$ equals its capacity, then $x(e) = \alpha$ holds. Upon the arrival of a new request which is accepted by the algorithm, we can infer for the cost of a specific edge $e \in E_S$:
\begin{alignat}{7}\label{eq10}
\resizebox{.9 \textwidth}{!} 
{$
	x(e) = \frac{1}{|G_S|\cdot C_E} \cdot \left[\mathrm{\mathrm{\mathrm{exp}}}\left(\frac{\mathrm{ln}(1+|G_S|\cdot C_E\cdot \alpha)}{\mathtt{cap}(e)} \cdot (\sum_{r_i \in R} \sum_{m^k_i \in M(r_i)}f(i,k) \cdot l(i,k,e))\right)-1\right] < b(r_i)
$}
\end{alignat}
This inequality is true because the if-condition line 7 of Algorithm \ref{alg1} must have evaluated to true in order to accept $r_i$. In addition, if $m^k_i$ is embedded and $e \in \mathcal{M}^k_{E,i}$, the load on an edge $e \in E_S$ is lower bounded by 1 ($l(i,k,e) \geq 1$) because $\mathcal{B}_i \geq 1$. The parameter $f(i,k)$ is 1 for the mapping $m^k_i$ that has been chosen by the algorithm out of the set of valid mapping $M(r_i)$ and else 0. Hence, the product  $f(i,k) \cdot l(i,k,e)$ is only greater than 0 for the chosen mapping. Since the load on edge $e$ increases per additional embedding it is included in, the cost $x(e)$ also increases because the exponent grows. From Equation \ref{eq10} it follows that 
\allowdisplaybreaks
\begin{alignat}{7}\label{eq11}
 &(1+ |G_S|\cdot C_E\cdot \alpha)^{\frac{\sum_{r_i \in R} \sum_{m^k_i \in M(r_i)}f(i,k) \cdot l(i,k,e)}{\mathtt{cap}(e)}} &<& 1+|G_S|\cdot C_E \cdot b(r_i) 
\\\label{eq12}
&&<& 1 + |G_S|\cdot C_E \cdot \alpha
\\\label{eq13}
\Leftrightarrow	&\quad\quad\quad\frac{\sum_{r_i \in R} \sum_{m^k_i \in M(r_i)}f(i,k)\cdot l(i,k,e)}{\mathtt{cap}(e)} &<&  1
\\\label{eq14}
\Leftrightarrow	&\quad\quad\quad\quad\quad\sum_{r_i \in R} \sum_{m^k_i \in M(r_i)}f(i,k) \cdot l(i,k,e)&< & \mathtt{cap}(e)
\end{alignat}
From Equation \ref{eq11} to \ref{eq12}, the definition of $b(r_i) \leq \alpha$ is used in order to upper bound the benefit and to have the same basis for a logarithmic operation which is done in Equation \ref{eq13}.
The Equation \ref{eq14} is the same as the definition of the dual linear program constraints in Equation \ref{eq2} for the edges which shows that the capacity constraints will not be violated as long as $x(e) \leq \alpha$ which is guaranteed by the constraint $x(e) \leq b(r_i)$ in Equation \ref{eq10}. If $x(e) > b(r_i)$ then the algorithm will not put any more load on this edge. The proof for the node capacity constraints is analog to this one. So, upon the acceptance of $r_i$, which means that the if-condition in line 7 of Algorithm \ref{alg1} evaluated to true, it is ensured that the embedding does not oversubscribe the capacity of any node or edge.

After accepting $r_i$ and knowing that Equation \ref{eq14} holds, the maximum additional load that can be put on a node or edge $y \in E_S \cup V_S$, that is part of $m_i^O$, is its capacity $\mathtt{cap}(y)$ because the VC-ACE algorithm will not put more load on $y$ than its capacity. Hence, the new load after accepting $r_i$ is upper bounded by:
\begin{alignat}{7}
\resizebox{.9 \textwidth}{!} 
{$
	\sum_{r_j \in R \setminus\{r_i\}} \sum_{m^k_j \in M(r_j)}f(j,k) \cdot l(j,k,y) + \mathtt{cap}(y)< \mathtt{cap}(y) + \mathtt{cap}(y) = 2 \cdot \mathtt{cap}(y)
$}
\end{alignat}
Since, the load is now greater than the capacity on $y$, $y$ will not be included in any further embedding in upcoming requests because the costs for $y$ are now greater than $\alpha$, even when $\alpha$ is increased later on, as shown in the previous section. In conclusion, the maximum load on a node or edge is twice its capacity, resulting in a violation of at most factor 2. $\square$

After showing the competitive ratio for the dual constraints, the competitive ratio for the objective function is shown.
\paragraph{Upper bounding the Competitiveness of the Objective Function:}
\newtheorem{lemma}[theorem]{Lemma}
\begin{lemma}
	Consider the formula $\sum_{e \in \mathcal{M}_{E,i}^k} l(i,k,e)$ which is the sum of the loads on the edges $e \in \mathcal{M}_{E,i}^k$ when embedding mapping $m^k_i$ of request $r_i$. Furthermore, the VC-ACE algorithm will not put more load on an edge than its initial capacity. Thus, the load $l(i,k,e)$ can be upper bounded by the maximum capacity over all edges. Hence, the following inequality holds:
	$\sum_{e \in \mathcal{M}_{E,i}^k} l(i,k,e) \leq	|\mathcal{M}_{E,i}^k| \cdot C_E$.
\end{lemma}
Let $P$ and $D$ be the values of the objective functions of the primal and dual linear program, respectively. Initially $P=D=0$ since no request has been embedded yet.
The derivative with respect to $f(i,k)$ of the dual objective function's value $D$ is 
\begin{alignat}{7}
	\frac{\delta D}{\delta f(i,k)} = b(r_i)
\end{alignat} 
since the value of the new generated dual variable corresponding to the incoming request is set to 1 if the request is accepted. For computing the derivative of the primal objective function, let $P_E. P_V$ and $P_z$ be:
\begin{alignat*}{7}
	P_E &= \sum_{e \in E_S} \mathtt{cap}(e) \cdot x(e)\\
	P_V &= \sum_{v \in V_S} \mathtt{cap}(v) \cdot x(v)\\
	P_z &= \sum_{r_i \in R} z(r_i).\\
	\text{such that } 	P &= P_E + P_V + P_z.
\end{alignat*}
These variables are used to determine the derivative of the primal objective function in separated terms.
\begin{alignat}{7}\label{eq16}
	\frac{\delta P_E}{\delta f(i,k)} =&\quad \sum_{e \in \mathcal{M}_{E,i}^k} \frac{\mathtt{cap}(e)}{|G_S|\cdot C_E} \cdot \bigg[\frac{\mathrm{ln}(1+ |G_S|\cdot C_E\cdot \alpha)}{\mathtt{cap}(e)} \cdot l(i,k,e) 
	  \\ \nonumber &\cdot\mathrm{\mathrm{\mathrm{exp}}} \bigg(\frac{\mathrm{ln}(1+|G_S|\cdot C_E\cdot \alpha)}{\mathtt{cap}(e)} 
	\quad\sum_{r_i \in R} \sum_{m^k_i \in M(r_i)}f(i,k) \cdot l(i,k,e)\bigg)\bigg]
\\\label{eq17}
	=& \quad\mathrm{ln}(1+ |G_S|\cdot C_E\cdot \alpha) \sum_{e \in \mathcal{M}_{E,i}^k} l(i,k,e) \bigg(\frac{1}{|G_S|\cdot C_E}
	 \\ \nonumber &
	\resizebox{.78 \textwidth}{!} 
	{$
	\cdot\bigg[\mathrm{\mathrm{\mathrm{exp}}} \bigg(\frac{\mathrm{ln}(1+|G_S|\cdot C_E\cdot \alpha)}{\mathtt{cap}(e)} 
	\quad
	\sum_{r_i \in R} \sum_{m^k_i \in M(r_i)}f(i,k) \cdot l(i,k,e)\bigg)-1\bigg] + \frac{1}{|G_S|\cdot C_E}\bigg)
	$}
\\\label{eq18}
	=&\quad \mathrm{ln}(1+|G_S|\cdot C_E\cdot \alpha) \sum_{e \in \mathcal{M}_{E,i}^k} l(i,k,e) (x(e) + \frac{1}{|G_S|\cdot C_E})
\\\label{eq19}
	\leq&\quad \mathrm{ln}(1+|G_S|\cdot C_E\cdot \alpha)\cdot (b(r_i) + \frac{1}{|G_S|\cdot C_E} \cdot \sum_{e \in \mathcal{M}_{E,i}^k} l(i,k,e))
\\\label{eq21}
	\leq&\quad \mathrm{ln}(1+|G_S|\cdot C_E\cdot \alpha)\cdot (b(r_i)+\frac{1}{|G_S|\cdot C_E} \cdot |\mathcal{M}_{E,i}^k| \cdot C_E) 
\\\label{eq22}
	\leq&\quad \mathrm{ln}(1+|G_S|\cdot C_E\cdot \alpha)\cdot (b(r_i)+1) 
\end{alignat}
Equation \ref{eq16} shows the derivative of the primal cost caused by the changes of the edges which where affected by the embedding of request $r_i$. In Equation \ref{eq17}, the term is described in an equivalent way, such that its subterm 
\begin{equation*}
\frac{1}{|G_S|\cdot C_E}
\cdot \bigg[\mathrm{exp} \bigg(\frac{\mathrm{ln}(1+|G_S|\cdot C_E\cdot \alpha)}{\mathtt{cap}(e)} \sum_{r_i \in R} \sum_{m^k_i \in M(r_i)}f(i,k) \cdot l(i,k,e)\bigg)-1\bigg]
\end{equation*}
represents the function of $x(e)$ which is then replaced by $x(e)$ in Equation \ref{eq18}. In Equation \ref{eq19} the true-evaluated if-condition in line 7 of Algorithm \ref{alg1} is used to upper bound the term $\sum_{e \in \mathcal{M}_{E,i}^k} l(i,k,e)\cdot x(e) < b(r_i)$. In Equation \ref{eq21} the load $l(i,k,e)$ is upper bounded by the maximum capacity over all edges $e \in E_S$ due to Lemma 2. Therefore, the sum $\sum_{e \in \mathcal{M}_{E,i}^k} l(i,k,e)$ is upper bounded by the product $|\mathcal{M}_{E,i}^k| \cdot C_E$. Finally in Equation \ref{eq22} the fraction $\frac{|\mathcal{M}_{E,i}^k|\cdot C_E}{|G_S|\cdot C_E}$ is upper bounded by 1 because $G_S > |\mathcal{M}_{E,i}^k|$.

The derivative of the nodes is analogously 
\begin{equation}\label{eq23}
\frac{\delta P_V}{\delta f(i,k)} \leq	\mathrm{ln}(1+|G_S|\cdot C_V\cdot \alpha)\cdot (b(r_i)+1) 
\end{equation}
The derivative of the primal variable $z(r_i)$ equals the value that it is assigned to in line 11 of Algorithm \ref{alg1} since $z(r_i)$ is introduced and set in the iteration in which $r_i$ is handled and is initially 0. $z(r_i)$ is upper bounded by the benefit $b(r_i)$ since the value of $z(r_i)$ can be at most $b(r_i)$:
\begin{equation}\label{eq24}
	\frac{\delta P_z}{\delta f(i,k)} \leq b(r_i) - \mathrm{cost}(m_i^O)\leq b(r_i) 
\end{equation}
Summing up all the partial derivatives concludes in the change of the primal profit:
\begin{equation}\label{eq25}
	\frac{\delta P}{\delta f(i,k)} \leq \mathrm{ln}(1+|G_S|\cdot C_E\cdot \alpha)\cdot (b(r_i)+1) + \mathrm{ln}(1+|G_S|\cdot C_V\cdot \alpha)\cdot (b(r_i)+1) + b(r_i)
\end{equation}
Let $C_{max} = \mathrm{max}\{C_E,C_V\}$. Then the derivative of the primal profit is upper bounded by
\begin{equation}
 \frac{\delta P}{\delta f(i,k)} \leq 2 \cdot \mathrm{ln}(1+|G_S|\cdot C_{max}\cdot \alpha)\cdot (b(r_i)+1) + b(r_i)
\end{equation}
So, after accepting the incoming request, the change in the dual solution is $b(r_i)$ and in the primal solution at most $2 \cdot \mathrm{ln}(1+|G_S|\cdot C_{max}\cdot \alpha)\cdot (b(r_i)+1) + b(r_i)$. 

To recapitulate, the competitive ratio is the quotient of the profit of the optimal offline solution and the profit of the online solution in the worst case. The change of the latter is known as ${\delta D}/{\delta f(i,k)}$, but the optimal offline solution is unknown. Nevertheless, it is known that the optimal offline solution is feasible. Hence, the idea is to apply weak duality on the optimal offline solution and the primal feasible solution of the primal-dual algorithm, so that the optimal offline profit is upper bounded by the primal cost. Now, the quotient of the worst case primal cost and dual profit gives an upper bound on the competitive ratio of the primal-dual algorithm. Alternatively, the derivatives of the profit and the cost can be used to, as shown in the background Section \ref{background}. Formally:
\begin{alignat}{7}
	\frac{\delta D}{\delta f(i,k)} &\leq \frac{\delta P}{\delta f(i,k)}\\
	\Leftrightarrow  b(r_i) &\leq 2 \cdot \mathrm{ln}(1+|G_S|\cdot C_{max}\cdot \alpha)\cdot (b(r_i)+1) + b(r_i)\\
	\Leftrightarrow 1 &\leq   2 \cdot \mathrm{ln}(1+|G_S|\cdot C_{max}\cdot \alpha)\cdot (1+\frac{1}{b(r_i)}) + 1 \\
	&\leq  2 \cdot \mathrm{ln}(1+|G_S|\cdot C_{max}\cdot \alpha)\cdot (1+\frac{1}{1}) + 1 \\
	&=  4 \cdot \mathrm{ln}(1+|G_S|\cdot C_{max}\cdot \alpha) + 1=: \beta
\end{alignat}
In (49) to (50), $1/b(r_i)$ is upper bounded by 1. In (51) we define the factor $\beta$, i.e. the factor by which the cost and profit maximally differ. $\beta$ is an upper bound on the competitive ratio between the optimal offline profit and the online profit because the primal cost upper bounds the optimal offline profit due to weak duality. Thus, in the worst case, when the dual online profit increases by 1, the primal cost and therefore also the optimal offline profit increases by at most $\beta$. 
Hence, the competitive ratio on the objective function is $\beta$. In conclusion, the COVCE algorithm is $4 \cdot \mathrm{ln}(1+|G_S|\cdot C_{max}\cdot \alpha) + 1$-competitive on the objective function. $\square$

\paragraph{Remark 1.}

An interesting modification would be to initially set $\alpha$ as a read-only parameter. To achieve this, the user has to determine \emph{in advance} the highest possible benefit over the whole execution time. If this information is given, $\alpha$ can be hard-coded before execution and therefore is already the maximum benefit during the whole execution. This is beneficial for the provider since the algorithm already knows the highest possible benefit and will be therefore stricter with accepting requests. However, estimating the highest benefit can lead to a violation of the competitive ratio on the capacity constraints  if the estimation is wrong because then a request with a higher benefit than the estimated value might be embedded, even if this embedding involves edges/nodes with a capacity violation of 2.
\paragraph{Remark 2.} The COVCE algorithm is in general applicable to embedding problems in which cost-optimal embeddings are computed because it expects an embedding from an oracle procedure and then adjusts the cost for the nodes and edges. Thus, this oracle procedure is generic and can therefore include diverse embedding models. Relating to the VC model, the VC-ACE algorithm \cite{vcace} has been used as the oracle procedure. 

\subsubsection{COVCEload Algorithm} Next, we present a variant of the COVCE algorithm, the COVCEload algorithm. The only difference to the COVCE algorithm is that the VC-ACE algorithm, working as the oracle procedure in the COVCE algorithm, receives the residual capacities and not the initial capacities when being called from the COVCE algorithm. Due to this variation, no capacities are violated but the algorithm is not competitive anymore. Its advantages and disadvantages are further discussed in the evaluation Section \ref{evaluation}.

\subsubsection{The GVOP Algorithm}
After proving and discussing the competitive ratios of the COVCE algorithm, the second competitive algorithm of this thesis, designed with the help of \cite{gvop}, is presented. 
The \textbf{G}eneral all-or-nothing \textbf{V}Net Packing \textbf{O}nline \textbf{P}acking algorithm as described in \cite{gvop} is a generic algorithm for designing competitive online virtual network embedding argorithms. It is based on the primal-dual approach for online problems of \cite{survey}. Its main functionality is also to decide whether to accept the found minimal-cost embedding or not. The GVOP algorithm is applicable to the virtual cluster embedding problem because the VCE problem is a special case of the general VNet embedding problem. The main difference to the COVCE algorithm is the formula of the costs for the primal variables $x(e)$ and $x(v)$ which also induces the differences in the competitive ratios. 

More formally, the authors assume a capacitated substrate network $G_S = (V_S,E_S)$ with undirected edges $\{u,v\} \in E_S$ for $u,v \in V_S$ , where each edge $e$ has a capacity $\mathtt{cap}(e) \geq 1$ and a cost $\mathtt{cost}(e) \geq 0$. Furthermore, a sequence of requests $R = \langle r_1,r_2,r_3,...\rangle$ arrives in an online fashion and each $r_i$ shall be either accepted and embedded or rejected by the online embedding algorithm. A request consists of the virtual network (VNet) that shall be embedded and a benefit $b(r_i)$ that the provider receives for embedding the request. An embedding involves a mapping, placing the VNet's nodes onto the substrate nodes, and a set of allowed traffic patterns between these nodes, which is the analogon to the set of physical paths on which the virtual links of the VC model are mapped to. 

This model can be extended to suit the VCE problem. To this end, node capacities $\mathtt{cap}(v) \geq 1, v \in V_S$ and costs $\mathtt{cost}(v) \geq 0$ have to be introduced to the substrate network. Consequently, another set of dual constraints is created for the node capacities and accordingly another set of primal variables $x(v)$ which represent the node costs as in the COVCE algorithm. The nodes' cost are assigned with the corresponding $x(v)$ values that are calculated by the GVOP algorithm. These new constraints and variables are maintained independently of the edge constraints and variables. Furthermore, the edge cost formula is independent from the set of nodes and vise versa.   

Moreover, it also provides the \emph{all-or-nothing} characteristic and assumes that no allocation of a single request on an edge or node $y \in E_S \cup V_S$ oversubscribes its capacity. Though, it is still possible that the composition of multiple allocations of multiple requests may oversubscribe its capacity. The VC-ACE algorithm also does not oversubscribe the initial capacities when being called but multiple calls may oversubscribe them. Thus, the VC-ACE algorithm is applied as oracle procedure to compute a minimum-cost embedding where the costs are given by the $x(e)$ and $x(v)$ calculated by the GVOP algorithm.

The structure of the pseudo-code in Algorithm \ref{alg2} looks similar to the COVCE algorithm. The only addition, that is made to the original GVOP algorithm in \cite{gvop}, is the for-loop handling the primal variables of the nodes. The notation being used is the same as for the COVCE algorithm.
\begin{algorithm}[ht]
	\caption{GVOP Algorithm\label{alg2}}
	
	Initially all $x(e)$ and $x(v)$ are set 0\;
	Upon the arrival of a new request $r_i$, the corresponding primal constraint in Equation \ref{eq5} and the dual variable $f(i,k)$ are introduced:\\
	\If{there exists an embedding $m_i^O \in M(r_i)$ with cost($m_i^O) < b(r_i)$}{
	//\textit{\textbf{accept} $r_i$}\\
	\Indp $f(i,O) \leftarrow 1$\;
	$f(i,k) \gets 0$ for all other valid embeddings $m_i^k \in M(r_i)\setminus \{m_i^O\}$\;   	
	 $z(r_i) \leftarrow b(r_i)$ - cost($m_i^O)$	\;
	 \For{$e \in \mathcal{M}_{E_S,i}^O$}{
	 	$x(e) \leftarrow x(e) \cdot 2^{l(i,O,e) / \mathtt{cap}(e)} + \frac{1}{\sum_{e \in \mathcal{M}_{E,i}^O} l(i,O,e)} \cdot(2^{l(i,O,e) / \mathtt{cap}(e)} - 1)$
	 	
	 }
	 \For{$v \in \mathcal{M}_{V,i}^O$}{

	 $x(v) \leftarrow x(v) \cdot 2^{l(i,O,v) / \mathtt{cap}(v)} + \frac{1}{\sum_{v \in \mathcal{M}_{V,i}^O} l(i,O,v)} \cdot(2^{l(i,O,v) / \mathtt{cap}(v)} - 1)$
	 }
	}
	\Else{
		// \textit{\textbf{reject} $r_i$}\\
		$z(r_i) \gets 0$\;
		$f(i,k) \gets 0$ for all $k$\;   
		// \textit{all $x(e), e \in E_S$ and $x(v) , v \in V_S$ are unchanged}\\
		// \textit{$\rightarrow$ change in primal and dual objective functions is 0}
	}
\end{algorithm}

According to the work presented in \cite{gvop}, the GVOP algorithm is (2, $\mathrm{log}_2 (1+ 3\cdot (\mathrm{max}_{i,k}w(i,k)\cdot \mathrm{max}_{i}b_i))$)-competitive, where \\$w(i,k) = \sum_{e \in \mathcal{M}_{E_S,i}^k} l(i,k,e) + \sum_{e \in \mathcal{M}_{V_S,i}^k} l(i,k,v)$ for an valid embedding $k$ of request $i$. So it is competitive on the objective function up to the constant factor 2 and has a logarithmic factor in the violation of the resource constraints of the dual program. The logarithmic factor has a similar structure to the one of the competitive ratio on the objective function of the COVCE algorithm. Both involve the product of the maximum benefit and the maximum allocation on the substrate network. The proofs for the competitive ratios of the GVOP are shown in \cite{gvop}.

As a result, a logarithmic factor in the capacity violation for the GVOP algorithm is worse than a violation of at most 2 for the COVCE algorithm when no capacity multiplication is allowed.
Nevertheless, the GVOP algorithm has therefore a 2-competitiveness on the objective function but the COVCE algorithm's competitive ratio is logarithmic. The competitive ratio tuples of both algorithms are the contrary of each other. Therefore, it is interesting to find out the advantages and disadvantages of both algorithms in certain scenarios which, among other things, is investigated in the evaluation Section \ref{evaluation}.

\paragraph{Remark 3.} In the GVOP algorithm's paper \cite{gvop}, Corollary 2 says that the GVOP algorithm serves a dual feasible solution under a certain assumption when scaling down the capacities by the competitive ratio $\beta$ on the resource constraints. This assumption says that $\mathrm{min}_e \mathtt{cap}(e)/\beta \geq \mathrm{max}_{i,k} l(i,k,e)$. In words, in order to result in a dual feasible solution, the minimum substrate edge capacity divided by $\beta$ has to be greater or equal to the maximum potential load induced by any feasible embedding for any request on the respective resource. This means that no request is allowed to allocate more bandwidth than the minimum edge capacity in the substrate divided by $\beta$ on an edge. This assumption may not be applicable in realistic scenarios considering that the minimum capacity is defined to be greater or equal to 1 in this model.

In the following evaluation, multiple scenarios and configuration are presented in order to figure out where each algorithm has its advantages and disadvantages.

\section{Evaluation}\label{evaluation}
In this section, the presented algorithms are evaluated together with a fourth algorithm. The latter one is the VC-ACE algorithm itself. By including it, the objective is to find out whether accepting any valid embedding (as a non-competitive online algorithm does) is more profitable for the provider in certain scenarios than using competitive-online algorithm. So, besides running the GVOP, the COVCE and the COVCEload algorithm, a fourth algorithm, the Greedy VC-ACE algorithm is implemented which receives a sequence of requests in an online fashion as the other algorithms too. It applies the VC-ACE algorithm with respect to residual capacities when processing a request. 

The comparison between the four algorithms is done upon running them on different patterns of requests and interpreting their results whether one algorithm outperforms the others for a certain type of input. 
The metrics that are used to make a formal comparison are 
\begin{enumerate}
	\item the acceptance ratio, i.e. how many requests are accepted in a certain subsequence of requests,
	\item the capacity violation, i.e. the maximum load to capacity ratio over all substrate nodes and edges,
	\item the average capacity violation, i.e. the average load to capcity ratio over all substrate nodes and edges,
	\item the relative profit, i.e. the sum of benefits of accepted requests divided by the previously defined capacity violation,
	\item the average relative profit, i.e. the sum of benefits of accepted requests divided by the previously defined average capacity violation.
\end{enumerate}
The metrics are discussed in details later.
The main results are that COVCE algorithm has the best performance with respect to the average relative profit on all considered request sequences. The Greedy VC-ACE algorithm shows the best performance with respect to the relative profit when the benefit function is dependent on the size of the request/embedding. The COVCEload and COVCE algorithms exhibit the best performance when the benefit to embedding cost ratio is high, which is for example the case when the benefit function does not increase proportional to the size of the request/embedding. All in all, there is no algorithm that outperforms the others in all considered scenarios.
Before going into the analysis of the algorithms' performances, consider the technical details.
\subsection{Experimental Setup}
The experiments are run on the network testbed of the INET chair which includes servers having a Intel(R) Xeon(R) CPU L5420 running at 2.50GHz with 16GB RAM. The algorithms are implemented in Java 7 with multithreading to make use of the 8 cores of the server. 
\subsubsection{Implementation}
During the execution of any of the four algorithms, the VC-ACE algorithm is executed running the successive shortest path (SSP) algorithm for each node in the network being considered as the target node in order to find the best logical switch placement. More precisely, the SSP algorithm itself solves a minimum-cost flow problem in which negative costs also occur. Thus, the Bellman-Ford algorithm can be applied since the Dijkstra algorithm cannot deal with negative costs. However, the Bellman-Ford algorithm induces a noticeable higher running time if the network is large and the SSP algorithm runs several shortest path iterations so that the running times stack up. 

More precisely, the run time complexity of the Dijkstra algorithm is $O(|V_S|\mathrm{log}(V_S) + E_S)$ when using a Fibunacci-Heap \cite{sneyers2006dijkstra} and for the Bellman-Ford algorithm it is $O(|V_S|\cdot |E_S)$. Consider an example with 500 nodes and 1000 edges in the substrate graph $G_S$. This results in a value of approximately 4107 for the Dijkstra run time complexity and 500000 for the Bellman-Ford algorithm. In addition, the shortest path algorithm is executed several times per SSP iteration, dependent on the demand of the request represented by the $\mathcal{N}_i$  requested VMs. In this example the run time complexity of the Bellman-Ford algorithm can be up to 100 times higher than the one of the Dijkstra algorithm. In addition to that, the shortest path algorithm is executed at most $\mathcal{N}_i$ times because in each SSP iteration a demand of at least 1 is served if a path exists. On top of that, the SSP algorithm is run for each node (server and switches) in the network and the VC-ACE algorithm is run for each of the four competing algorithms.

So, the shortest path algorithm becomes the bottle-neck of the VC-ACE algorithm and therefore the Dijkstra algorithm needs to be applied to reduce the running time. To this end, the \emph{reduced cost} model with \emph{node potentials} \cite{networkflows} is applied which does not change the shortest path tree but avoids negative costs. Thus, the Dijkstra algorithm is again applicable. 

Another improvement of the running time is the \emph{early-termination} of the Dijkstra algorithm. Inside the SSP iteration, it is only necessary to find a shortest path from the source to the target node, so that the Dijkstra algorithm can be interrupted as soon as the target node has been removed from the data structure, that maintains seen nodes, because at this state the shortest distance to the source node is determined. This early-termination does not violate the shortest path tree and therefore the node potentials and reduced costs of the edges maintain valid as shown in \cite{networkflows}. Furthermore, \cite{networkflows} show that only the potentials of the nodes, that have been removed from the data structure yet, have to be updated. As a result, with the early termination and potential updating, the running time of a SSP iteration reduces a lot.

As already mentioned, per VC-ACE algorithm execution, there are multiple SSP iterations running, which are independent of each other, since the best logical switch placement shall be found. Because of the independence of the SSP iterations, multithreading is applicable in the execution of the VC-ACE algorithm, such that 8 threads are run in parallel to compute all SSP results for one VC-ACE execution.

\subsubsection{Substrate Network}
The following settings are applied for the substrate network. Each edge in the substrate graph has a bandwidth capacity of 20. Furthermore, each node has a compute units capacity of 20, meaning it can serve up to 20 VMs dependent on $\mathcal{C}_i$. Each request's demands are randomly picked out of certain intervals. The amount of requested VMs are chosen between \{3,4,5,...,14\}, the amount of requested bandwidth per edge is randomly chosen between 20-60\% of an edge's capacity (which is 20 for all edges) and is rounded to an integer value. The amount of compute units are randomly chosen between 20-100\% of a node's capacity (which is also 20 for all nodes) and is also rounded to an integer value. The interval for the requested bandwidth is smaller since an embedding needs at least two edges to connect two nodes if $\mathcal{N}_i>$ 1 and $\mathcal{C}_i >$ 10 \footnote{the VC-ACE algorithm does not allocate 2 VMs onto 1 PM if $\mathcal{C}_i > 10$ because then the capacity of this node would be violated and the VC-ACE algorithm does not violate the initial capacities}. Due to this fact, the bandwidth interval is smaller than the compute units interval in order to balance the utilization of nodes and edges, so that more embeddings can be accepted when the substrate network is close to be fully reserved. In the following, the \emph{request size} as term is used to describe the sum of VMs, bandwidth and compute units of a certain request.

After discussing the problem of how to implement the algorithms, the ideas of how to experiment with them are presented. The four mentioned algorithms are being tested considering different \emph{configurations}.
\subsubsection{Configuration}
A \emph{configuration} includes four parameters where each of them has an influence on the results. 

\paragraph{Network Topology:}
The first parameter is the \emph{network topology} (NT). It describes which network topology shall be chosen for the experiment. The fat-tree \cite{fattree} and the MDCube \cite{mdcube} are the network topologies being applied because they are widely considered data center topologies and also applied for evaluations as in \cite{vcace}, \cite{vc}, \cite{timedvc}. 

Dependent on the structure of the topology, the computed embedding for a certain request $r_i$ might consume a different amount of resources. If topology $A$ is built in a way that the embeddings are usually consuming more resources than in topology $B$, then $B$ might be able to serve more requests overall. Additionally, it is interesting to find out whether there is a correlation between a certain (competitive) online embedding algorithm and the network topology.
\paragraph{Topology Size:}
Furthermore, the second parameter is the \emph{topology size}. The size of the graph in general has an influence on the running time because, among other things, the algorithms have to iterate over all nodes in the VC-ACE algorithm. Another point is that the size of the graph also determines how many requests can be accepted overall. So, to test with  ``big'' virtual cluster requests, a  ``big'' graph is also needed because testing on a  ``small'' graph with  ``big'' requests may cause a too fast progression in the algorithm's behavior, such that some steps are skipped and information is lost. 
The sizes of the topologies are chosen analogously to the ones presented in \cite{vcace}. Therefore, a fat-tree(12), meaning 12 ports, is used and a MDCube(1,12,1,4) which means it is one-dimensional and consists of 4 BCubes \cite{bcube} with $n$=12 and $k$=1. The topologies respectively have 432 and 576 servers, so that the experiments are evaluated on similar sized graphs. According to these sizes, the range of the parameters of the requests have been chosen as shown in the latter section.
\paragraph{Benefit Pattern:}
The third parameter is the most interesting one, the \emph{benefit pattern} (BP). The BP is a function that assigns a certain benefit to each request. Different BP functions are tested in order to find out whether there is a correlation between the BP and the behavior of the algorithms. As observed later, this parameter has an important influence on the behavior of the COVCE algorithm and is therefore deeply discussed later on.

Next, the several BPs are explained. At first, there is the $\mathtt{random}$ BP which basically means that the benefit is a randomly chosen integer in the interval $b_i \in [1,1000]$ for all $r_i \in R$. This BP shall depict the scenario in which there is no correlation between the benefit and the size of the request. 

In contrast to this, the $\mathtt{reqsize}$ BP captures the scenario where the benefit of serving the request is proportional to the requested resources. This BP represents a more economical scenario in which the provider earns more money for leasing more resources. In other words, the tenant's cost which is the provider's benefit, increases proportionally to her demands. For this BP the benefit function is chosen to be the following, motivated by \cite{reqsize}. Consider the request $r_i = (\mathcal{N}_i,\mathcal{B}_i,\mathcal{C}_i, b_i)$ then $b_i = \mathcal{N}_i \cdot (\mathcal{B}_i + \mathcal{C}_i)$ is the benefit that increases proportionally to the amount of requested bandwidth, compute units and VMs. Referring to \cite{reqsize}, it is assumed that the cost per compute and bandwidth unit is one. 

However, this BP assumes that the embedding costs for nodes and edges are equal when $\mathcal{B}_i = \mathcal{C}_i$, although at least two additional edges are needed to embed one additional node on the fat-tree and the MDCube. Therefore, the $\mathtt{vcesize}$ BP is introduced respecting the amount of allocated physical links and nodes in the benefit function. In other words, while the $\mathtt{reqsize}$ BP considers the size of the request, the $\mathtt{vcesize}$ BP considers the size of the embedding. Formally:
\begin{align}\label{eq59}
	b(r_i) = \frac{|E_S|\cdot C_E}{|V_S| \cdot C_V}\cdot \bigg(\frac{\mathcal{N}_i \cdot \mathcal{C}_i}{|V_S|\cdot C_V}\bigg)^2 + \bigg(\frac{\mathrm{empirical\_cost}_E (\mathcal{N}_i, \mathcal{B}_i,\mathcal{C}_i)}{|E_S|\cdot C_E}\bigg)^2.
\end{align}
The left summand in Equation \ref{eq59} describes the benefit contribution due to node allocations and the right one for edges. More precisely, in the left summand, $q = (|E_S|\cdot C_E)/(|V_S| \cdot C_V)$ is the ratio between the overall edge capacity $|E_S|\cdot C_E$ and the overall node capacity $|V_S| \cdot C_V$ in the substrate network. So, there are $q$-times more edge capacities than node capacities. The factor $(\mathcal{N}_i \cdot \mathcal{C}_i)/(|V_S|\cdot C_V)$ describes the ratio between the demanded to the overall node capacities. This fraction is squared in order to induce superlinear increasing benefits for bigger requests as \cite{fattree} show that the relation between the estimated cost for the provider for a maximum possible number of hosts in a substrate network increases superlinear. The factor $q$ is multiplied to the squared factor to scale the node benefit up/down when there are more/less edges than nodes, so that node and edge benefits contribute similarly to $b(r_i)$. The right summand in Equation \ref{eq59} includes the function $\phi = \mathrm{empirical\_cost}_E (\mathcal{N}_i, \mathcal{B}_i,\mathcal{C}_i)$ which maps this triplet onto a value $\phi$ that tries to capture the embedding costs for this VC.
To generate empirical values for $\phi$, 100 experiments, consisting of 6400 requests respectively, were executed using the Greedy VC-ACE algorithm on a 12-fat-tree in order to receive multiple values for $\phi$ for each combination of $(\mathcal{N}_i, \mathcal{B}_i,\mathcal{C}_i)$. $\phi$ is then assigned the amount of allocated edges for serving $VC_i = (\mathcal{N}_i, \mathcal{B}_i,\mathcal{C}_i)$. 
Hence, for each of the 640,000 requests the result is logged, while invalid embeddings were not considered. At different states of the substrate network, a virtual cluster embedding with the same virtual cluster consumes a different amount of edges since longer distances occur between the nodes when the network is close to be fully satisfied. Thus, multiple $\phi$ values are measured for the same VC which are therefore averaged. In Equation \ref{eq59}, $\phi$ is then divided by the overall edge capacity to express how much percentage of the overall edge capacity $r_i$ consumes. This ratio is again squared in order to make bigger requests more expensive in a superlinear fashion.

Another interesting scenario is the case when there is a fluctuation of the benefit over the time. Consider a scenario in which the costs for an allocation are at day time more expensive than at night time, since the amount of arriving requests might be higher at day time. So, at day time the benefit for the provider is higher than at night time. This behavior is captured by the $\mathtt{wave}$ BP which is represented by the sinus function $f(x) = (300 \cdot sin(0.1\cdot x)) + 400$ where $x$ is the index of the current request. The amplitude 300 is smaller than the range of the $\mathtt{random}$ BP in order to find out whether this lower amplitude already shows differences in the behavior of the four algorithms. The frequency at 0.1 is chosen in order to increase the amount of periods over the request sequence.  Thus, the wave peaks with a benefit value of 700 stand for day times on which the request ratio is the highest and the wave troughs with a benefit value of 100 depict night times where the request ratio is the lowest. 

Finally, the $\mathtt{peak}$ BP is a scenario to show an extreme case for benefit fluctuations, i.e. there are single peaks of a very high benefit in between a sequence of very low benefits. For example, such a sequence of benefits could look like this: $\langle1000,1,1,...,1,1,1,1000,1,1,1,1,...,1,1,1,1000,1,1,1,...\rangle$. For the $\mathtt{peak}$ experiments, every 100th request, a peak is included. There are three different $\mathtt{peak}$ scenarios with the peak values 1000, 100 and 10 in order to find out in which of these three scenarios noticeable changes of the algorithms' behavior show up. The difference to the $\mathtt{wave}$ scenario is that a wave has frequent wave peaks followed by frequent wave troughs whereas there is a single peak every 100th request in the $\mathtt{peak}$ scenario. 

\paragraph{Request Sequence Length:}
The fourth parameter is the \emph{request sequence length} (RSL) which describes the length of the sequence of requests which is the input to the tested algorithm. This parameter is determined dependent on the progress of the \emph{acceptance ratio} over the sequence of requests. The acceptance ratio describes how many requests are accepted over a certain interval of requests. More precisely, the sequence of requests is divided into windows of 100 requests and in each of these windows the amount of accepted requests is counted. Thus, the acceptance ratio for a certain window $w$ is the ratio $\frac{\text{accepted requests in } w}{100}$.	

The influence of the RSL on the results shows up as soon as the acceptance ratio of the competitive algorithms decreases since as observed later on, the COVCE algorithm will stop accepting any requests much earlier than the GVOP algorithm. In order to observe this, a certain RSL has to be maintained to not miss important events. In Figure \ref{fig37} the acceptance ratios of two fat-tree and two MDCube experiments are shown with four of the seven benefit patterns. 
\begin{figure}	
	\includegraphics[width=0.5\textwidth]{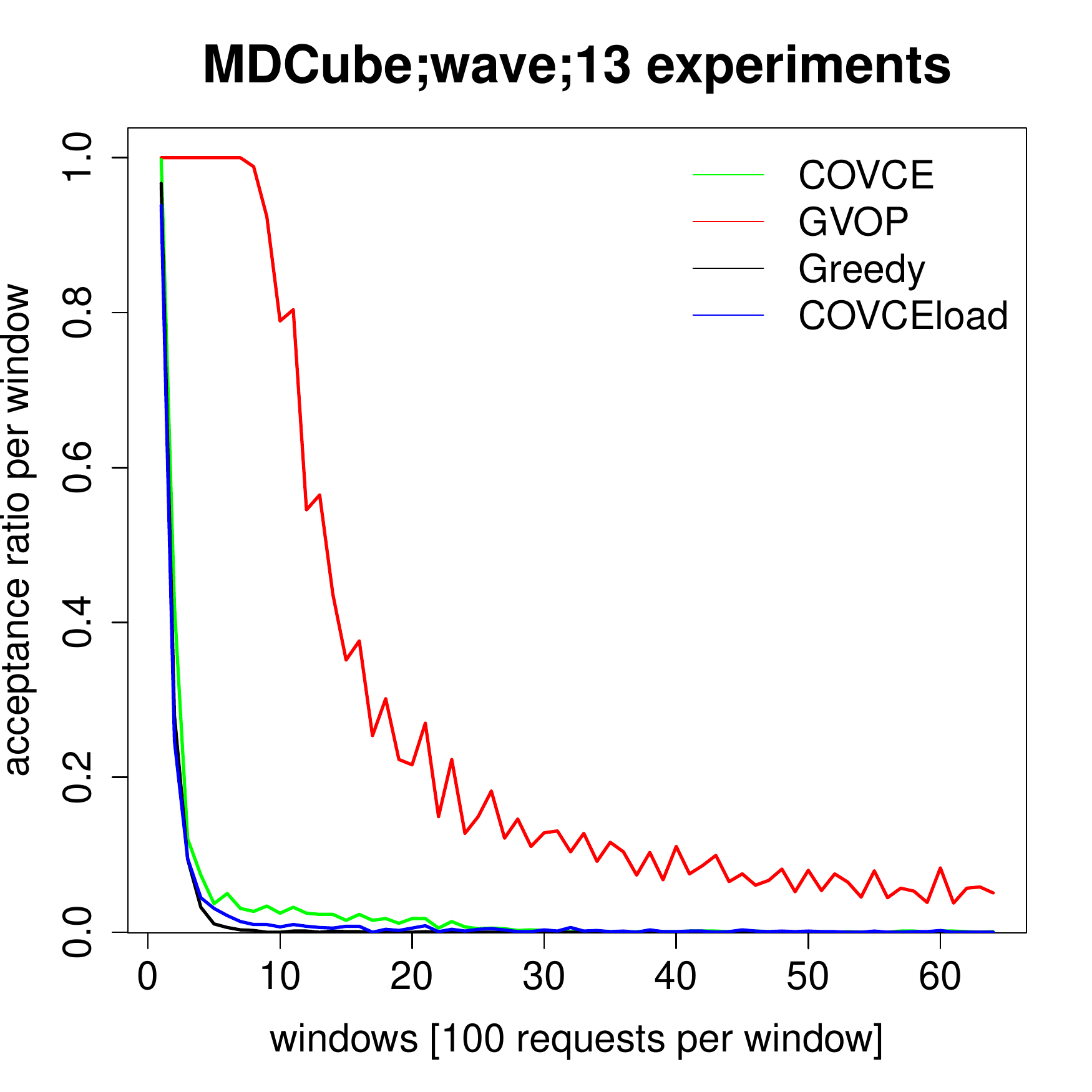}
	\includegraphics[width=0.5\textwidth]{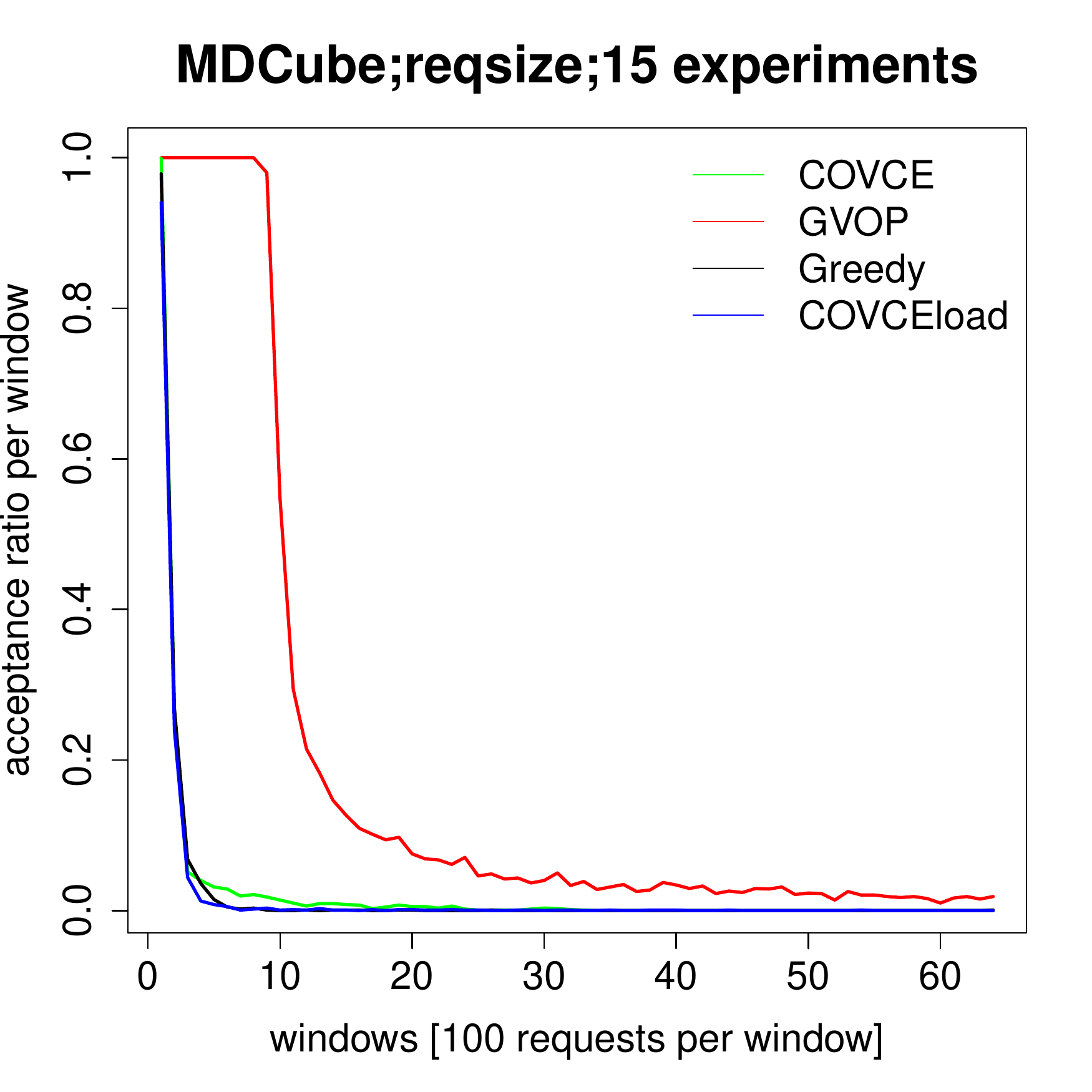}
	\includegraphics[width=0.5\textwidth]{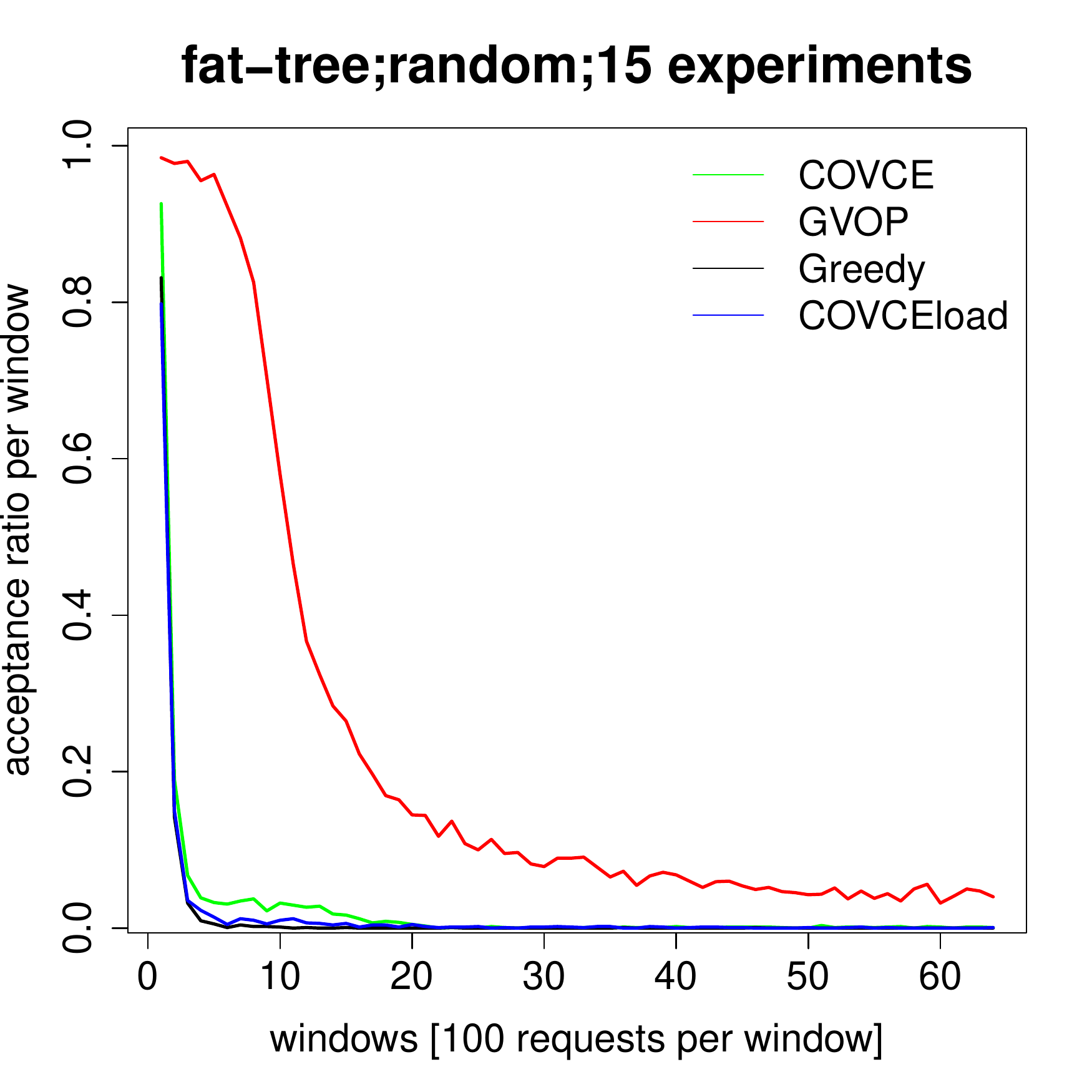}
	\includegraphics[width=0.5\textwidth]{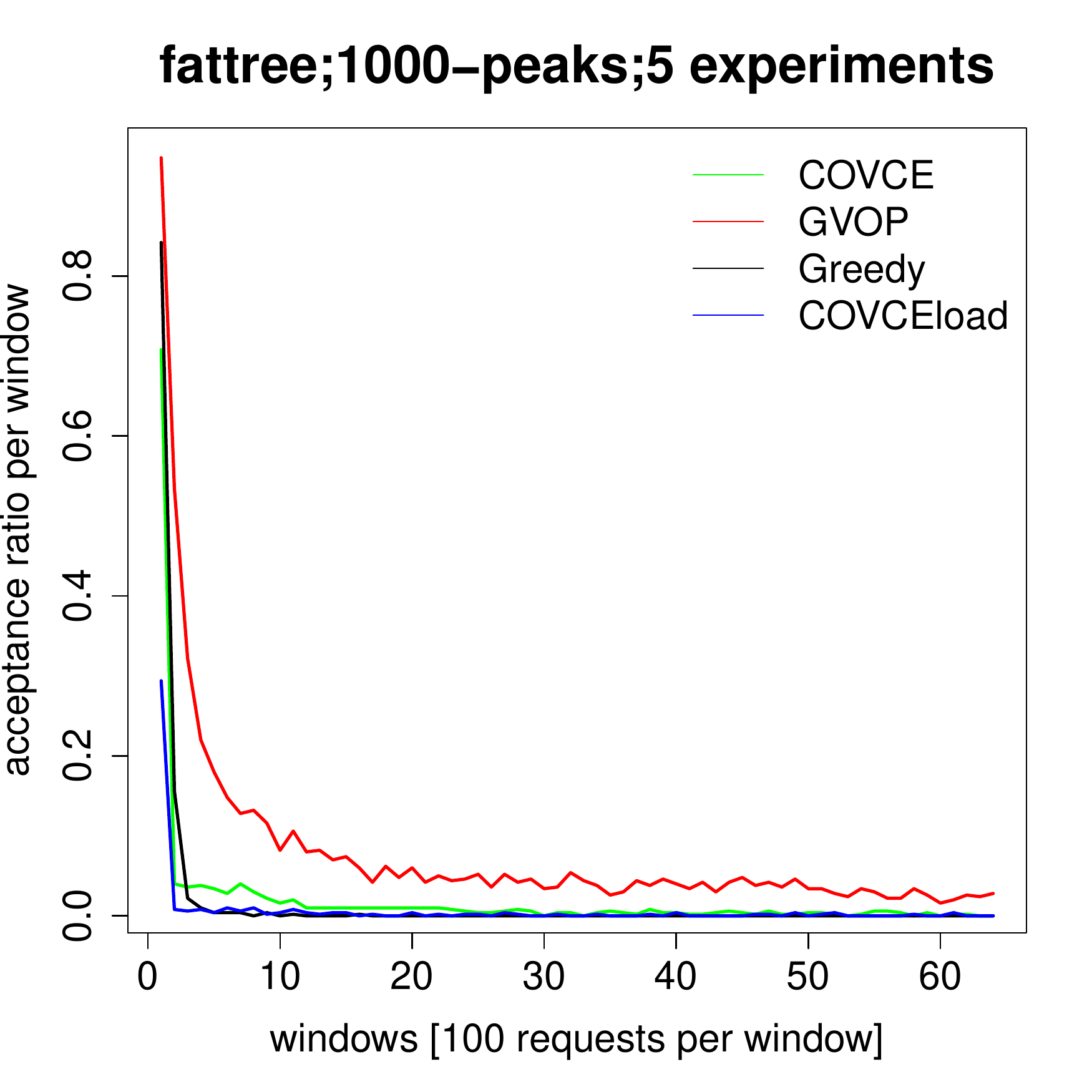}
	\caption{Illustrated are the acceptance ratios for the experiments on MDCube (top) the 12-port fat-tree (bottom) for the 4 studied algorithms and across four different benefit patterns: $\mathtt{wave}$ (top, left), $\mathtt{reqsize}$ (top, right), $\mathtt{random}$ (bottom, left), 1000-$\mathtt{peaks}$ (bottom, right).
}
	\label{fig37}
\end{figure} 
In these sample graphs one can see that the Greedy VC-ACE algorithm stops accepting requests after request 1000 because the algorithm has accepted enough requests without rejecting any that provided valid embeddings. In contrast to this, the competitive algorithms reject requests and additionally violate the capacity constraints, so that the acceptance ratio drops later. The 2-competitiveness of the COVCE algorithm shows up with the similar progress of its acceptance ratio in comparison to the Greedy VC-ACE algorithm. It stops accepting requests after approximately request 2000. Further accepted requests are the ones with a small request size, such that they still fit in the substrate network. Similarly to the COVCE algorithm, the COVCEload algorithm shows small fluctuations even after request 3000 in the $\mathtt{peak}$ scenario because of its rejection functionality. The biggest change in the acceptance ratio of the GVOP algorithm occurs during the first 2000 requests. There are still visible fluctuations after 6400 requests but in all four benefit patterns the acceptance ratio drops below 5\%. In addition, the red curves do not show significant amplitudes after the 4000-th request, such that it is expected that the acceptance ratio converges to zero as for the other three algorithms and that no significant increase will occur. For these reasons, it is assumed that a sequence length of 6400 is decent to extract meaningful information out of the results.

To sum up, these four parameters define a four-tuple configuration. For the experiments, the parameters are chosen as follows: RSL = \{6400\}, NT = \{fat-tree, MDCube\}, TS = \{12-fat-tree,(1,12,1,4)-MDCube\} and BP =\\ \{$\mathtt{1000-peaks,100-peaks,10-peaks,wave,reqsize,vcesize,random}$\}.
There \\are now four algorithms being tested upon seven different BPs and on two different network topologies which means 56 different scenarios. The $\mathtt{random}$ and $\mathtt{reqsize}$ scenarios are repeated 15 times. The $\mathtt{wave}$ BP is 13 times repeated and each $\mathtt{peak}$ scenario five times. The $\mathtt{vcesize}$ BP is repeated 13 times on the fat-tree and 8 times on the MDCube topology.

\subsubsection{Metrics}\label{metrics}
Since the COVCE and GVOP algorithms have their focus on respectively one of the two mentioned competitive ratios, meaning the COVCE algorithm is 2-competitive on the capacity constraints and the GVOP algorithm is 2-competitive on the objective function, a metric is needed to make these algorithms comparable. To achieve this, the \emph{relative profit}($i$) (RP) is introduced. It describes the sum of the benefits of accepted requests until request $r_i$ relatively to the current maximum capacity violation. Formally, let 
\begin{equation}\label{eq34}
l^i(e) = \sum_{r_i \in R} \sum_{m_i^k \in M(r_i)} l(i,k,e) \cdot f(i,k) 
\end{equation}
be the total load on edge $e \in E_S$ after the $i$-th request, $l^i(v)$ is analogously defined for all $v \in V_S$. The maximum capacity violation after the $i$-th request is 
\begin{equation}
\mathtt{violation}(i) = \mathrm{max}\{1,\{\mathrm{max}_{e \in E_S} \{\frac{l^i(e)}{\mathtt{cap}(e)}\}, \mathrm{max}_{v \in V_S} \{\frac{l^i(v)}{\mathtt{cap}(v)}\}\}.
\end{equation}
Then, the relative profit at the sequence index $i$ is
\begin{equation}
	\mathtt{relative \;profit}(i) = \frac{\sum_{r_i \in R | r_i \mathrm{ accepted}}b(r_i)}{\mathtt{violation}(i)}.
\end{equation}
Note that it is possible to compare all four algorithms since the COVCEload and the Greedy VC-ACE algorithms have a constant violation factor of 1. Besides, the violation describes by which factor the capacities among all nodes and edges have to be multiplied in order to avoid over-subscription until request $i$ of this specific request sequence.

With the help of the relative profit, the results of the four algorithms can be further compared and interpreted. More precisely, consider two algorithms $A$ and $B$ achieving the same relative profit with $A$ having a greater violation factor than $B$. So, accomplishing the same relative profit does not mean that both algorithms perform equally well but $B$ performs better than $A$ because $B$ has the same relative profit as $A$ with a less capacity violation and is even better to apply in a practical scenario since the higher the violation the more the capacities have to be increased to overcome the violation factor.

The relative profit is the main parameter which is analyzed for the different BPs because the algorithms differ in their behavior per BP. Therefore, the first parameter being discussed is the relative profit, to show which algorithm has the best benefit to violation ratio for a certain scenario. These results are presented as a line plot with one line per algorithm. Next, the underlying capacity violations are shown to have an impression what the actual absolute values look like behind the competitive ratios. Furthermore, the violation factor contributes to the decision (finding the best algorithm for a certain scenario) because \cite{fattree} show that the relation between the estimated cost for the provider for a maximum possible number of hosts in a substrate network increases superlinear. For this reason, the cost for extra capacities due to a certain capacity violation factor do not increase linear with the extra profit gained by this violation factor. So, deciding whether a capacity violation by the competitive online algorithms is worth the superlinear increasing cost, depends on the percentage of the increased relative profit, the violation factor and the superlinear increasing costs.

In addition to the relative profit, the \emph{average relative profit} (ARP) is introduced. 
The relative profit function involves the maximum violation over all nodes and edges until request $i$. 
Another perspective on the capacity violation is provided with the \emph{average} violation. It describes the average capacity oversubscription over all nodes and edges until request $i$ and will induce less capacity multiplication costs for the provider since a single node or edge with the maximum oversubscription does not influence the average violation factor as much as the violation factor. The average violation is formally defined as
\begin{equation}\label{eq35}
\mathtt{violation}_{\mathrm{avg}}(i) = \mathrm{max}\bigg\{1, \frac{1}{|E_S|}\cdot \sum_{e \in E_S} \frac{l^i(e)}{\mathtt{cap}(e)} + \frac{1}{|V_S|}\cdot\sum_{v_\in V_S}\frac{l^i(v)}{\mathtt{cap}(v)}\bigg\}
\end{equation}
using the definitions of $l^i(e)$ and $l^i(v)$ from Equation \ref{eq34}. At each request $r_i$, the average violation describes the sum of the loads on all nodes and edges relatively to the maximum feasible allocation. Consequently the average relative profit is defined as
\begin{equation}\label{eq36}
\mathtt{relative \; profit}_{\mathrm{avg}}(i) = \frac{\sum_{r_i \in R | r_i \mathrm{ accepted}}b(r_i)}{\mathtt{violation}_{\mathrm{avg}}(i)}.
\end{equation}

Finally, the acceptance ratio over the sequence of requests is illustrated in windows of 100 requests, as explained before, in order to see the influence of the competitive ratios on the ``admission control behavior'' of the algorithms. In addition, the acceptance ratio gives an overview of how many resources are still available in the substrate network (with the respectively added capacities) in order to state whether an algorithm's profit will still experience significant changes after the 6400th request. For a more detailed view on the acceptance patterns, a bar plot is shown in order to analyze the ``local'' behavior of the algorithms.

Upon these metrics, we discuss which algorithm is the most appropriate for a certain scenario. In doing so, we consider the trade-off between the relative profit in relation to the violation against the background that capacity multiplication costs increase superlinear with respect to the violation factor. The same procedure is executed for the average violation and the ARP. For instance, if a competitive online algorithm has a 1\% higher relative profit than the one of a non-competitive online algorithms, which means that the profit increases superlinear with respect to the violation factor, then this does not necessarily mean that the relative profit is worth the superlinear increasing cost when the actual violation factor is 100. The point is that the percentage of higher relative profit has to suit the extent of the capacity violation.

Since for each BP type multiple experiments are run, the relative profits of the results of the same BP type are averaged. For the violation and the acceptance ratio line plots, the results for the same BP type are averaged, too. 

\subsection{Results}
This section is divided into several parts. At first, the results on the before mentioned benefit patterns are shown on the fat-tree topology applying the RP metric. Afterwards, the results on the fat-tree are discussed again applying the ARP metric to give another view over the experiments. Next, another topology, the MDCube is applied as network topology. Finally, a results are summarized and discussed. 

Before going into the analysis of the several benefit patterns, observe the following introductory example in Figure \ref{fig2}.
\begin{figure}
	\centering
	\includegraphics[width=1.0\textwidth]{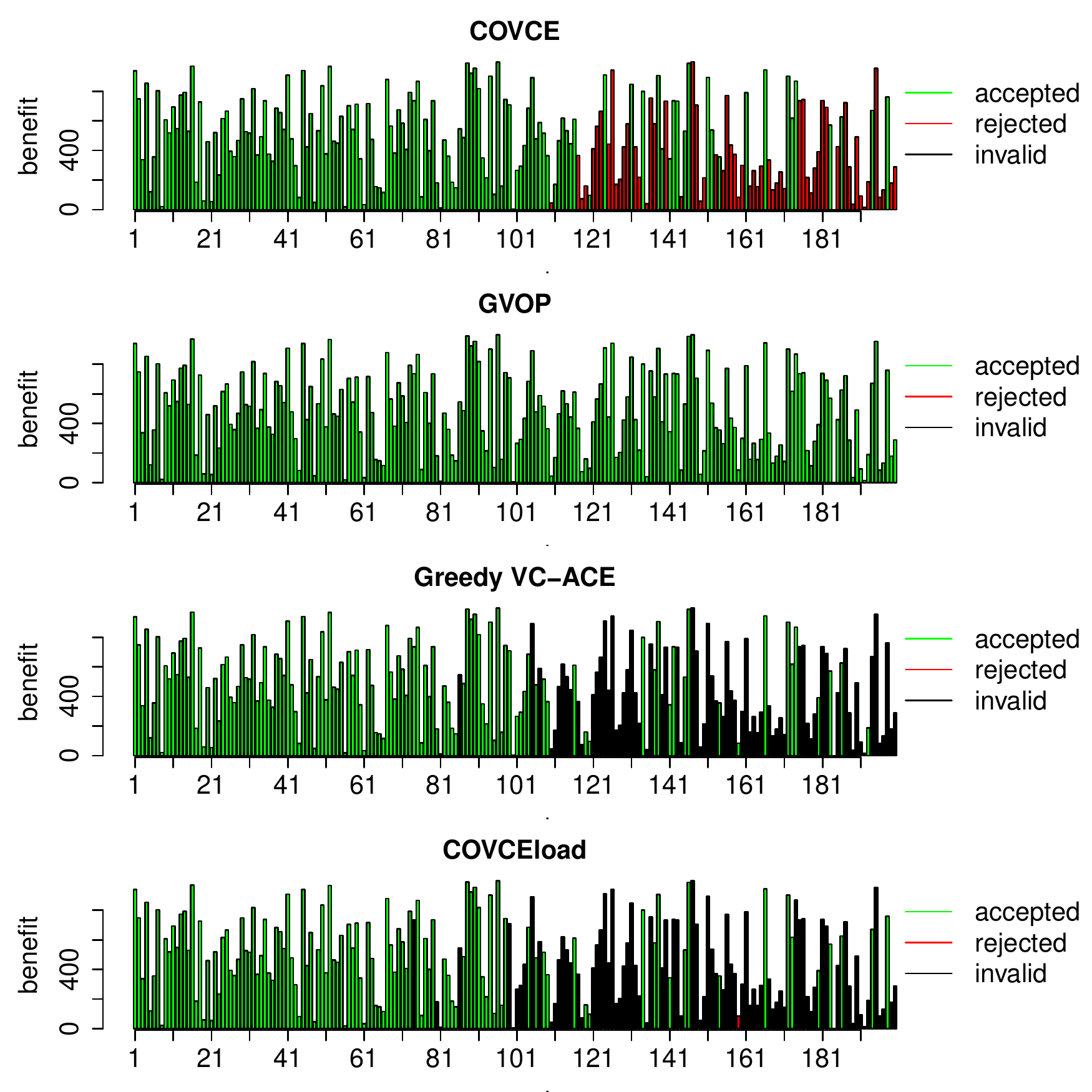}
	\caption{Depicted are the benefits as bars over the first 200 of 6400 requests of a single experiment on the MDCube topology for the $\mathtt{random}$ benefit pattern. The colors show whether the corresponding request is accepted (green), rejected (red) or invalid (black).}
	\label{fig2}
\end{figure}
It refers to a single experiment of the $\mathtt{random}$ BP where the benefit interval is chosen to be $b(r_i) \in [1,1000]$ to give an impression of how the typical result looks like. It shows the outcome of the four algorithms over the first 200 of 6400 requests on an (1,12,1,4)-MDCube. The height of the bars stands for the benefit of the current request. The color depicts whether the algorithm accepted, rejected or declared the request as invalid.

In general, all four algorithms accept the first 70 requests because the substrate network has enough resources for about 70 requests of the before mentioned sizes. In addition, the benefit of the requests is greater than their corresponding minimum-cost embeddings' costs. After the 70th request the behavior of the algorithms starts to differ.

Considering the result of the COVCE algorithm in Figure \ref{fig2}, the first request being rejected is $r_{110}$ because its benefit is lower than its minimum-cost embedding's cost. There are previous requests having a similar benefit and were accepted but the costs of the substrate nodes and edges are now higher, such that it is less likely for a request with a similar demand and benefit to be accepted again. So, the output of the if-condition in the COVCE algorithm is false and consequently rejects $r_{110}$. One can already recognize a pattern for the following requests, more precisely, that requests with a relatively low benefit are being rejected. Thus, one can expect that later on the "peaks" of this plot will be colored green and all other requests will be colored red. 

On the other hand, the GVOP algorithm does not manifest any rejections yet. Another point is that the Greedy VC-ACE algorithm declares many requests as invalid after request 110 which means that requests with a relatively high size do not fit anymore into the substrate network. In fact, this finding contributes to the starting point of rejections at request 110 of the COVCE algorithm because it does not violate any resource's capacity by more than a factor of 2. The COVCEload algorithm starts to declare some requests as invalid earlier than the Greedy VC-ACE algorithm. 

After about 780 requests, the GVOP algorithm starts rejecting requests with a very low benefit as seen in Figure \ref{fig3}. In contrast to this, the COVCE algorithm rarely accepts a request at this state of the execution which shows its competitiveness to not violate the resource capacities any further than a factor of 2. The ones that are accepted (green-marked) have a higher benefit to embedding cost ratio than the ones that are rejected. On the other side, the Greedy VC-ACE algorithm does not accept any further request in the shown sequence because after greedily accepting all requests with valid minimum-cost embeddings, relatively small sized requests do not even fit anymore into the substrate network. Thus, if the Greedy VC-ACE algorithm rejected the minimum-sized request $r_{i,\text{min}}$ with $\mathcal{N}_{i,\text{min}}$ = 3, $\mathcal{B}_{i,\text{min}}$ = 4 and $\mathcal{C}_{i,\text{min}}$ = 4, the algorithm will not accept any further request. Last but not least, the COVCEload algorithm still accept a few requests since it does not purely embed requests greedily but has a rejection functionality, which also shows up for some requests (red-marked).
\begin{figure}
	\centering
	\includegraphics[width=1.0\textwidth]{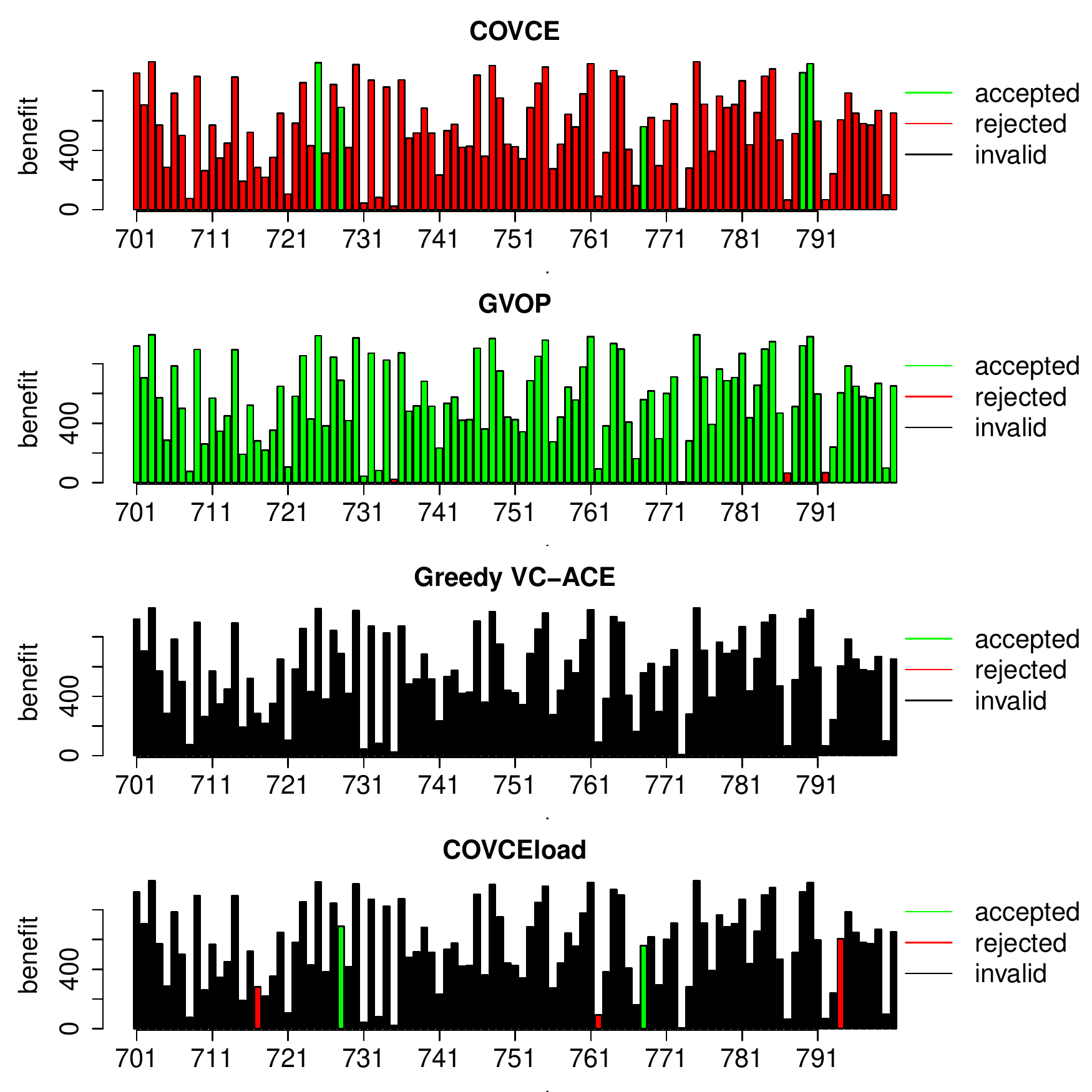}
	\caption{Depicted are the benefits as bars over the sub-sequence of the requests 700-800 of a single experiment on the MDCube topology. The $\mathtt{random}$ benefit pattern is applied for the four algorithms. The colors show whether the corresponding request is accepted (green), rejected (red) or invalid (black).}
	\label{fig3}
\end{figure}
\subsubsection{Results on Fat-Tree Topology}
To begin with, the results of the seven benefit patterns on the fat-tree topology are presented upon the relative profit metric and afterwards the results are reviewed with the ARP metric.
\paragraph{Random BP}
To recapitulate, the $\mathtt{random}$ BP consists of randomly chosen integer benefits out of a certain interval. For the following experiments, this interval is chosen to be $b(r_i) \in [1,1000]$. At first, the acceptance ratio of the algorithms is shown for the $\mathtt{random}$ BP in Figure \ref{fig4} (left). For this graph, the results of 15 experiments are averaged. On the x-axis the indexes of the windows are listed and on the y-axis the acceptance ratio per window is illustrated. To recapitulate, the acceptance ratio per window is the sum of all accepted requests in that window divided by the width of that window. The windows have a width of 100 requests, so that there are 64 windows for a certain request sequence.
\begin{figure}[htpb]
	\includegraphics[width=0.5\textwidth]{graphics/acceptance_ratio_window_fattree_random.pdf}
	\includegraphics[width=0.5\textwidth]{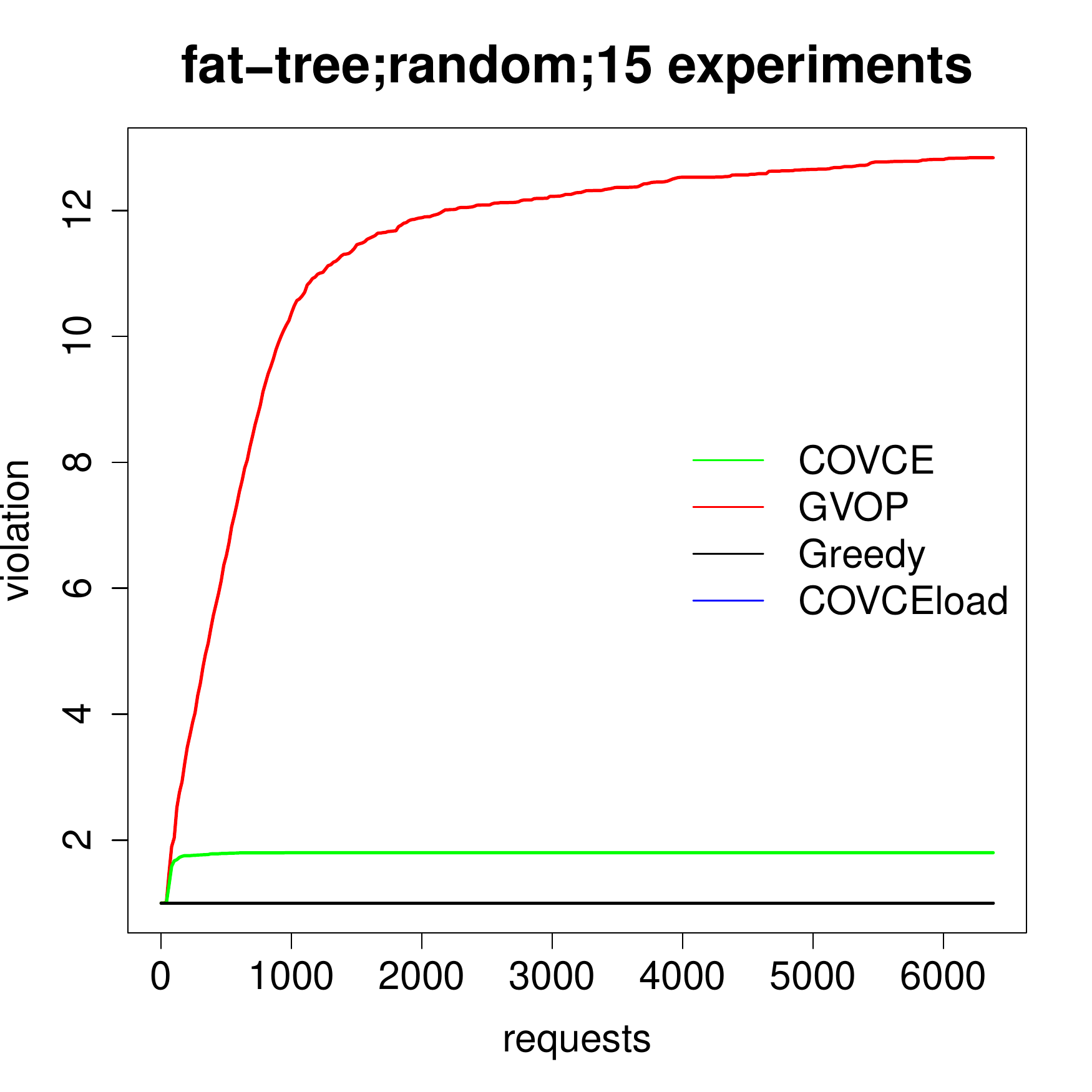}
	\caption{Depicted are the acceptance ratio (left) and violation (right) for $\mathtt{random}$ BP over 6400 requests on the 12-fat-tree topology. The results of the 15 experiments over the four algorithms are averaged for both metrics.}
	\label{fig4}
\end{figure} 
One can observe that the gradient of the greedy, COVCEload and COVCE algorithms are similar to each other but the GVOP algorithm's gradient noticeably differs from them. The curve of the GVOP algorithm drops later and slower than the others because it allows a capacity violation of approximately factor 12 after 6400 requests, as Figure \ref{fig4} (right) shows. The fluctuations during the progress of the curves in Figure \ref{fig4}(left) are due to the randomness of the request sizes and benefits since small-sized requests with a higher benefit are more likely to be accepted than big-sized requests with a lower benefit. The GVOP algorithm generally accepts more requests than the other algorithms because it allows for a higher capacity violation.

Accordingly to the difference in the acceptance ratios between the COVCE and the GVOP algorithm, the latter shows also a similar distance to the other algorithms considering the relative profit in Figure \ref{fig6} (left).
\begin{figure}[htpb]
	\includegraphics[width=0.5\textwidth]{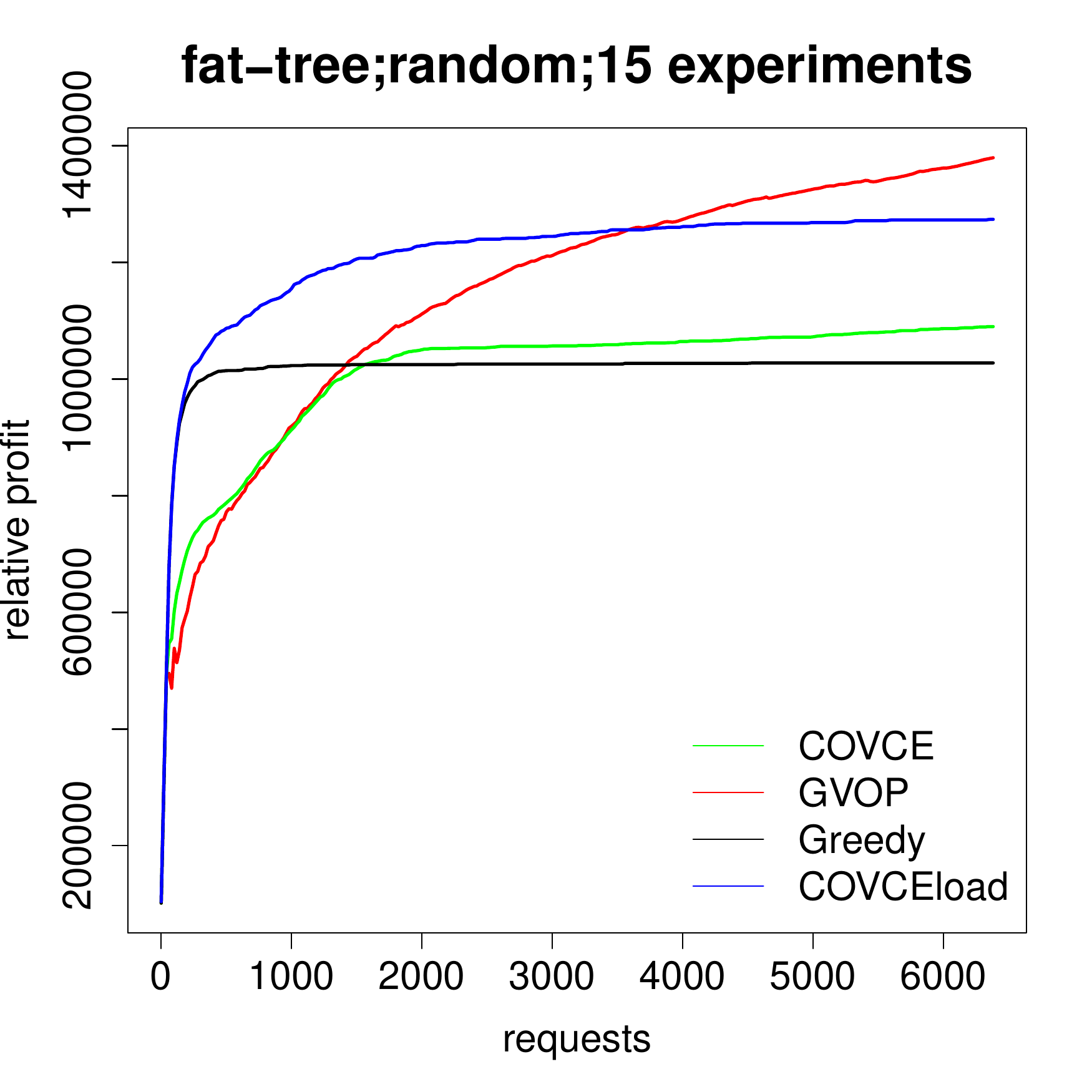}
	\includegraphics[width=0.5\textwidth]{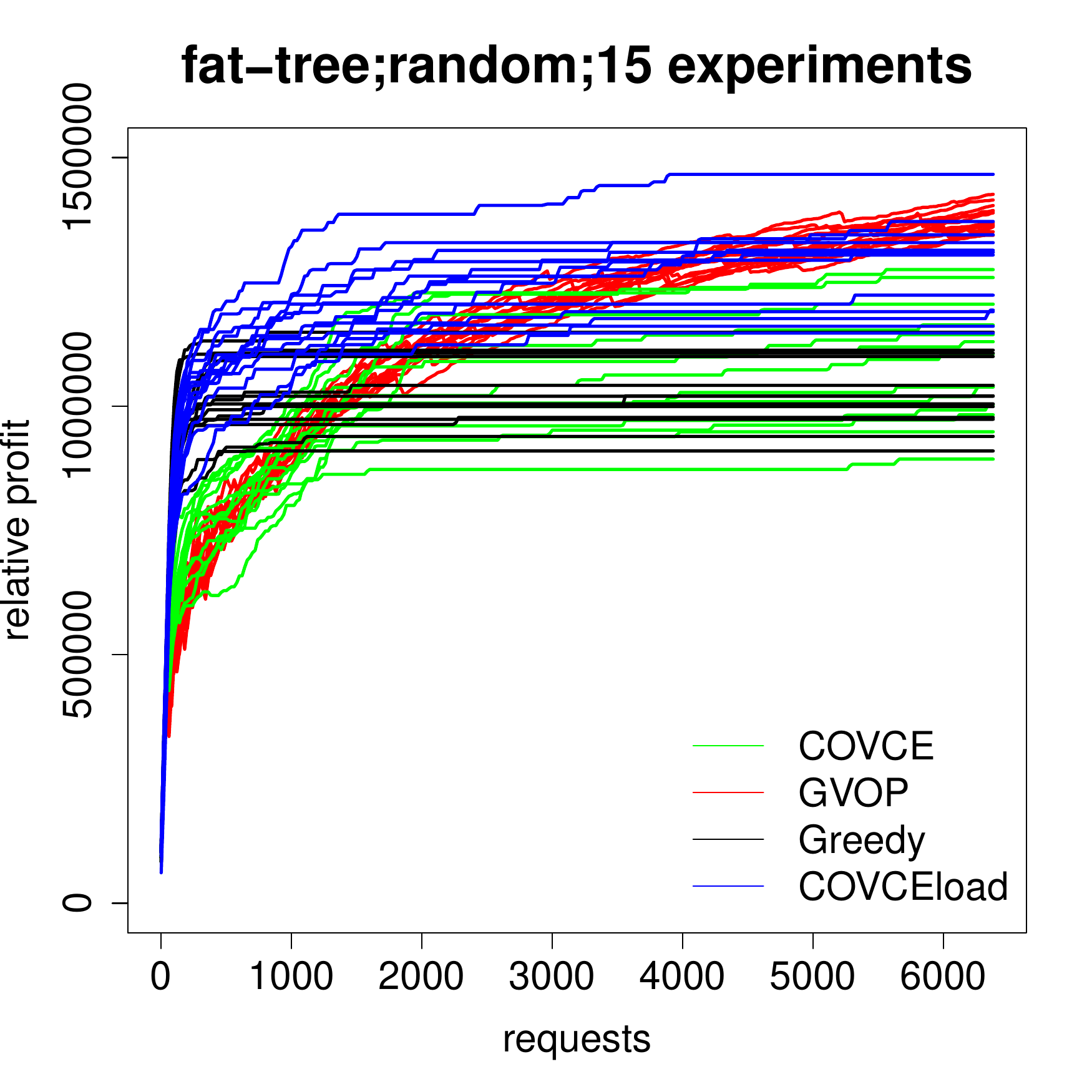}
	\caption{Depicted are relative profits of the four algorithms over 6400 requests for the $\mathtt{random}$ BP on the 12-fat-tree topology. The results of the 15 experiments are averaged (left) and respectively shown (right).}
	\label{fig6}
\end{figure} 
In this plot, one can see how the correlation between the competitiveness on the objective function and the capacity constraints works in terms of the relative profit.
The relative profits of the 15 experiments are averaged. From an overall view, the GVOP algorithm has a higher relative profit than the other algorithms after overtaking the COVCEload algorithm at approximately request 3700. Before request 3700, the COVCEload algorithm has the highest relative profit. The first violations occur in between the first 100 requests which are noticeable by the zig-zag progression of the GVOP algorithm. However, the closer the algorithm progresses to its maximum capacity violation per node and edge, the lesser the probability that big-sized and low-beneficial requests are accepted because the cost functions of each node and edge increase per further allocation on them and therefore the deciding if-condition in the algorithms is more likely to be false. Regarding the COVCE algorithm, it has a slightly higher relative profit than the Greedy-VC-ACE algorithm after the latter stops accepting requests because its capacities are exhausted. As soon as the capacity violation is getting closer to 2, the COVCE algorithm accepts less requests with a poor benefit and a high request size in contrast to the Greedy VC-ACE algorithm which accepted any request with a valid embedding. Thus, the subset of requests accepted by the COVCE algorithm with a high benefit to embedding size ratio is greater than the one for the greedy algorithm which results in a higher benefit to violation ratio (the relative profit). For the same reason, the GVOP algorithm's relative profit is later on higher than the greedy algorithm's one. The COVCEload algorithm has the second best result. The concept of combining residual capacities for the VC-ACE algorithm with the COVCE cost functions and rejection functionality shows a better performance than its related greedy and its related competitive algorithm. As a result, the GVOP algorithm performs the best in this scenario when only considering the relative profit. 

When comparing the competitive online algorithms to each other, the GVOP algorithm has a higher relative profit. One can interpret that the combination of its objective function and its capacity violation competitiveness performs better than the combination of the competitive ratios of the COVCE algorithm. Nevertheless, the sequence of requests is not long enough to know how the relative profit of the GVOP algorithm behaves in the future because the theoretical maximum violation is not reached yet since the curve is still increasing.

Without regarding the violation factors, the competitive online algorithms are outperforming the greedy non-competitive one considering that the COVCE has an approximately 6.8\% and the GVOP algorithm an approximately 37\% higher relative profit than the Greedy VC-ACE algorithm after 6400 requests. However, when considering the first 1000 requests in which the Greedy VC-ACE algorithm still accepted requests, it outperforms the competitive algorithms which suffer in performance because of resource violations. On top of that, the COVCEload algorithm has the best performance until request 3700 and a 24\% higher value than the Greedy VC-ACE algorithm. As a result, violating capacity constraints only results in a better performance after approximately request 3700 of the request sequence since then requests with a relatively high benefit to embedding cost ratio are only accepted. 

When involving the violation factors of the competitive online algorithms, as shown in Figure \ref{fig4} (right), one gets an overview on how the violations look like throughout the progress of the sequence of requests instead of only mentioning competitive ratios.
The results of 15 experiments are averaged. Note that the greedy VC-ACE and the COVCEload do not violate resource (violation factor = 1). The gradients for the competitive algorithms are the highest in the first 1000 requests. So, the violation factor increases relatively fast which explains the performance drop of the COVCE and GVOP algorithms' relative profit approximately in between the 300th to 2000th request. For the COVCE algorithm, the gradient of its curve radically drops in between the first 100 requests, but for the GVOP algorithm this behavior stretches over an interval of 1000 requests (between request 1000 and 2000). Due to this behavior, one can expect that the violation of the GVOP algorithm will converge to approximately 14. To show the falseness of this estimation, the competitive ratio can also be manually calculated upon the formula 
\begin{equation} \label{eq30}
	log_2(1+3\cdot \mathrm{max}_{i,k}\{\sum_{e \in \mathcal{M}_{E,i}^k} l(i,k,e) + \sum_{v \in \mathcal{M}_{V,i}^k} l(i,k,v)\} \cdot \mathrm{max}_i \{b_i\})
\end{equation}
which is taken from the competitive ratio of \cite{gvop} and interpreted with the notation of Section \ref{theory}. At first, the maximum benefit for the $\mathtt{random}$ BP is 1000. The term in the first maximization function describes the sum of the allocations over all node and edges that are included in embedding $k$ of request $i$. In this case, if an embedding involves all 1296 edges and 432 nodes of the 12-port fat-tree (each having a capacity of 20), then the maximum allocation is 34560. This results in 26.62, so that no resource capacity is violated more than a factor of 27. In conclusion, the estimation of a capacity violation of 14 is far away from the theoretical result. 

Consider the state at request 6400 for the relative profit and the capacity violation. When comparing them, the COVCE algorithm has a 6.8\% higher relative profit than the Greedy VC-ACE algorithm for violating the resource by a factor of at most 2. Thus, from a provider's perspective it may not even pay off to buy additional resources for a slightly higher relative profit. Accordingly, the GVOP algorithm achieves an approximately 37\% higher relative profit for a violation of 12.8 at this state of the execution. Again, it may not pay off for the provider to a have 37\% more relative profit when multiplying the capacities by 12.8. When comparing the competitive online algorithms, there is a factor of approximately 6 between their capacity violation at request 6400. It is again arguable whether 25.6\% higher relative profit are worth a 6 times higher capacity violation. Whether capacity violations are affordable or not, the COVCEload algorithm is at least the second best opportunity in this scenario and only has an approximately 8\% less relative profit than the GVOP algorithm with no capacity violation. 
Nevertheless, the provider has different opportunities dependent on, e.g., her financial situation. However, it is unknown how the relative profit of the GVOP algorithm looks like in the future since the capacity violation is far from the theoretical maximum as calculated before. 
On top of that, when assuiming superlinear increasing costs for the multiplication of capacities, the COVCEload algorithm pays off more than the GVOP algorithm, considering
a theoretical capacity violation and therefore multiplication of 27 for GVOP algorithm and no capacity violation for the COVCEload algorithm.

Another view on the relative profit is given in Figure \ref{fig6} (right) in order to help with this decision. In this graph, the relative profit of the 15 experiments for each algorithm are shown respectively. The first thing to notice is that the COVCEload, the COVCE and the greedy algorithm have a similar behavior and that the GVOP algorithm differs from them. The relative profits of the latter continue as a bundle, none of them is clearly separated. In contrast to this, the relative profits of the other three algorithms also start as bundles but then are more separated of each other after approximately 100 for the greedy and COVCEload and 1000 requests for the COVCE algorithm. In other words, the variance of them is higher than the one of the GVOP algorithm as soon as their maximum theoretical violation for each node and edge is approximated. Thus, the GVOP algorithm's curves still continue as a bundle because it does not reach its theoretical maximum violation of 27 in this request sequence. Another interesting point is that the before mentioned 8\% higher relative profit of the GVOP algorithm in Figure \ref{fig6} (left) do not show up to be consistent over all experiments as seen in Figure \ref{fig6} (right). In fact, there is at least one experiment of the COVCEload algorithm that outperforms the GVOP algorithm. 

So, a capacity violation of factor 2 or 12.8 (and theoretically 27) together with a superlinear increasing capacity multiplication cost and a not always higher relative profit of the GVOP algorithm does not pay off compared to the COVCEload algorithm in this scenario. Overall, the COVCEload algorithm's performance is the best in this scenario.

\paragraph{Benefit proportional to the Size of the Request}
As in the latter section, the analysis is the same for the $\mathtt{reqsize}$ BP. The benefit of this pattern is computed dependent on the request size. This benefit function is motivated by the fact that the tenant usually pays for the amount of resources she reserves \cite{reqsize}. Thus, for a certain request $r_i$, its benefit is determined upon the formula
\begin{equation}
	b(r_i) = \mathcal{N}_i \cdot ( \mathcal{B}_i +  \mathcal{C}_i).
\end{equation}
The product $\mathcal{N}_i \cdot \mathcal{C}_i$ captures that per VM $k \in V_{VC} \setminus \{LS\}$, $\mathcal{C}_i$ compute units have to be allocated. Additionally, if the request cannot be served on a single physical machine, further edges have to be used to connect several physical machine for the embedding. Thus, the bandwidth $\mathcal{B}_i$ might be allocated several times per request which is captured by the product $\mathcal{N}_i \cdot \mathcal{B}_i$ in the formula. Nevertheless the latter product does not include that per additional VM $k$ multiple bandwidth allocations $\mathcal{B}_i$ might be needed. So, the expressiveness of this benefit function is limited. Further detailed analysis on this kind of benefit function are shifted to future work. 

To get a first impression of this scenario, consider Figure \ref{fig7} (left) in which the average acceptance ratio of 15 experiments is shown.
\begin{figure}
	\includegraphics[width=0.5\textwidth]{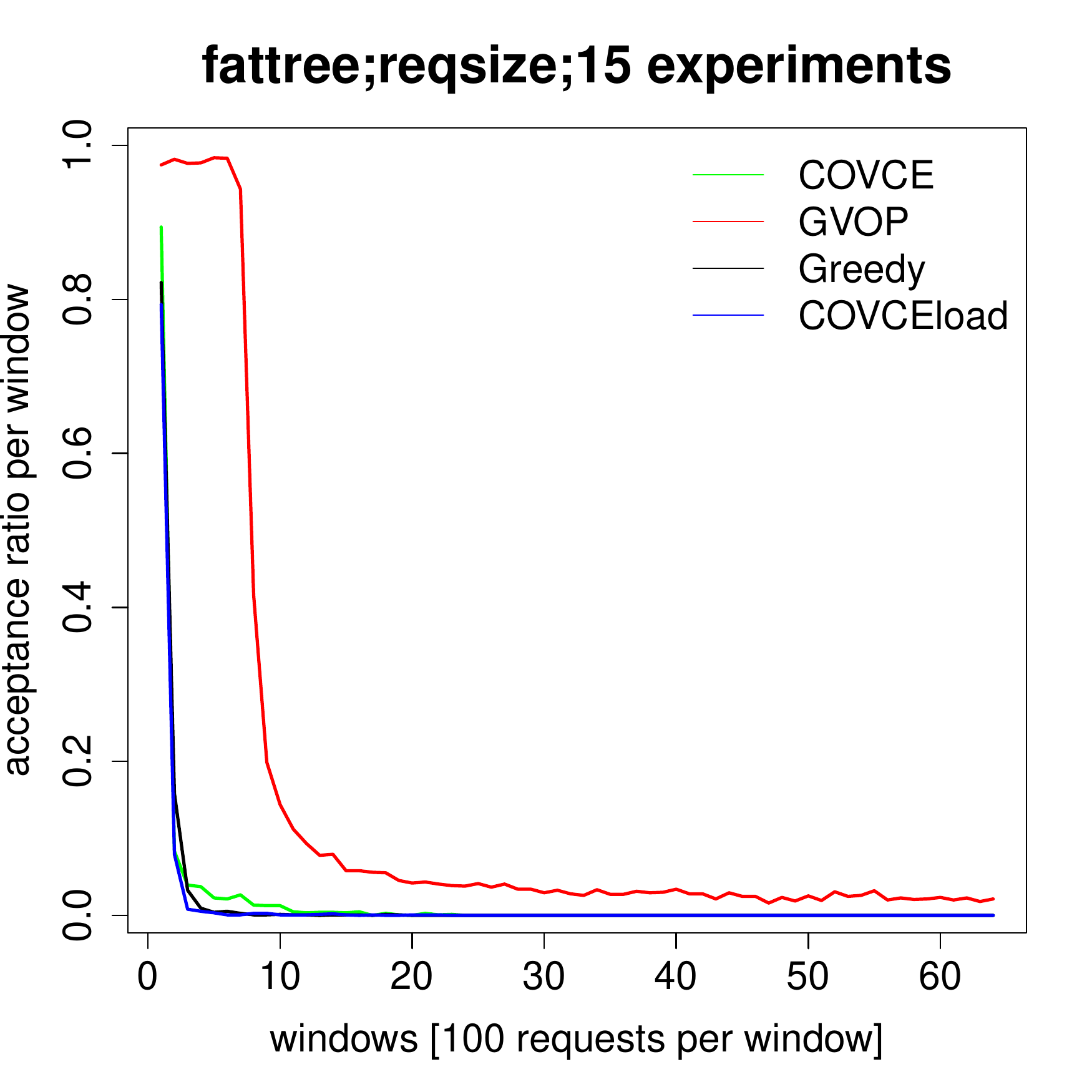}
	\includegraphics[width=0.5\textwidth]{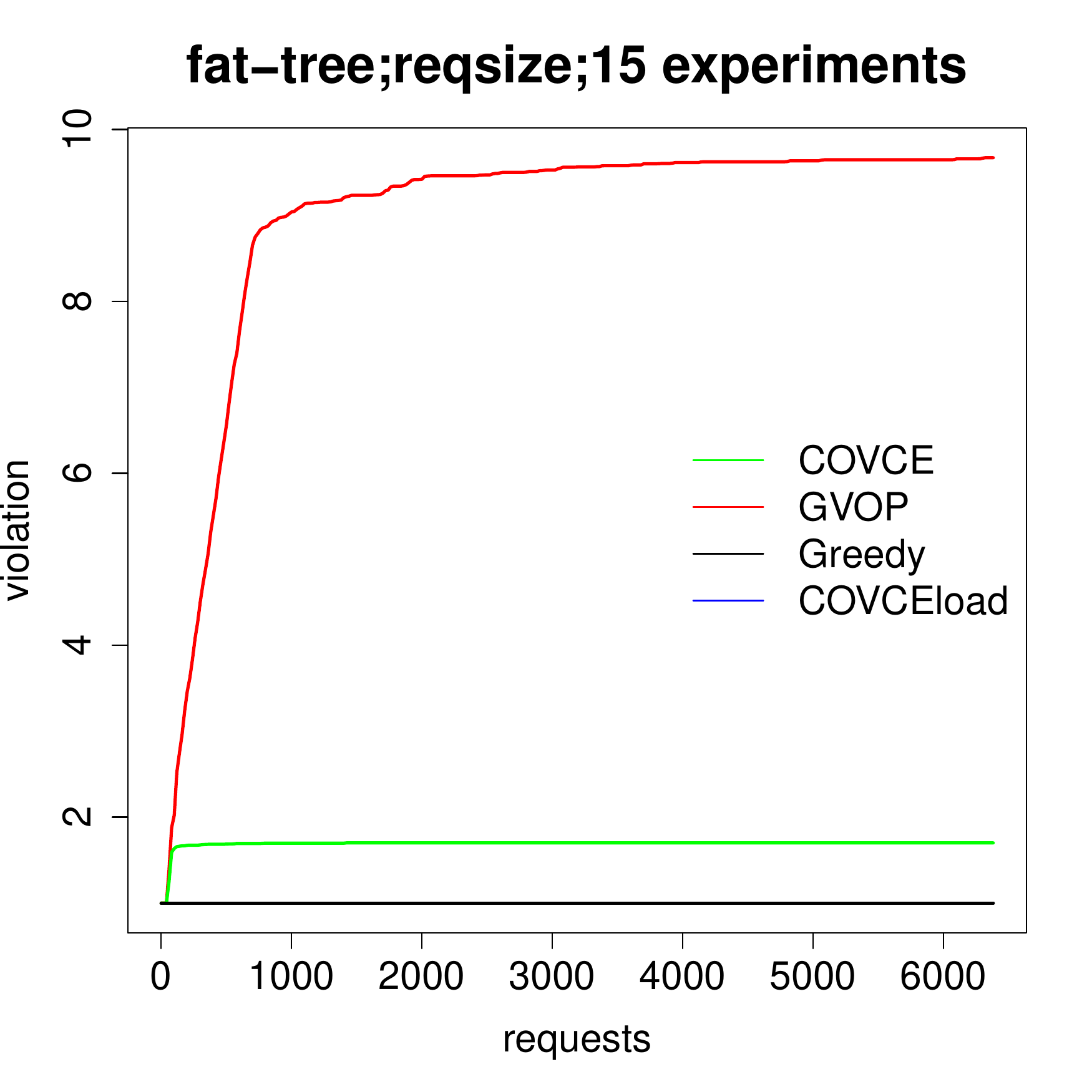}
	\caption{Depicted are the acceptance ratio (left) and violation (right) for $\mathtt{reqsize}$ BP over 6400 requests on the 12-fat-tree topology. The results of the 15 experiments over the four algorithms are averaged for both metrics.}
	\label{fig7}
\end{figure}
There is a radical drop which is even higher than the drop in the $\mathtt{random}$ BP scenario. A possible reason for this behavior is that the maximum and the minimum of the co-domain of the benefit function for this BP are closer to each other than in the co-domain of the $\mathtt{random}$ BP. The maximum benefit is $b_{\mathrm{max}} = 14 \cdot (12 + 20) = 448$ (12 and 20 being the maximum bandwidth and compute units) and the minimum benefit is $b_{\mathrm{min}} = 3 \cdot (4 + 4) = 24$, so the factor between the benefits can be maximally 19. This is only 1.9\% of the factor in the $\mathtt{random}$ scenario and therefore might influence the performance of the competitive online algorithms since there are less requests with an outstanding benefit to embedding cost ratio. In comparison to the $\mathtt{random}$ BP, it is therefore less likely that further requests are accepted after the first violation. Thus, a high drop rate of the acceptance ratio is induced for both competitive algorithms. 

The lower difference between the acceptance ratios in the $\mathtt{reqsize}$ BP also influences the relative profit of the competitive online algorithms. The results of the relative profit of all 15 experiments are averaged and shown in Figure \ref{fig24} (left) and respectively shown in Figure \ref{fig24} (right).
\begin{figure}
	\includegraphics[width=0.5\textwidth]{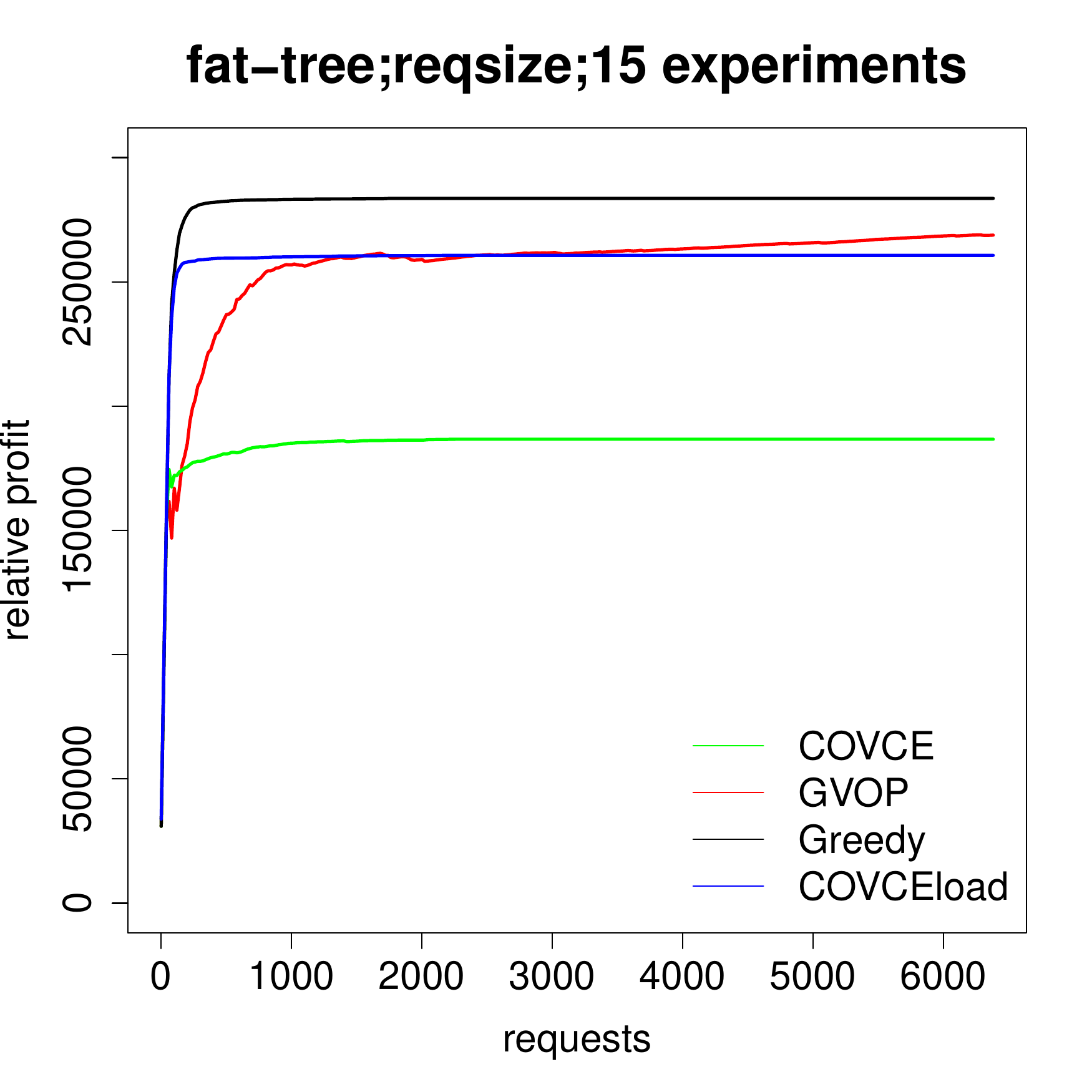}
	\includegraphics[width=0.5\textwidth]{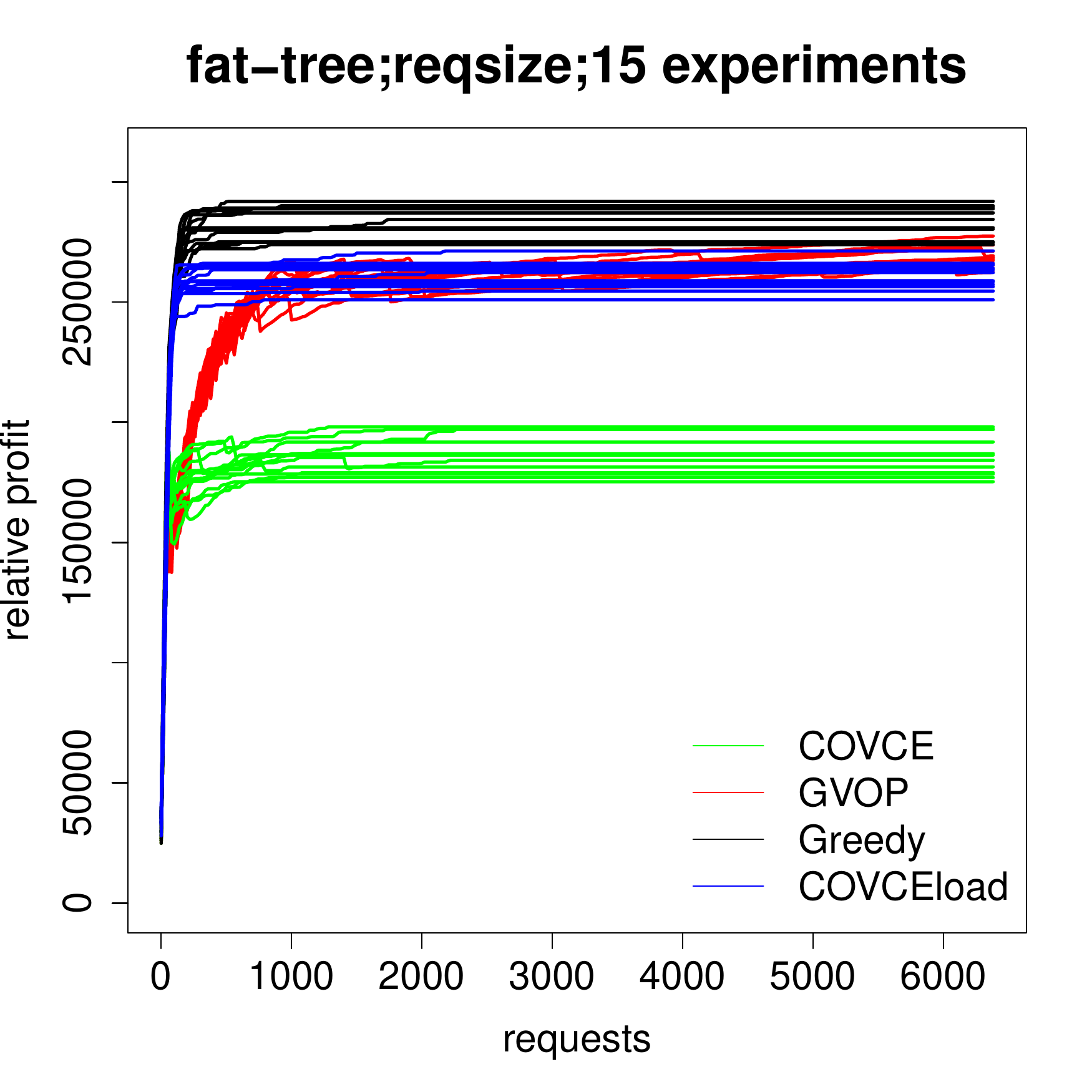}
	\caption{Depicted are relative profits of the four algorithms over 6400 requests for the $\mathtt{reqsize}$ BP on the 12-fat-tree topology. The results of the 15 experiments are averaged (left) and respectively shown (right).}
	\label{fig24}
\end{figure}
It is clearly noticeable that the non-competitive online algorithms outperform the competitive ones for the shown sequence of requests. Though, it might be possible that the GVOP algorithm outperforms the Greedy VC-ACE algorithm dependent on whether accepting further requests increases the average benefit to violation ratio or not.

The competitive online algorithms are not beneficial in this specific scenario because of the small co-domain of the benefit function. On top of that, outstanding benefit peaks will not be advantageous too because a high benefit in this scenario stands for a high request size inducing a high resource footprint and therefore high costs for the corresponding embedding.

Accordingly, the capacity violation is less than in the $\mathtt{random}$ scenario because the incentive for embedding more requests (in terms of the benefit to embedding cost ratio) is smaller. As shown in Figure \ref{fig7} (right), the violation of the COVCE algorithm is slightly less with 1.6 and the one of the GVOP algorithm dropped down to approximately 9.75 at request 6400. In fact, the competitive ratio of the GVOP algorithm depends on the maximum possible violation as shown in Equation \ref{eq30} and calculates to 25.46 assuming the maximum theoretical allocation of an embedding involves the whole substrate graph. Even if the violations are lower than in the $\mathtt{random}$ BP scenario, they do not pay off in comparison to the non-competitive online algorithm because they do not induce a higher relative profit. Furthermore, the Greedy VC-ACE algorithm has a higher relative profit than the COVCEload algorithm, so that the additional rejection functionality of the COVCEload algorithm does not pay off for small benefit to request size ratios in this scenario.

In conclusion, the competitive online algorithms show weaknesses when the benefit increases proportionally to the request size and show worse results on the relative profit than the non-competitive algorithms in this specific scenario. The Greedy VC-ACE algorithm has the best performance regarding the relative profit for this BP. In the next scenario, another similar motivated benefit pattern is discussed to further investigate their properties.

\paragraph{Benefits proportional to the Size of the Embedding}
In contrast to the previous BP whose benefits are calculated dependent on the request size, the following $\mathtt{vcesize}$ benefit pattern considers the size of the corresponding cost-minimum embedding. To recapitulate, the benefit is for this BP is defined accordingly to Equation \ref{eq59}:
\begin{align}
b(r_i) = \frac{|E_S|\cdot C_E}{|V_S| \cdot C_V}\cdot \bigg(\frac{\mathcal{N}_i \cdot \mathcal{C}_i}{|V_S|\cdot C_V}\bigg)^2 + \bigg(\frac{\mathrm{empirical\_cost}_E (\mathcal{N}_i, \mathcal{B}_i,\mathcal{C}_i)}{|E_S|\cdot C_E}\bigg)^2. \nonumber
\end{align}
After calculating the value of $b(r_i)$, it is scaled up in order to handle it as an integer value. Scaling up does not influence the result because the rate between the requests' benefits do not change. The lowest benefit value is 1 and the highest benefit value is approximately 1000, so that a factor 1000 is given between the lowest and highest benefit as in the $\mathtt{random}$ scenario.

To begin with, consider the results for the relative profit in Figure \ref{fig5}.
\begin{figure}
	\includegraphics[width=0.5\textwidth]{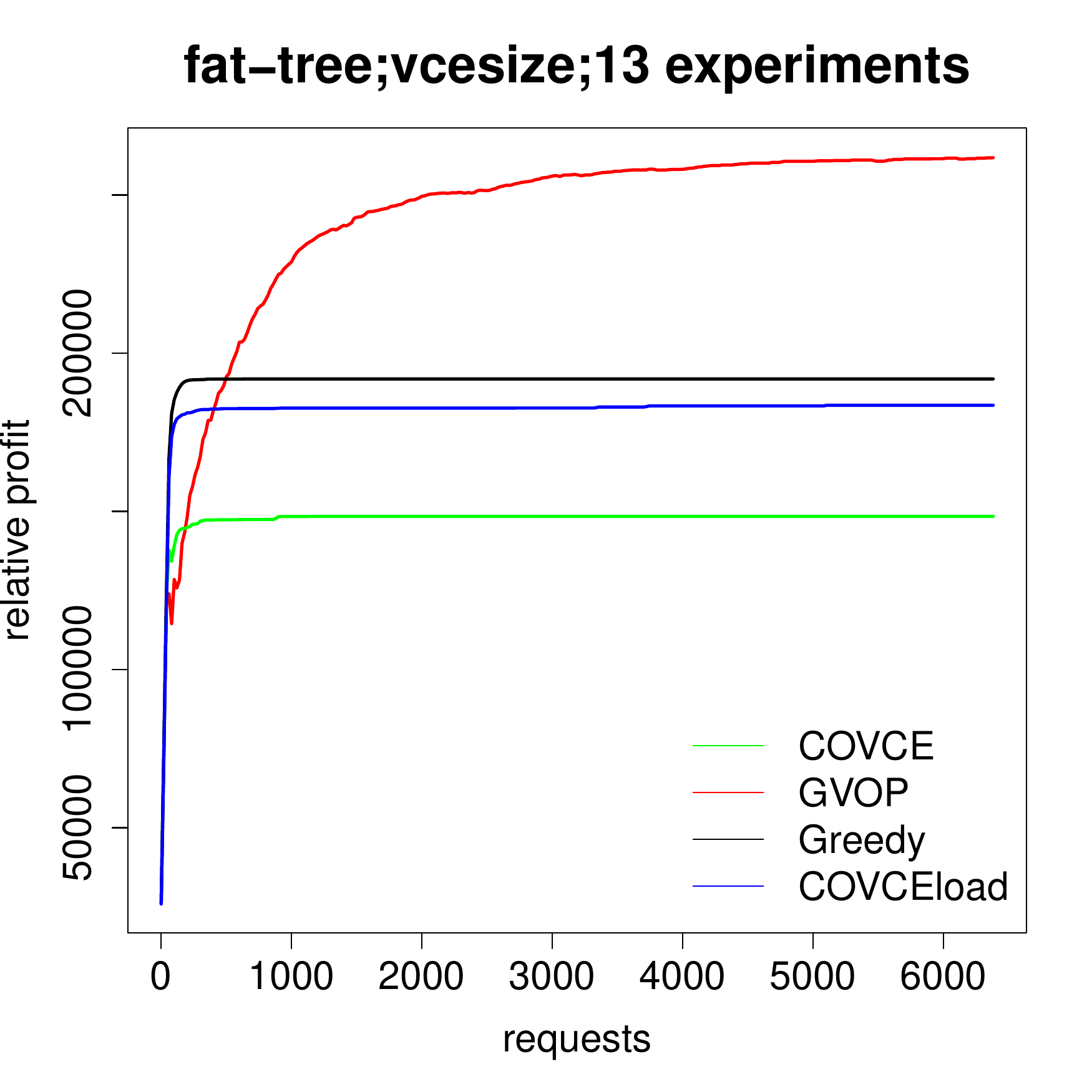}
	\includegraphics[width=0.5\textwidth]{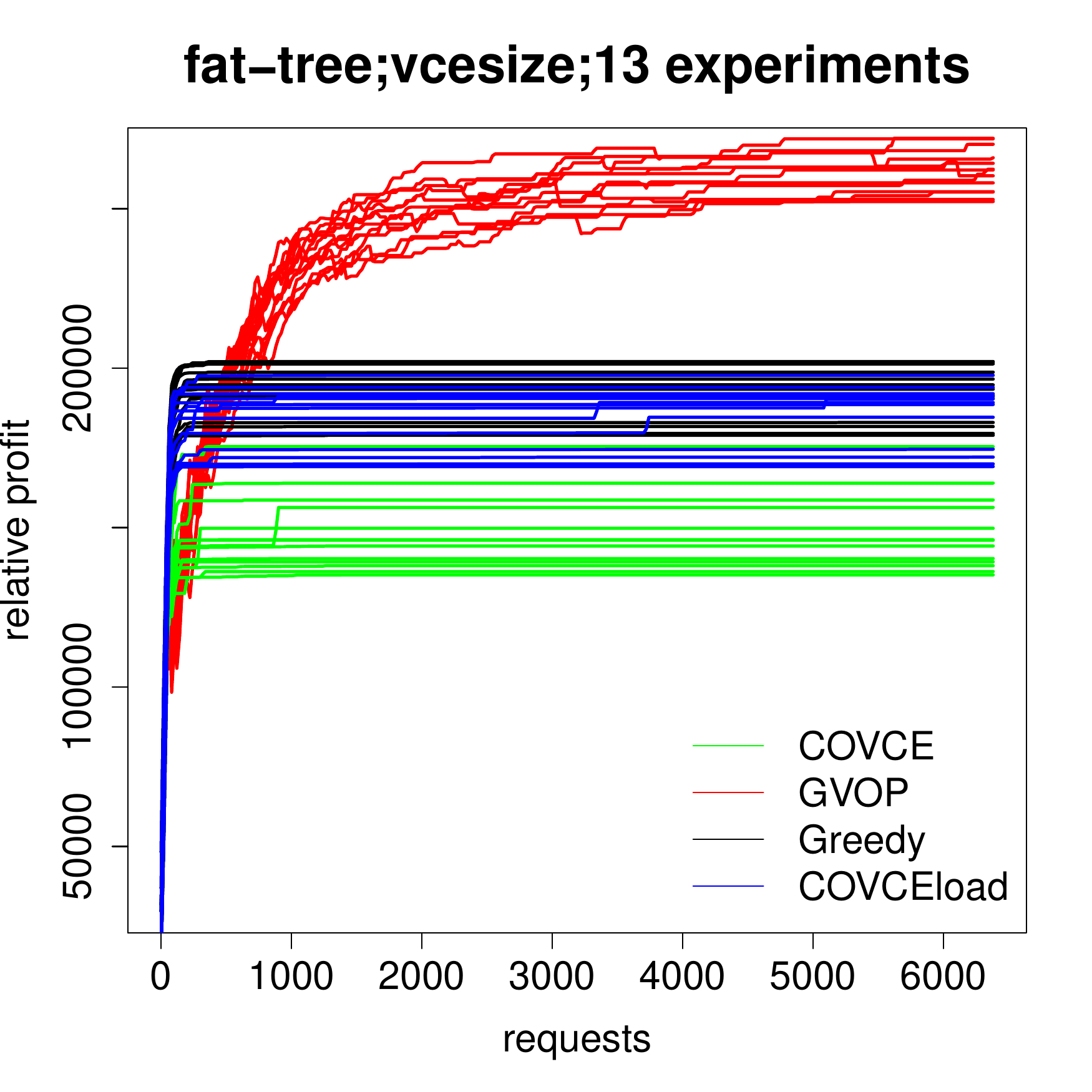}
	\caption{Depicted are relative profits of the four algorithms over 6400 requests for the $\mathtt{vcesize}$ BP on the 12-fat-tree topology. The results of the 15 experiments are averaged (left) and respectively shown (right).}
	\label{fig5}
\end{figure}
In general, the competitive algorithms show an improvement in their performance relatively to the non-competitive ones in comparison to Figure \ref{fig24}. More precisely, the GVOP algorithm's curve has a similar progress but outperforms the greedy algorithm in this case. The COVCE algorithm's performance also improved, its curve is closer to the non-competitive algorithms. Accordingly, the COVCEload algorithm also performs better than in the $\mathtt{reqsize}$ experiments because its relative profit is closer to the greedy algorithm's one. Generally speaking, the rejection functionality and competitiveness do pay off more for the $\mathtt{vcesize}$ than for the $\mathtt{reqsize}$ BP. One of the reasons could be the higher quotient between the minimum and maximum benefit of this BP. Figure \ref{fig5} (right) shows a similar pattern as Figure \ref{fig24} (right) since all algorithms progress as bundles of curves. 

The corresponding acceptance ratio of the GVOP algorithm in Figure \ref{fig8} drops faster than for the $\mathtt{reqsize}$ BP in Figure \ref{fig7} in which its acceptance ratio is visibly higher than the others even at the end of the request sequence. In contrast to this, the acceptance ratio is not distinguishable anymore after 4000 requests in Figure \ref{fig8}. The latter observations show that the GVOP algorithm accepted more requests in general. A possible reason for that might be a higher benefit to embedding cost ratio due to the higher benefit interval that is created upon the formula in Equation \ref{eq59}, so that the incentive for accepting requests is higher. The violation factor of the GVOP algorithm is also slightly higher in Figure \ref{fig8} than in Figure \ref{fig7} which contributes to the latter observation.
\begin{figure}
	\includegraphics[width=0.5\textwidth]{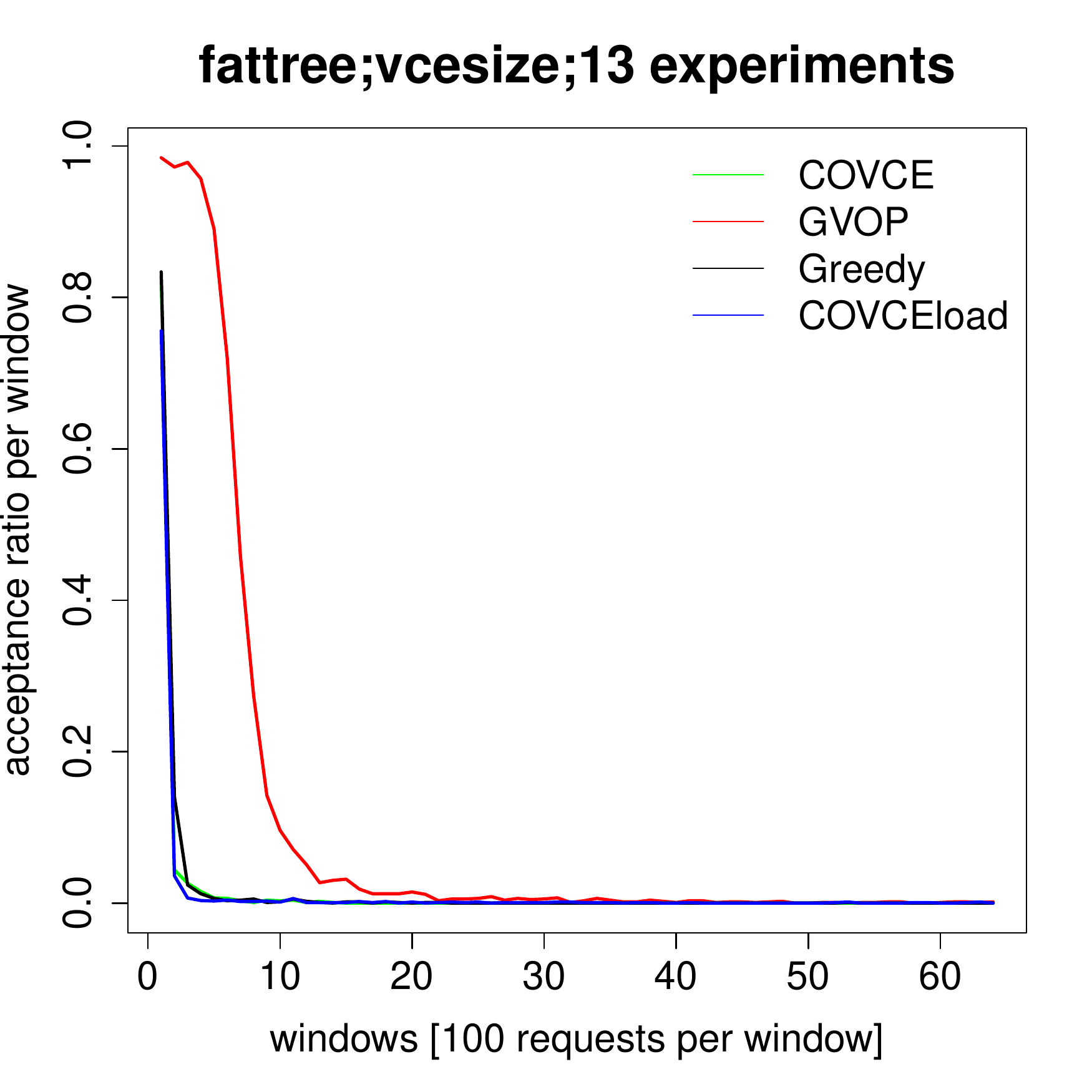}
	\includegraphics[width=0.5\textwidth]{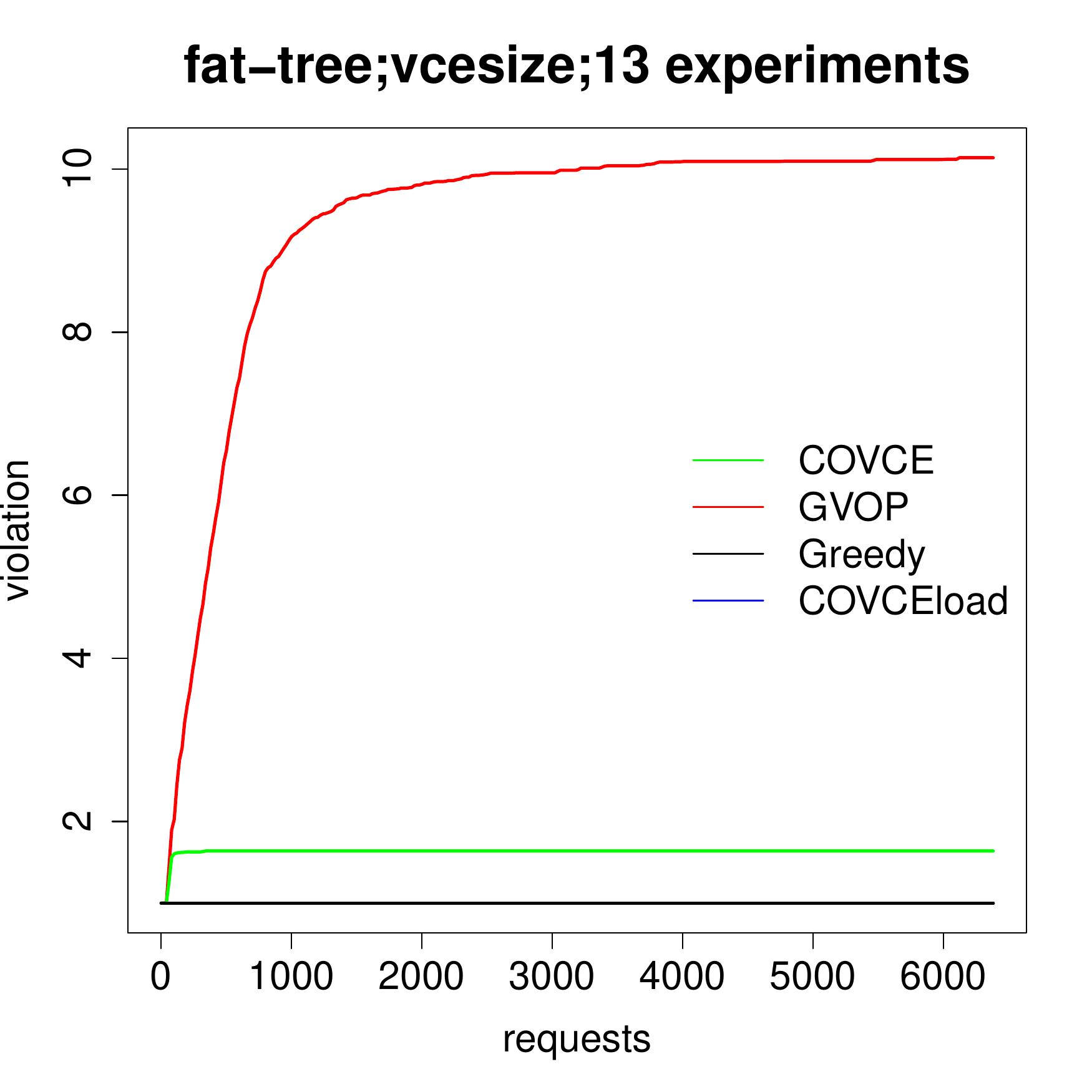}
	\caption{Depicted are the acceptance ratio (left) and violation (right) for $\mathtt{vcesize}$ BP over 6400 requests on the 12-fat-tree topology. The results of the 13 experiments over the four algorithms are averaged for both metrics.}
	\label{fig8}
\end{figure}
It has a violation factor of 10 at request 6400 and has an approximately 20\% higher relative profit than the Greedy VC-ACE algorithm. Thus, when assuming superlinear increasing costs for capacity multiplication, these 20\% higher relative profit may not pay off. On top of that the Greedy VC-ACE algorithm has a higher relative profit than the COVCE and the COVCEload algorithms, so that it pays off the most in this scenario.

In general, the non-competitive algorithms show a better performance than the competitive ones as soon as the benefit function is dependent on the request or embedding size because then smaller benefit to embedding cost ratios occur.
\paragraph{Wave BP}
In this section, the BP $\mathtt{wave}$ is discussed. The benefit function consists of the sinus function
\begin{equation}
	b(r_i) = 300\cdot \mathrm{sin}(0.1\cdot i) + 400
\end{equation}
The parameters of the sinus function are chosen like this because they result in an amplitude of 300 and a whole period is passed after approximately 60 requests. Thus, the corresponding benefit interval is similar to the previous experiments. In addition there are enough periods to make clear statements over 6400 requests. The first noticeable behavior is seen in the benefit bar plot in Figure \ref{fig10}.  
\begin{figure}
	\centering
	\includegraphics[width=\textwidth]{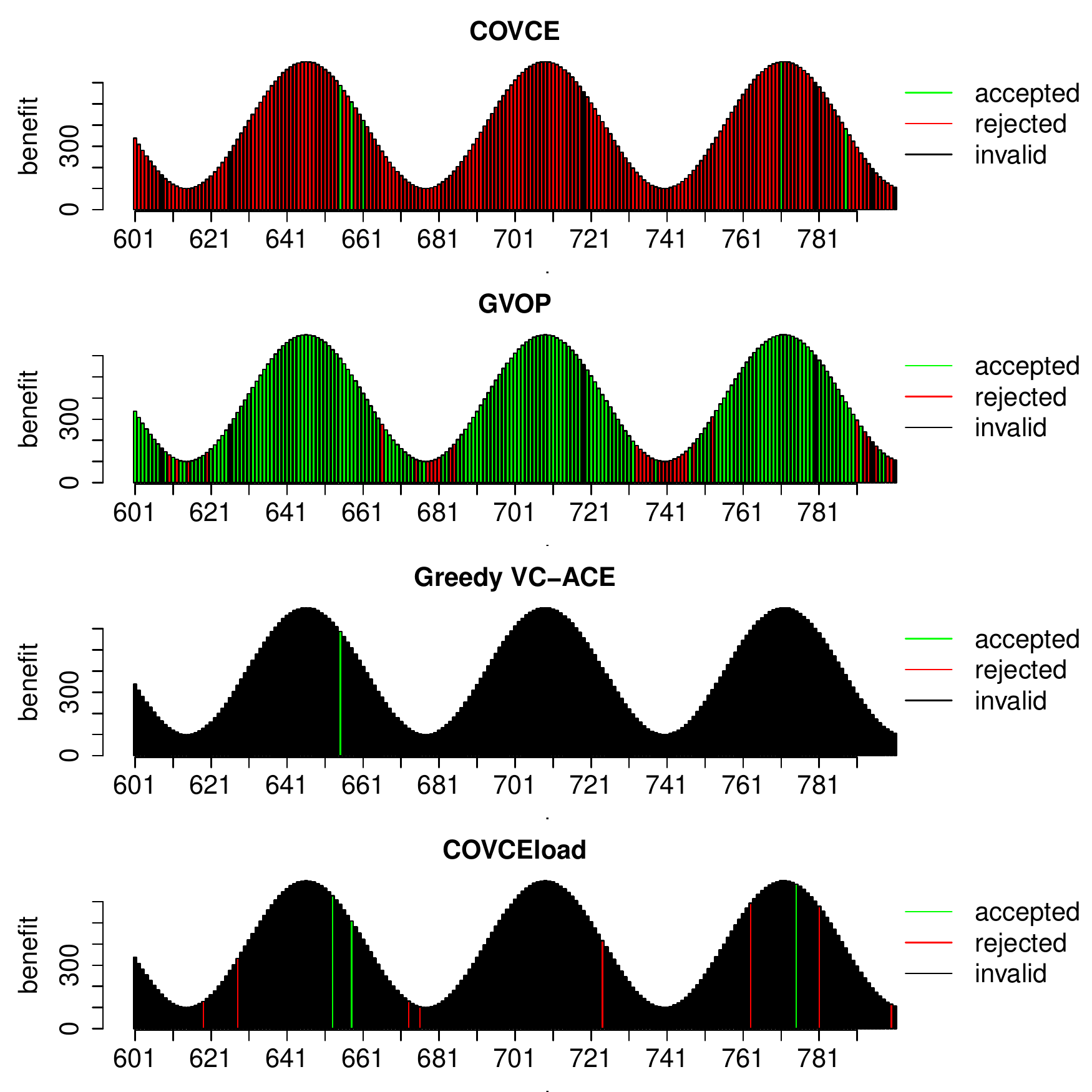}
	\caption{Illustrated are the benefits as bars of the requests 600-800 of a single $\mathtt{wave}$ experiment on the 12-fat-tree topology for all four algorithms. The colors depict whether the request was accepted (green), rejected (red) or invalid (black).}
	\label{fig10}
\end{figure}
In this plot, the acceptances and rejections of the subsequence 600-800 are shown. The COVCEload and the Greedy VC-ACE algorithm rarely embed a request at this state of the request sequence.
To recapitulate, the COVCE algorithm does not accept as many requests as the GVOP anymore because of its stricter competitive ratio on the capacity constraints. Nevertheless, the algorithms have in common that they mostly accept requests whose benefits are located around the peaks of the sinus curve. This behavior is more clearly observable for the COVCE algorithm because it is closer to reach its maximum violation at this state of the request sequence as seen in Figure \ref{fig26} and therefore accepts only requests with a high benefit to embedding-cost ratio.
\begin{figure}
	\includegraphics[width=0.5\textwidth]{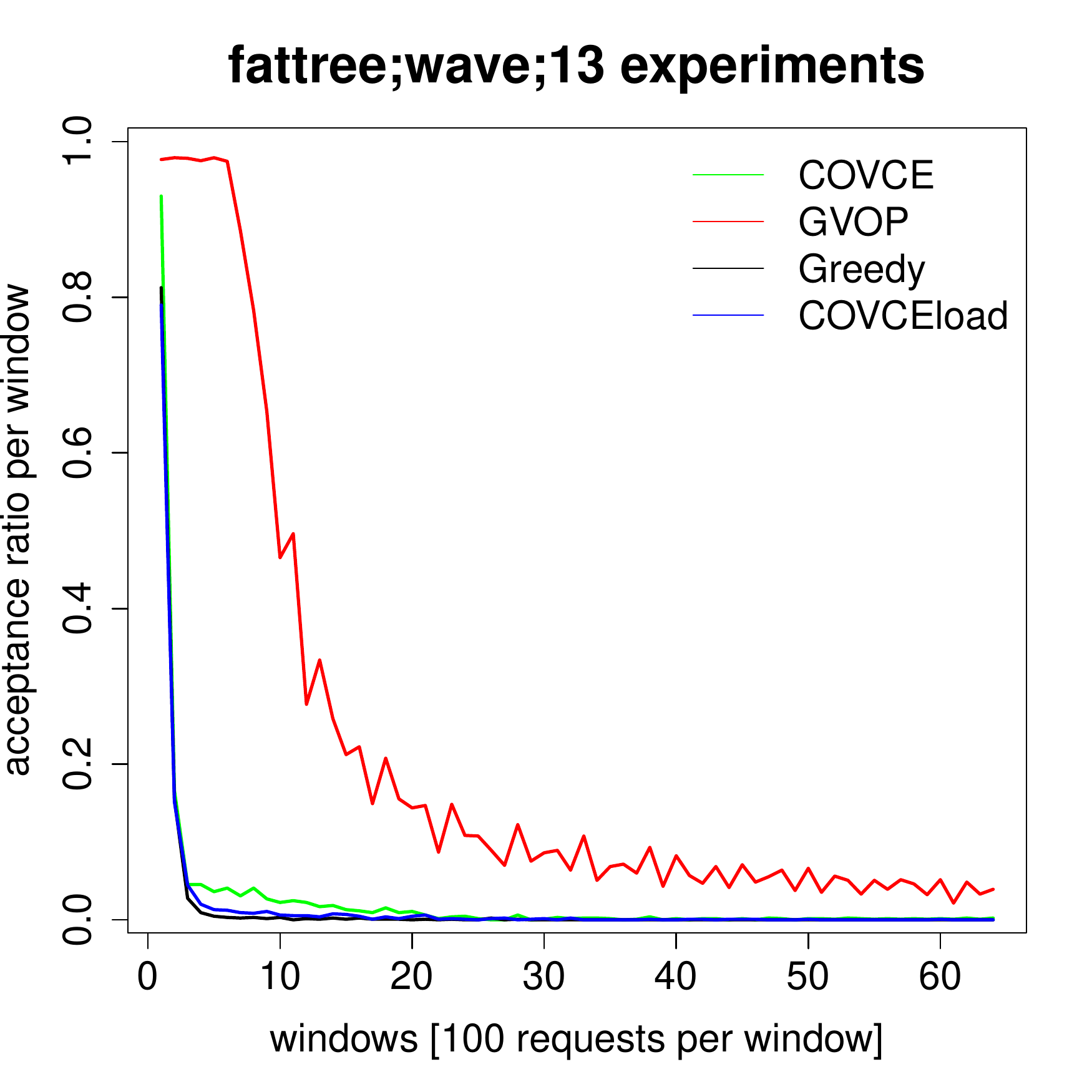}
	\includegraphics[width=0.5\textwidth]{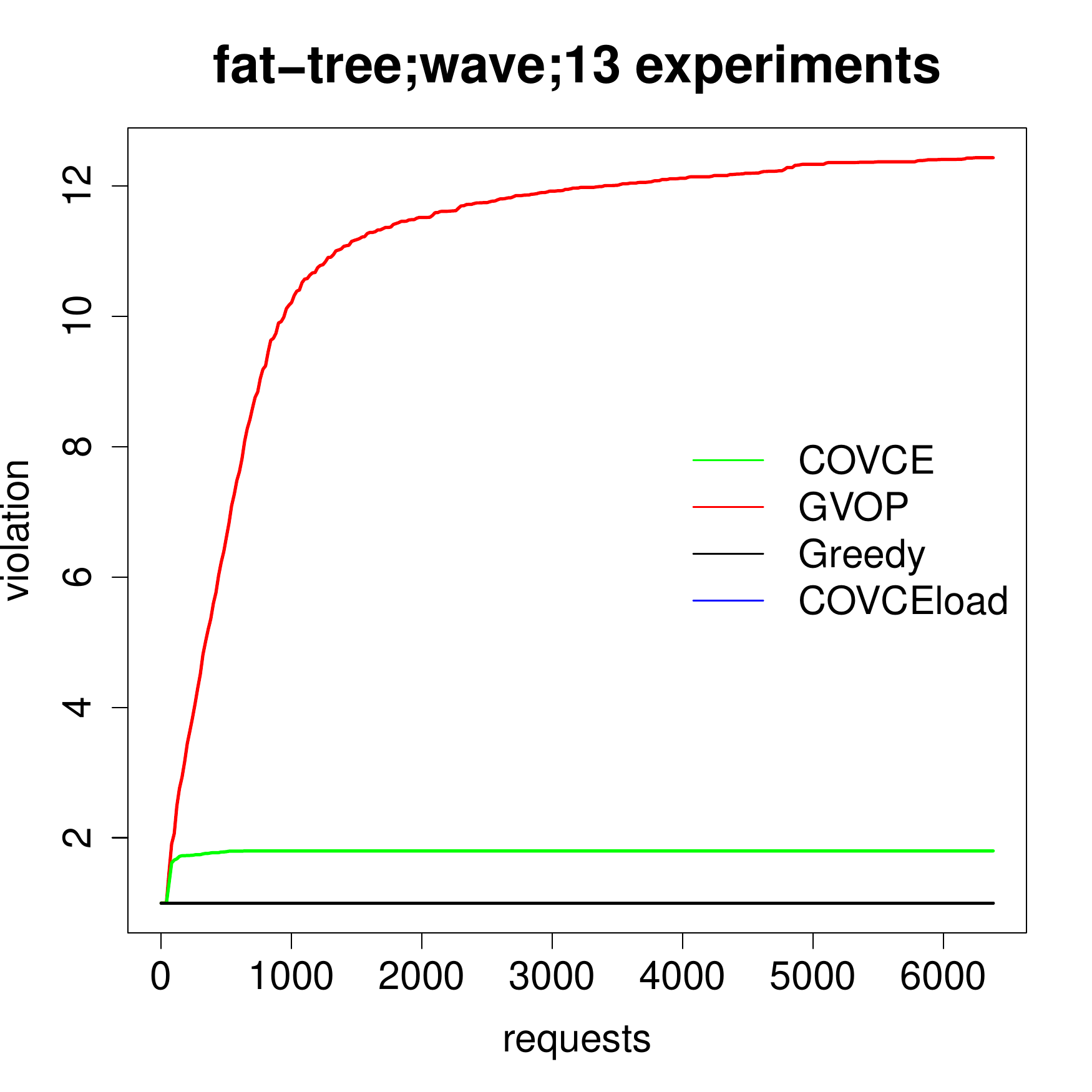}
	\caption{Depicted are the acceptance ratio (left) and violation (right) for $\mathtt{wave}$ BP over 6400 requests on the 12-fat-tree topology. The results of the 13 experiments over the four algorithms are averaged for both metrics.}
	\label{fig26}
\end{figure}
On the other hand, the GVOP algorithm still accepts many requests accordingly to its relatively low violation factor at this state. Nevertheless, there is already a clearly visible area in which the GVOP algorithm does not accept any requests, namely the wave troughs. Thus, the closer the maximum violation is approached, the wider this interval of non-accepted requests and accordingly the thinner the interval of accepted requests becomes. This interval is for the COVCE algorithm thinner than for the GVOP algorithm at this state of the execution. Furthermore, the ``acceptance interval'' gets thinner the more requests are being accepted which is beneficial for both algorithms. These acceptance intervals are noticeable in the zig-zag pattern of the GVOP algorithm in Figure \ref{fig26}. In addition, the difference between the four acceptance ratios is again higher in comparison to the acceptance ratios of the $\mathtt{reqsize}$ experiments. A reason for that might be a higher difference between the benefit to embedding cost ratio than for the $\mathtt{reqsize}$ BP. 

Next, the relative profit is shown in Figure \ref{fig12}. 
\begin{figure}
	\includegraphics[width=0.5\textwidth]{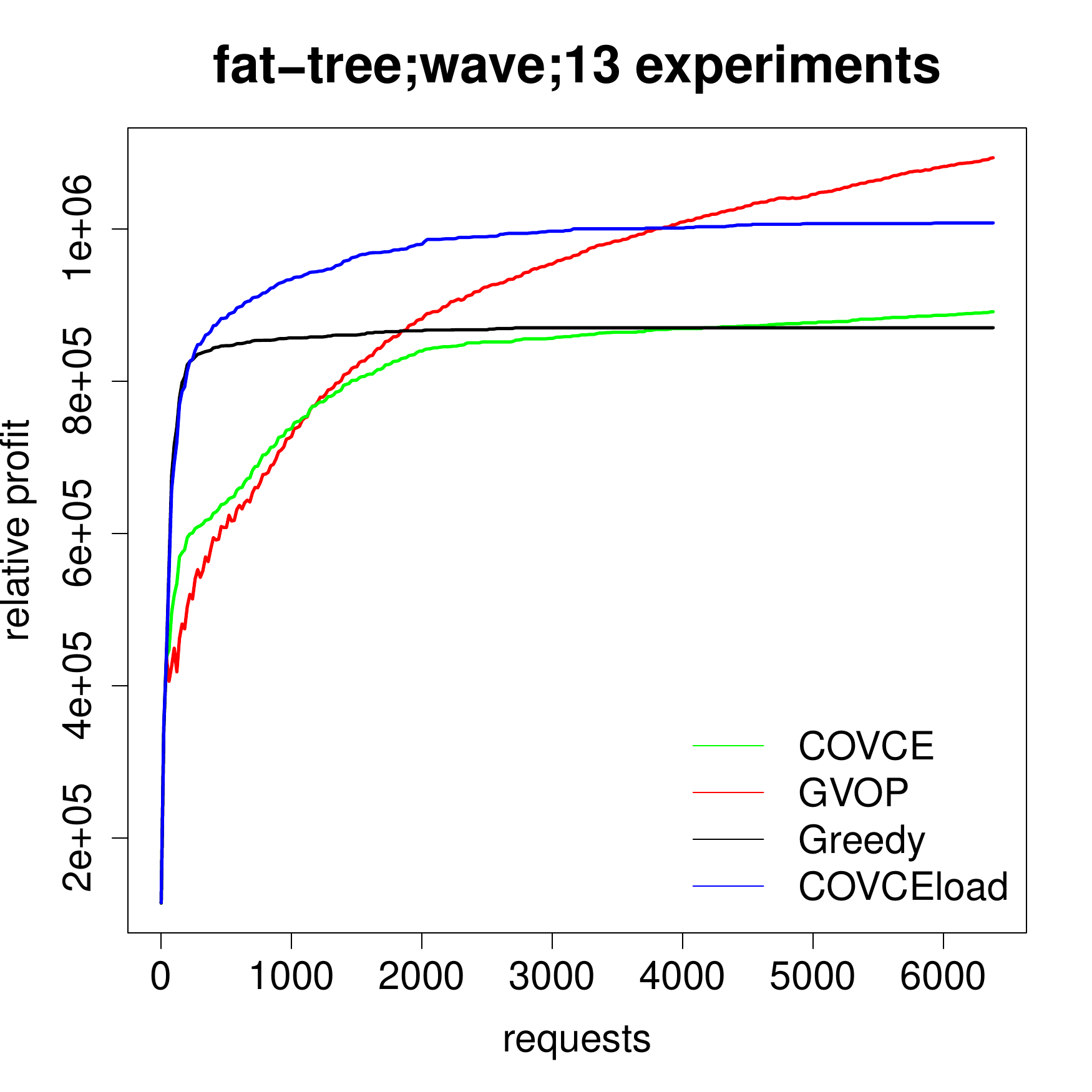}
	\includegraphics[width=0.5\textwidth]{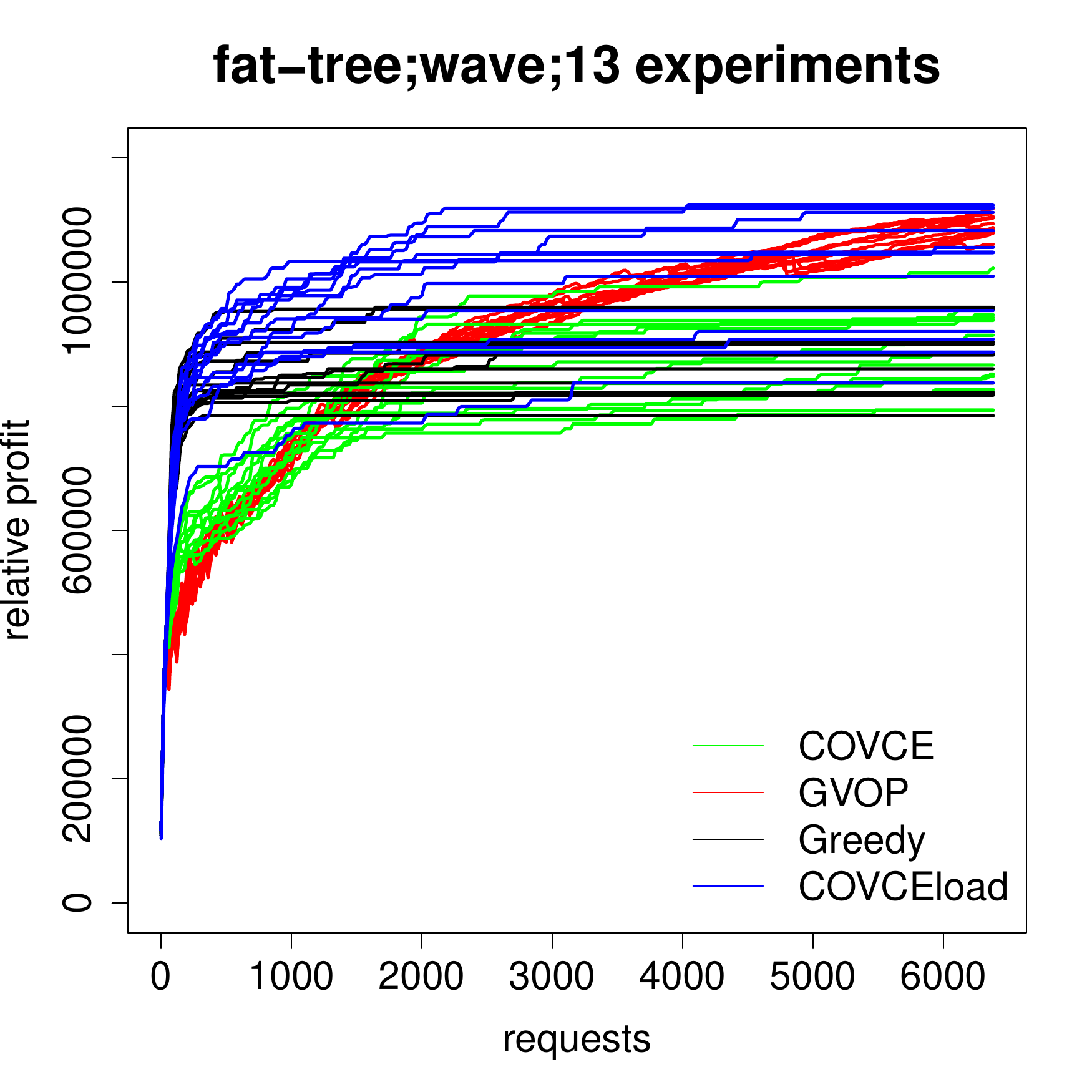}
	\caption{Depicted are relative profits of the four algorithms over 6400 requests for the $\mathtt{wave}$ BP on the 12-fat-tree topology. The results of the 13 experiments are averaged (left) and respectively shown (right).}
	\label{fig12}
\end{figure}
The gradients of the competitive online algorithms are higher than in the $\mathtt{reqsize}$ experiments because with every further period of the sinus wave, there is the possibility of an arriving request at the peak of the benefit function with a relatively small request size, such that the embedding cost are small too, in order to embed the request. So, per sinus wave period, there is a guaranteed benefit peak which is not guaranteed in the $\mathtt{reqsize}$ experiments. This shows that the competitive algorithms benefit a lot from peak benefits combined with low embedding costs as they show a higher relative profit than the Greedy VC-ACE algorithm. Though, the differences between the relative profits is not as high as for the $\mathtt{random}$ experiments. A reason for that might be the higher maximum to minimum benefit ratio of 1000 for the $\mathtt{random}$ BP and only 7 for the $\mathtt{wave}$ pattern. In this regard, the relative profit of the COVCE algorithm overtakes the one of the Greedy VC-ACE algorithm at approximately request 4000 in comparison to approximately request 1600 in the $\mathtt{random}$ BP scenario. The COVCEload algorithm has the second best performance, as in the $\mathtt{random}$ scenario, because it again benefits from the rejection functionality of the COVCE algorithm in addition to not violating any capacities. Besides, this $\mathtt{wave}$ experiment only shows results for this certain sinus benefit function. Hence, to show significant differences between the relative profits, one has to adjust the sinus wave to have a higher amplitude. 

To determine which of the algorithms performs the best, the violation is considered again in Figure \ref{fig26}.
The results of this measurement show that there is no visible difference between the violation of the $\mathtt{wave}$ and $\mathtt{random}$ BP. Referring to this, the acceptance ratios of both BP show similar drop rates when neglecting the local peaks for the $\mathtt{wave}$ BP. These are induced by windows in which the wave peaks occur.

The violation of approximately 12.2 at request 6400 for the GVOP algorithm is not worth the approximately 25\% more relative profit than the greedy algorithm when considering that the COVCE algorithm already results in a slightly higher relative profit for a violation of only factor 2. On top of that, the COVCEload algorithm only has a 9\% less relative profit than the GVOP algorithm with no capacity violation and has a higher relative profit for the first 3700 requests. In addition, when considering Figure \ref{fig12} (right), there is at least one instance of the COVCEload algorithm having a slightly higher relative profit than all instances of the GVOP algorithm. Therefore, the COVCEload has the best performance for this request sequence and benefit pattern.
Though, one has to consider that in other cases with a higher amplitude and therefore higher benefit peaks to embedding cost ratios, the relations between the relative profits might change. 

\paragraph{Peak BP}
The $\mathtt{peak}$ BP is introduced in order to examine such an extreme case with high benefit peaks. The benefit pattern consists of a sequence of 1s with every 100-th benefit being a \emph{peak}. 10, 100 and 1000 as peak values are investigated to reason about at which amplitude there is a significant difference between the COVCE and the GVOP algorithm. Furthermore, it is interesting to find out whether there is a factor at which there is no significant difference between the competitive and the non-competitive online algorithms.

At first, 1000 is considered as the peak value. As already expected, this extreme case shows significant differences in its results. To begin with, consider the acceptance ratio for this case shown in Figure \ref{fig40}. The first thing to notice is that the acceptance ratio of the COVCE algorithm is even lower than the one of the Greedy VC-ACE algorithm until approximately request 300. This is possible because the very first request is already a peak with a value of 1000. After this request the next 99 requests have a benefit of 1. So, the algorithm already experiences in a very early state, that such a high benefit can occur and therefore handles the requests much stricter. In general, the acceptance ratios of the both competitive online algorithms are dropping much faster in comparison to the previous experiments. Thus, a significant influence of the benefit difference between the requests is observable. 
\begin{figure}	
	\includegraphics[width=0.5\textwidth]{graphics/acceptance_ratio_window_fattree_peaks_1000.pdf}
	\includegraphics[width=0.5\textwidth]{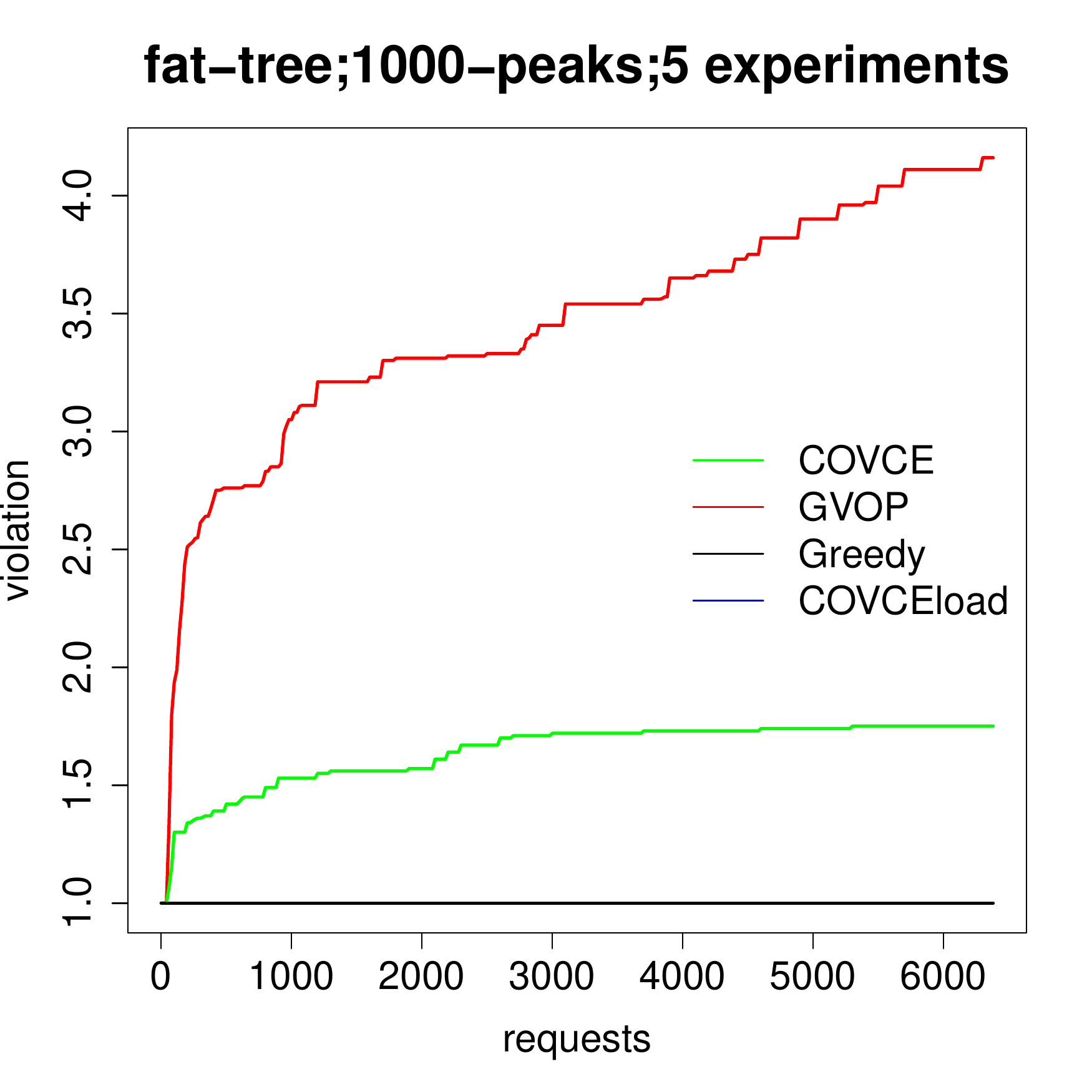}
	\caption{Depicted are the acceptance ratio (left) and violation (right) for the 1000-$\mathtt{peak}$ BP over 6400 requests on the 12-fat-tree topology. The results of the 5 experiments over the four algorithms are averaged for both metrics.}
	\label{fig40}
\end{figure}
This observation is also reflected in the violation and the relative profit. The violations of the competitive online algorithms are now less than in the previous experiments, shown in Figure \ref{fig40}.
The graphs also look more square-edged and abrupt than the previous due to mostly accepting only every 100-th request which also reflects the lower violation factor after 6400 requests in comparison to the previous experiments. This also shows up in the relatively high gradient in comparison to previous experiments. Accepting less requests induces also a lower relative profit than in the previous BPs, which is shown in Figure \ref{fig15}.
\begin{figure}
	\includegraphics[width=0.5\textwidth]{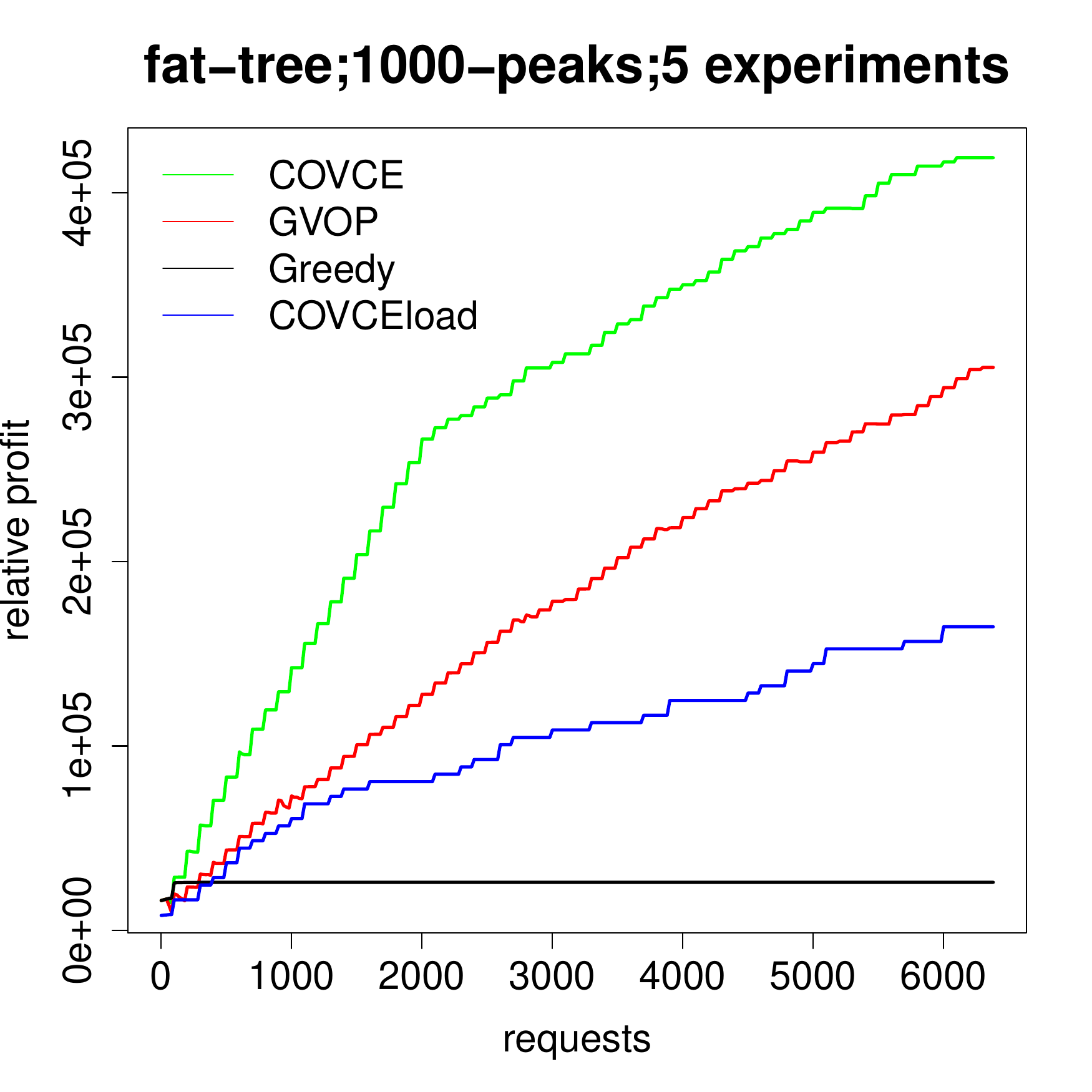}
	\includegraphics[width=0.5\textwidth]{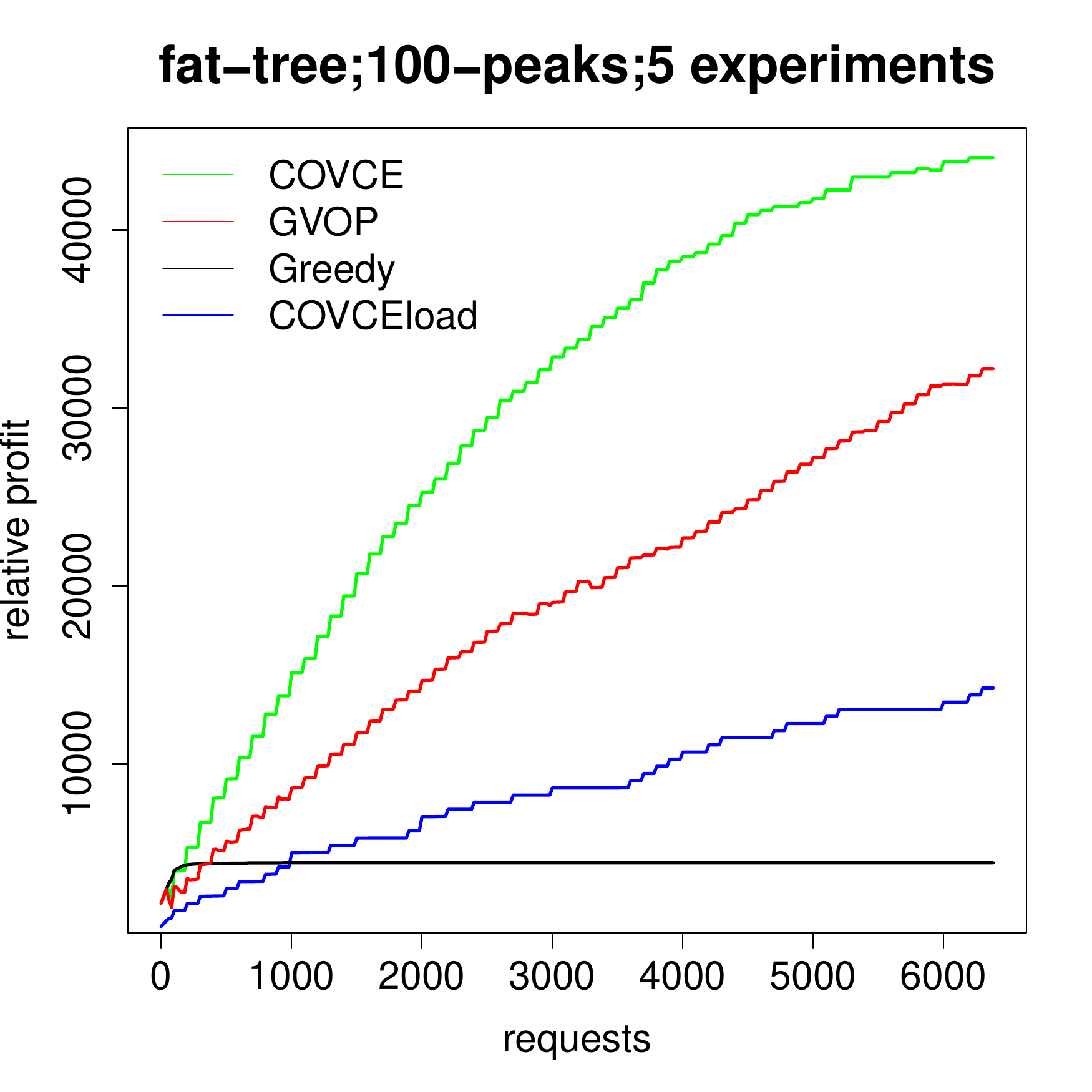}
	\caption{Depicted are the relative profits of the four algorithms for the 1000-$\mathtt{peak}$ BP (left) and the 100-$\mathtt{peak}$ BP (right). The 5 experiments are averaged and run over 6400 requests on a 12-fat-tree topology.}
	\label{fig15}
\end{figure}
Here, a clearly domination of the competitive over the non-competitive online algorithms is noticeable. The functions also have the characteristic of looking abrupt and square-edged because the most accepted requests are the ones with the 1000-peak-benefits which give every time an impulse to the relative profit function. It is though not predictable whether the GVOP algorithm will overtake the COVCE algorithm in an experiment with a longer sequence of request. 
The COVCE algorithm has the best performance and its relative profit is more than twice as high as the one of the COVCEload algorithm.

Next, a less extreme case is considered with a peak value of 100. The acceptance ratio and violation plots still show the same curve progressions as in the 1000-peak experiment and are therefore omitted. The more interesting case is now how the relative profit of the algorithms behaves. It is shown in Figure \ref{fig15}. The gap between the graphs looks similar as for the 1000-peak case. Nevertheless, the competitive algorithms still outperform the greedy and the COVCEload algorithm for the 100-peak factor. Again, it is recognizable that the COVCE algorithm's relative profit gains a lot from very early peak-benefits and high differences in between the requests' benefits.
\begin{figure}
	\includegraphics[width=0.5\textwidth]{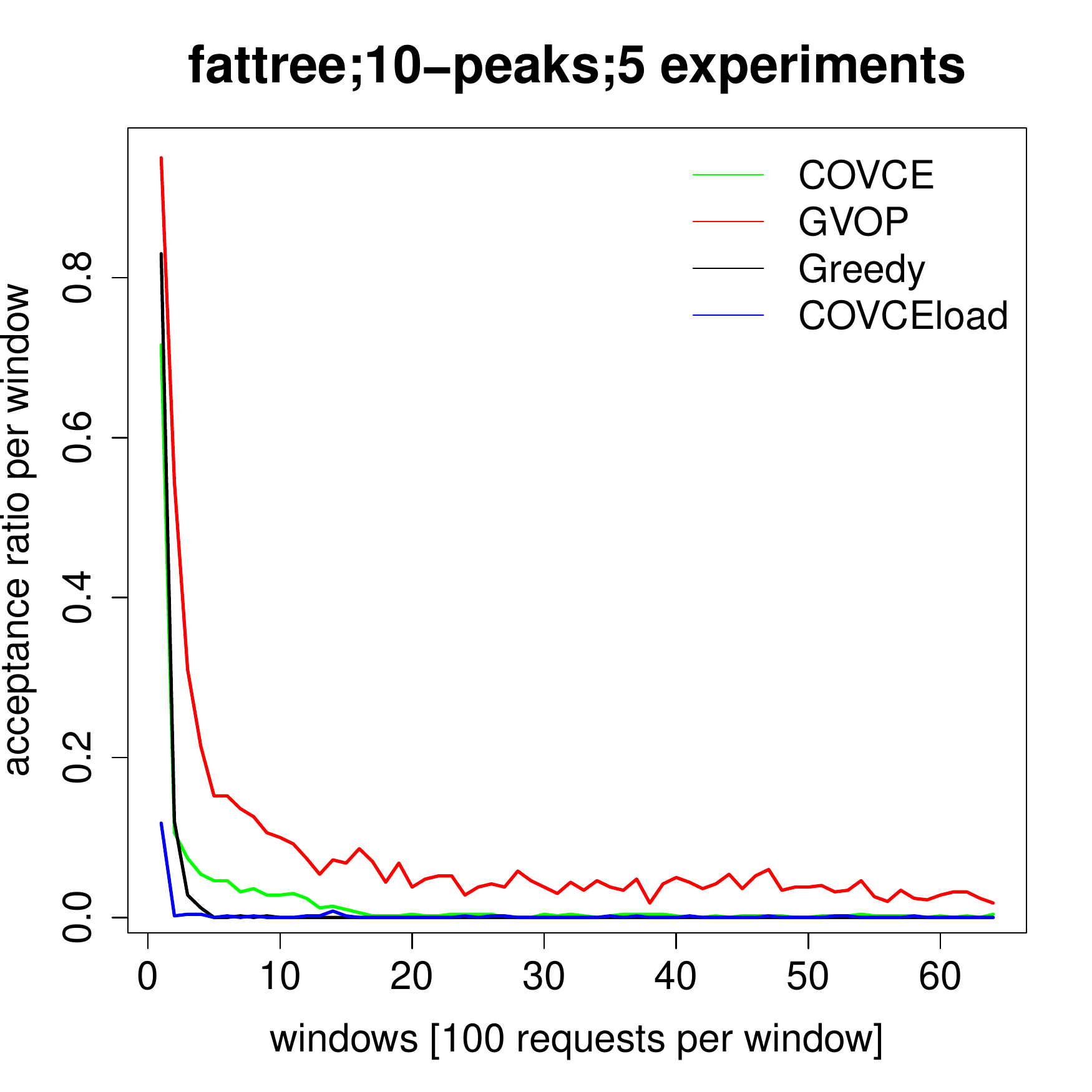}
	\includegraphics[width=0.5\textwidth]{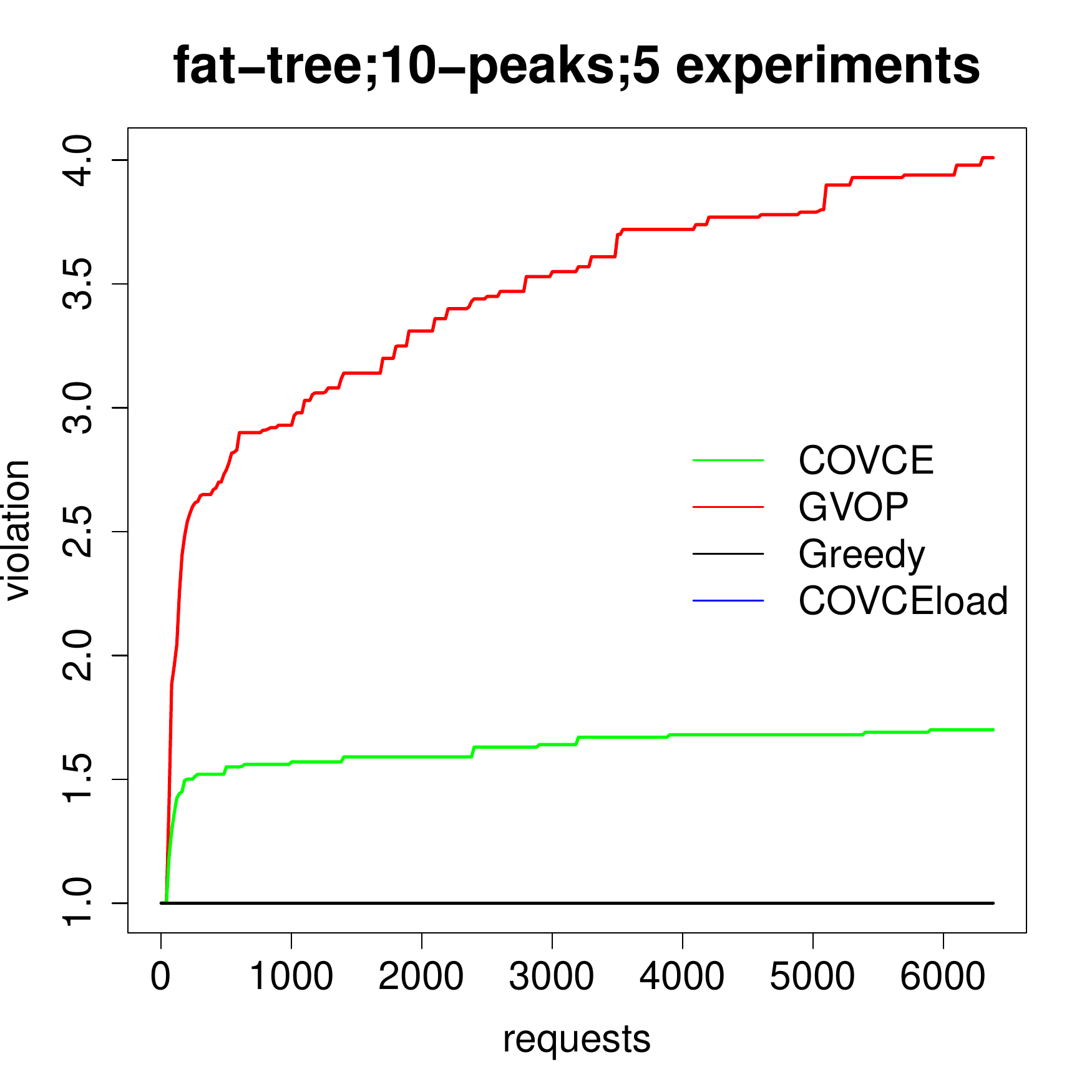}
	\caption{Depicted are the acceptance ratio (left) and violation (right) for the 10-$\mathtt{peak}$ BP over 6400 requests on the 12-fat-tree topology. The results of the 5 experiments over the four algorithms are averaged for both metrics.}
	\label{fig19}
\end{figure}

Finally, the 10-peak case is considered. The acceptance ratio is shown in Figure \ref{fig19}. The COVCE algorithm now accepts more requests in general. This is due to the much lower peak-benefit value of 10, so that the difference between the highest and lowest benefit value is barely noticeable for the COVCE algorithm. In addition, the benefit to embedding cost ratio is smaller. Accordingly, the COVCE algorithm has a lower relative profit by accepting more 1-benefit-requests. It even shows a worse relative profit than the GVOP algorithm, depicted in Figure \ref{fig20}.
\begin{figure}
	\centering
	\includegraphics[width=0.5\textwidth]{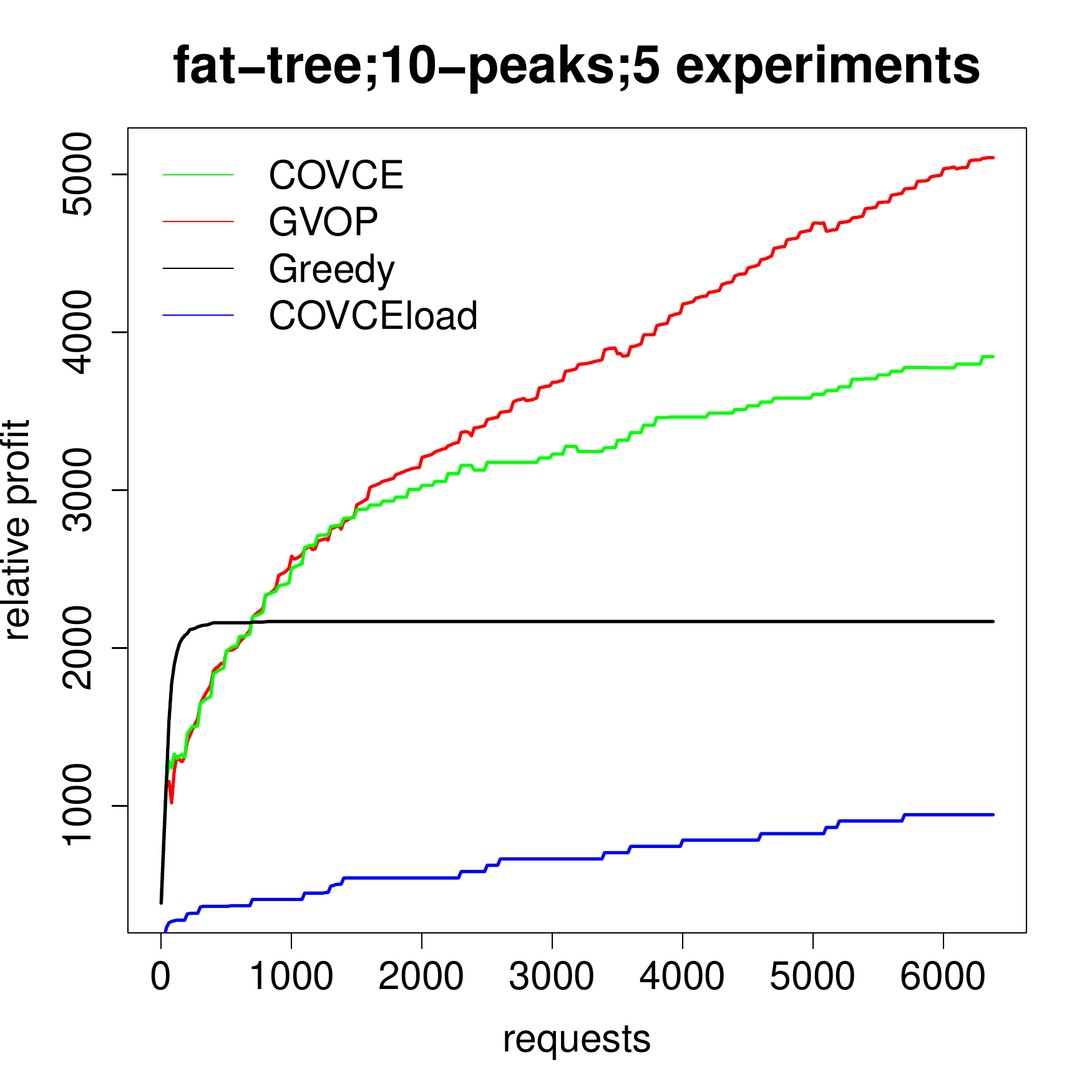}
	\caption{Depicted are the relative profits of the four algorithms for the 10-$\mathtt{peak}$ BP. The results of the 5 experiments are averaged and run on a 12-fat-tree topology for 6400 requests.}
	\label{fig20}
\end{figure}
Moreover, both competitive algorithms have now a lower relative profit difference to the Greedy VC-ACE algorithm since peaks of 10 do not make a significant difference anymore which is a similar factor to the one of the $\mathtt{wave}$ experiments (factor 7). The gap between the non-competitive and competitive algorithms may be higher in the 10-$\mathtt{peaks}$ scenario because here more 1-benefit requests are accepted than wave-trough-benefit requests in the $\mathtt{wave}$ scenario. The COVCEload algorithm has the worst performance. Due to its rejection functionality it rarely accepts a request in relation to the other algorithms, which is reflected in Figure \ref{fig19}, showing that its acceptance ratio is only approximately 12\% in the very first window.

Considering now the violation in Figure \ref{fig19}, the violation factor of the GVOP algorithm grows faster in comparison to the 100 and 1000 $\mathtt{peaks}$ experiments due to accepting more 1-benefit requests since the amplitude is that low that the competitiveness loses its abilities of rejecting relatively poor-benefit requests.

To sum up, at peak factor 1000 and 100 the COVCE algorithm shows its best performance over all the BPs and at peak factor 10 the GVOP algorithm starts to overtake the COVCE algorithm again in terms of the relative profit. The COVCE algorithm even reaches the double of the Greedy VC-ACE algorithm's relative profit at the 10-peaks experiment. Since the violation and relative profit curves of the COVCE algorithm are still increasing, one can expect the COVCE algorithm's relative profit to become even more than the double of the Greedy one's (Figure \ref{fig20}). So, doubling the resources results in a more than four times higher benefit which may even pay off with superlinear capacity cost multiplications. In conclusion the competitive online algorithms outperform the non-competitive ones for this BP, especially the COVCE algorithm havint the best performance for all the three cases. 

\paragraph{Another Metric: Average Relative Profit}
The relative profit function, that was used to analyze the results of the experiments, involved the maximum violation over all nodes and edges until request $i$. In other words, the violation described the factor by which the provider would have to multiply its resources in order to have no violation over the whole substrate until request $i$. Another perspective on the capacity violation is provided with the \emph{average} violation. It describes the average capacity oversubscription over all nodes and edges until request $i$ and will induce less capacity multiplication costs for the provider since a single node/edge with the maximum oversubscription does not influence the average violation factor as much as the violation factor from the previous section. The average violation and the ARP are formally defined in the metrics Section \ref{metrics}, Equations \ref{eq35} and \ref{eq36}. The ARP is shown in the following plots for the $\mathtt{random}$, $\mathtt{reqsize}$, $\mathtt{vcesize}$ and $\mathtt{wave}$ BP scenarios in order to show improvements of the competitive online algorithms.

Consider at first Figure \ref{fig29} 
\begin{figure}
	\includegraphics[width=0.5\textwidth]{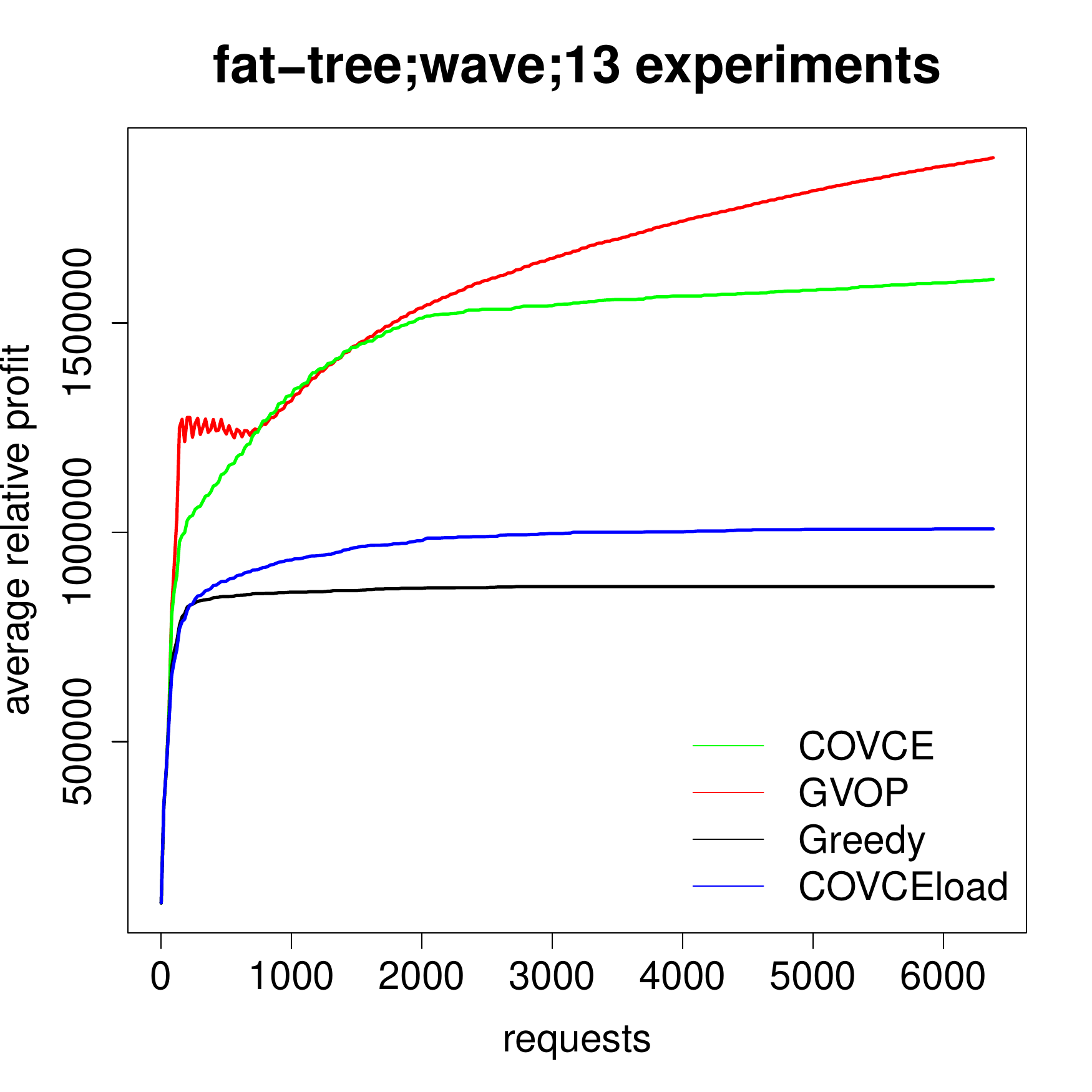}
	\includegraphics[width=0.5\textwidth]{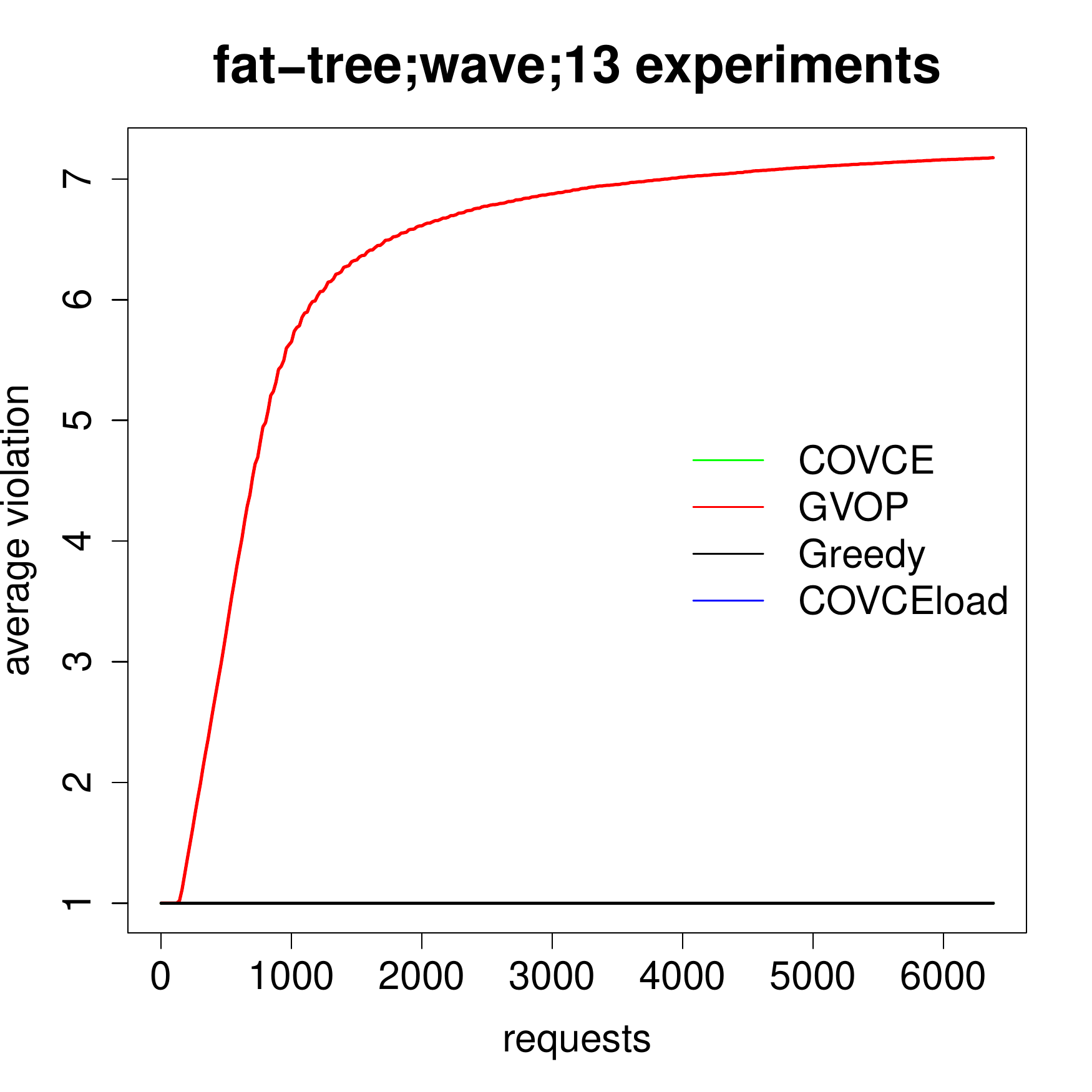}
	\caption{Depicted are the ARP for the $\mathtt{wave}$ BP (left) and the corresponding average violation (right) on a 12-fat-tree topology over 6400 requests. The results of the 15 experiments are averaged.}
	\label{fig29}
\end{figure}
showing the averaged ARP of 13 experiments for the $\mathtt{wave}$ BP and the corresponding average violation in Figure \ref{fig29} (right). Note that the average violation of the COVCE algorithm is exactly 1 as the non competitive algorithms' one. As one can expect, the ARPs of the non-competitive algorithms are the same as before since they do not violate any resources, so that the violation factor is always 1 for both violation definitions. Considering the factor at the end of the request sequence, the ARP of the COVCE algorithm is by factor 1.8 higher than its corresponding relative profit in the $\mathtt{wave}$ BP shown in Figure \ref{fig12}, for the GVOP algorithm it is the factor 1.7.  

Another point is that the ARP of the GVOP algorithm is by factor 2.5 greater than the one of the greedy algorithm (2.0 for COVCEload) and only by factor 1.2 greater than the one of the COVCE algorithm. Considering that the theoretical maximum resource capacity violation of the COVCE algorithm is 2 and that it has a 1.8 times higher ARP than the greedy algorithm (1.6 for the COVCEload), the COVCE algorithm has the best performance out of the four algorithms since the average violation of the COVCE algorithm is 1 and for the GVOP algorithm it is approximately 7.

A similar but also interesting ARP curve progression shows up for the $\mathtt{random}$ BP which is illustrated in Figure \ref{fig30} together with its average violation function in Figure \ref{fig30}.
\begin{figure}
	\includegraphics[width=0.5\textwidth]{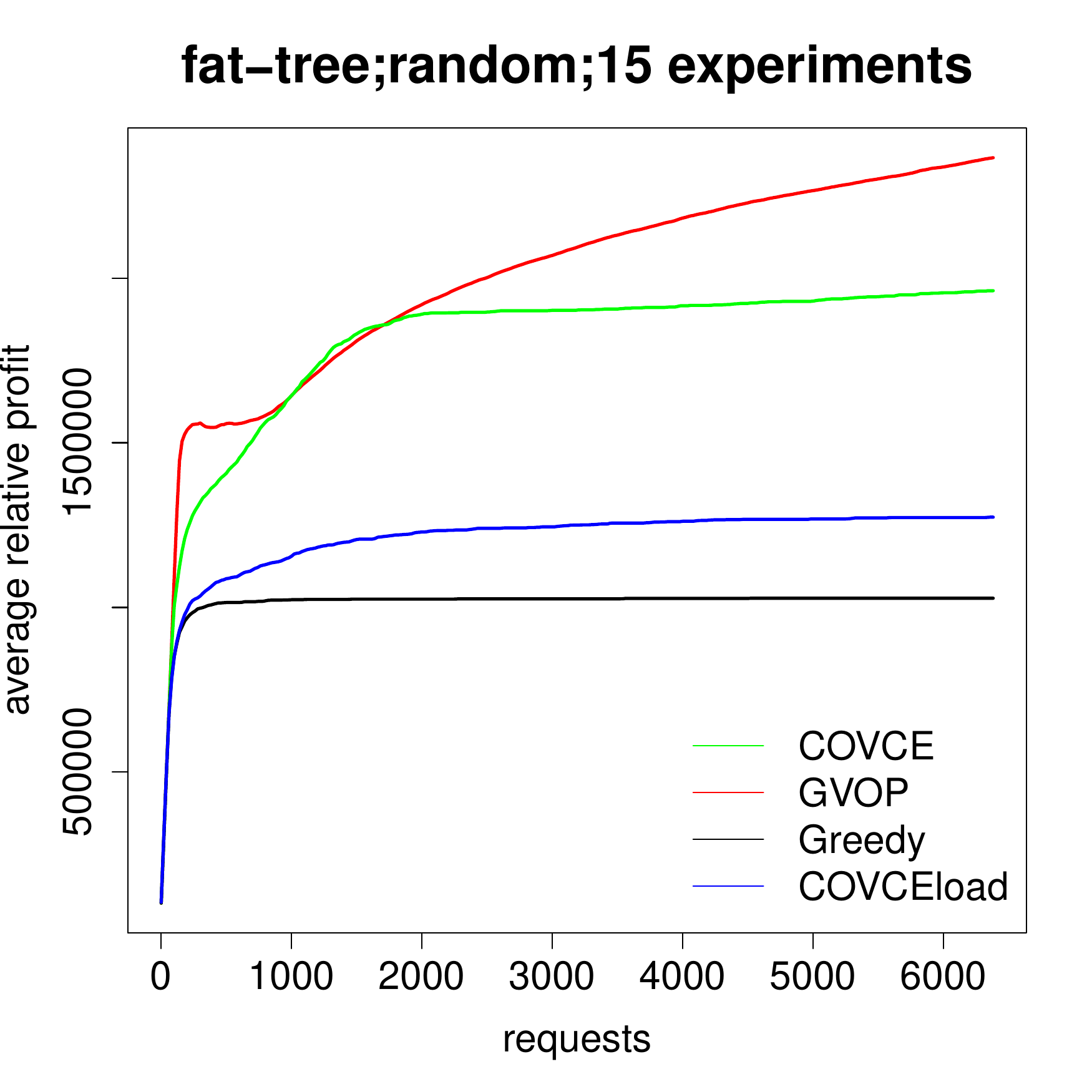}
	\includegraphics[width=0.5\textwidth]{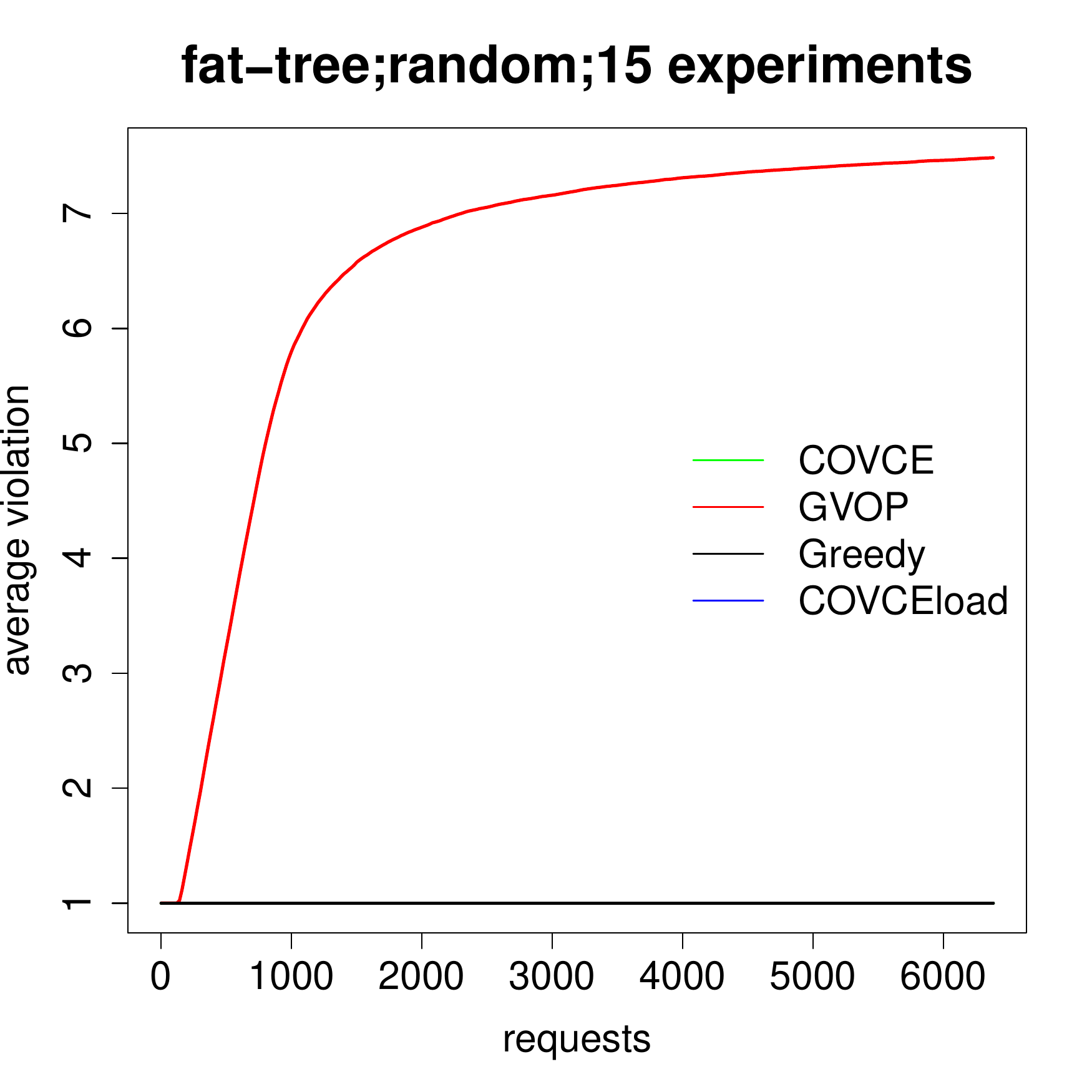}
	\caption{Depicted are the ARP for the $\mathtt{random}$ BP (left) and the corresponding average violation (right) on a 12-fat-tree topology over 6400 requests. The results of the 15 experiments are averaged.}
	\label{fig30}
\end{figure}
Since the benefit comes from a bigger interval, the COVCE algorithm even has a 1.9 times greater ARP than the greedy algorithm at request 6400 (factor 2.3 for the GVOP algorithm). Again, the COVCE algorithm has no capacity violation in the average violation case. Thus, it also outperforms the COVCEload algorithm in this scenario for this metric. Since the GVOP algorithm has a violation factor of approximately 7 for a 20\% higher ARP, one can conclude that the COVCE algorithm has the best performance.

Last but not least, the more realistic cases for the relative profit are re-considered. In Figure \ref{fig31} the $\mathtt{reqsize}$ BP results for the ARP and the average violation in Figure \ref{fig31} are shown.
\begin{figure}
	\includegraphics[width=0.5\textwidth]{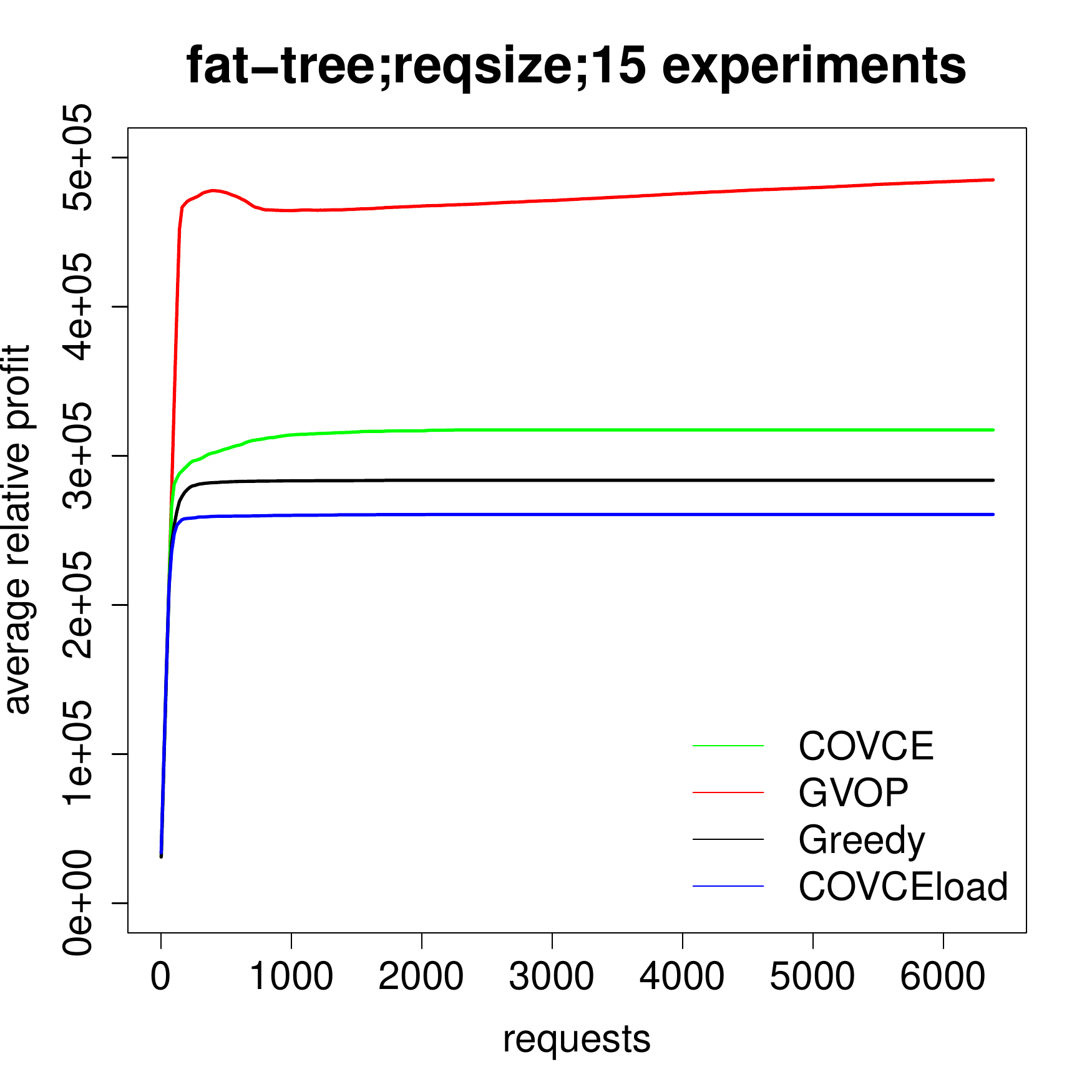}
	\includegraphics[width=0.5\textwidth]{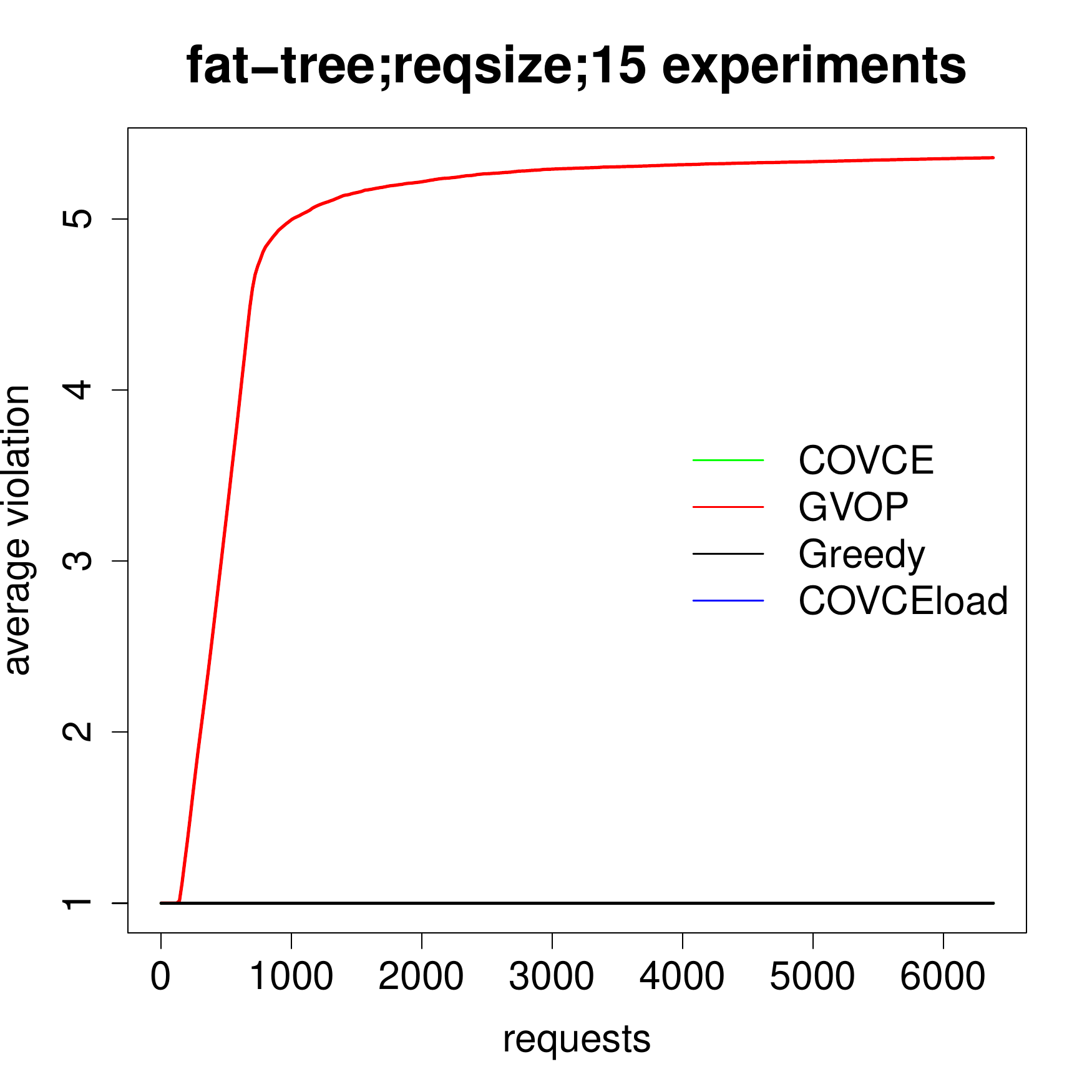}
	\caption{Depicted are the ARP for the $\mathtt{reqsize}$ BP (left) and the corresponding average violation (right) on a 12-fat-tree topology over 6400 requests. The results of the 15 experiments are averaged.}
	\label{fig31}
\end{figure}
Observe that the difference between the ARP of the COVCE and the Greedy VC-ACE algorithm is much smaller than for the other two cases. This shows again that the COVCE algorithm depends on high benefit to embedding cost ratio. The COVCEload algorithm's performance is worse than the Greedy VC-ACE algorithm. The violation factor of the GVOP algorithm is only 5 because it accepts less requests since there are less high-beneficial and simultaneously low-cost requests in comparison to the other two BPs.

Combining the two graphs in Figure \ref{fig31}, the COVCE algorithm has an average violation factor of 1 and therefore outperforms the Greedy VC-ACE and the COVCEload algorithm because it has a higher ARP. 
It is though arguable if the violation factor of 5 for the GVOP algorithm is worth the approximately 1.55 times higher ARP than the COVCE algorithm.

Lastly, the $\mathtt{vcesize}$ BP is re-visited. Its corresponding ARP and violation are shown in Figure \ref{fig60}.
\begin{figure}
	\includegraphics[width=0.5\textwidth]{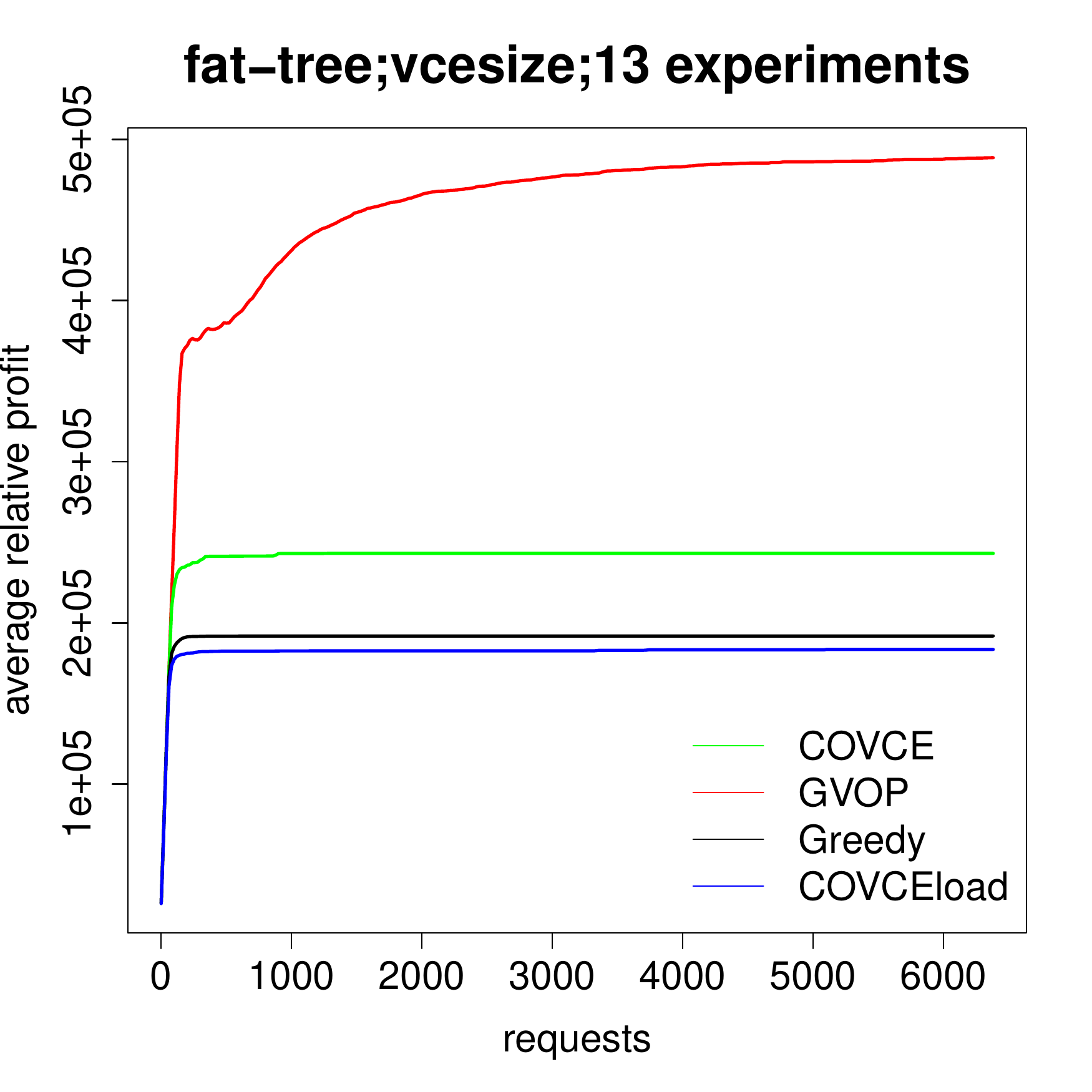}
	\includegraphics[width=0.5\textwidth]{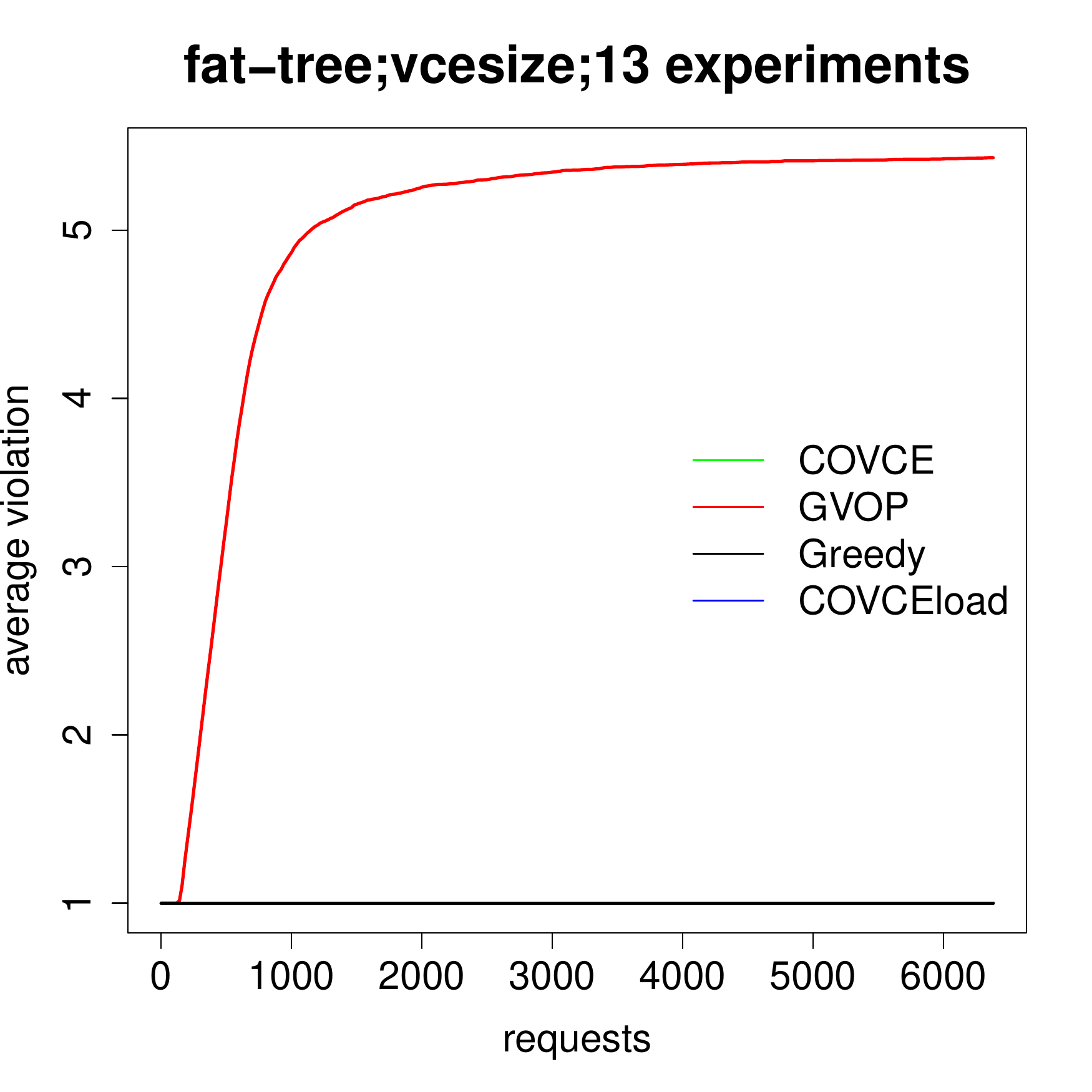}
	\caption{Depicted are the ARP for the $\mathtt{vcesize}$ BP (left) and the corresponding average violation (right) on a 12-fat-tree topology over 6400 requests. The results of the 15 experiments are averaged.}
	\label{fig60}
\end{figure}
Similarly to the $\mathtt{reqsize}$ BP, the GVOP algorithm has a relatively high distance to the other algorithms. Its ARP is nearly twice as high as the COVCE algorithm's one at request 6400. However, it shows a violation factor of 5 so it is arguable whether doubling the ARP is worth a 5 times average violation. The COVCEload and Greedy VC-ACE algorithm are not deeper considered since they perform worse than the COVCE algorithm because it has no capacity violation on average.

The result of this briefly discussed alternative to the relative profit is that the average violation of the competitive algorithms is not close to the previously defined violation since in all three BP cases the ARP is approximately 80\% higher than for the corresponding relative profit. For the shown cases, the COVCE algorithm has the overall best performance, As a result, dependent on the metric that is chosen, different interpretations of the data are possible.

\subsubsection{Results on MDCube Topology}  
In this section, the results of the experiments with an MDCube topology are discussed. The results are briefly presented because they do not provide additional information to the results of the fat-tree experiments.
\begin{figure}[h!]
	\includegraphics[width=0.5\textwidth]{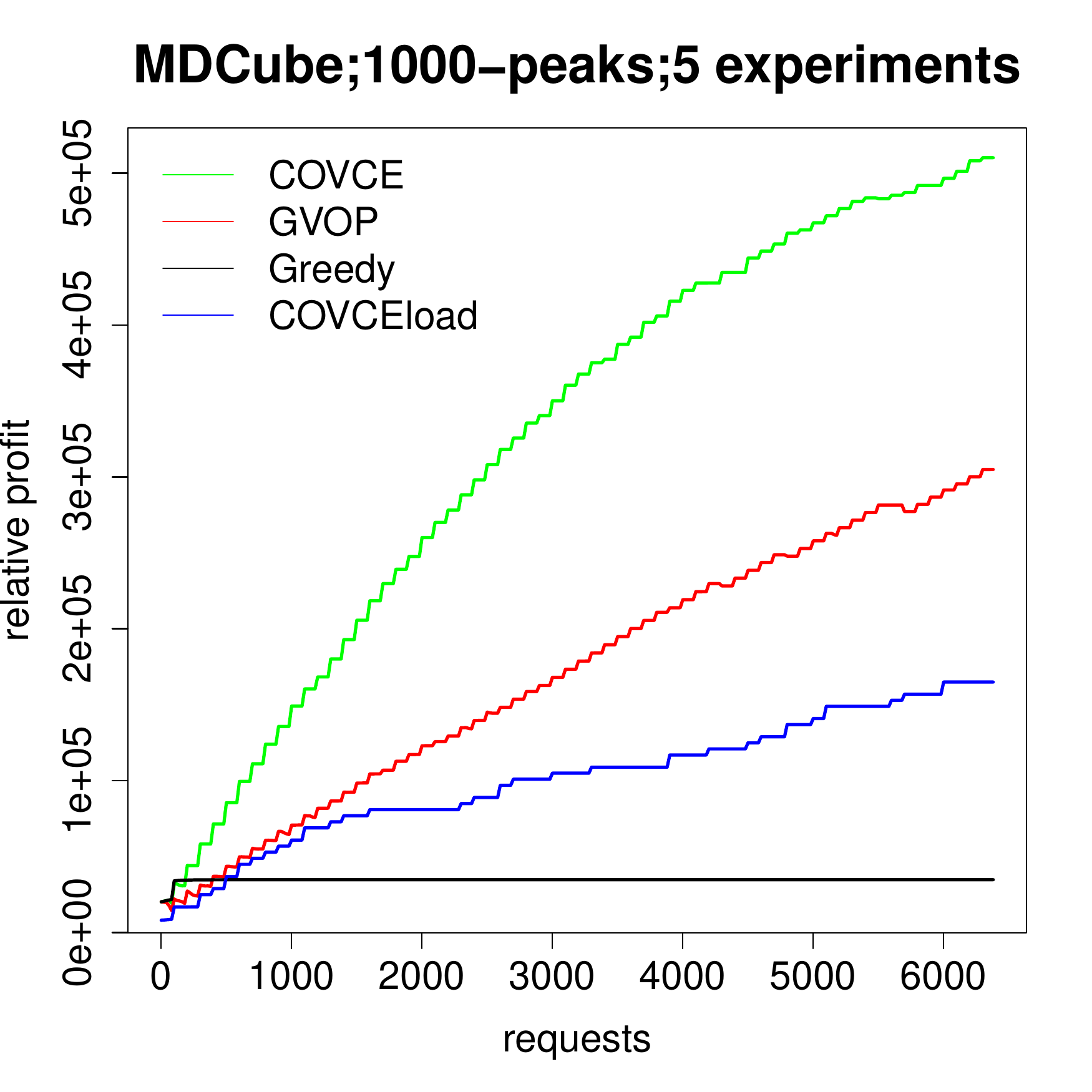}
	\includegraphics[width=0.5\textwidth]{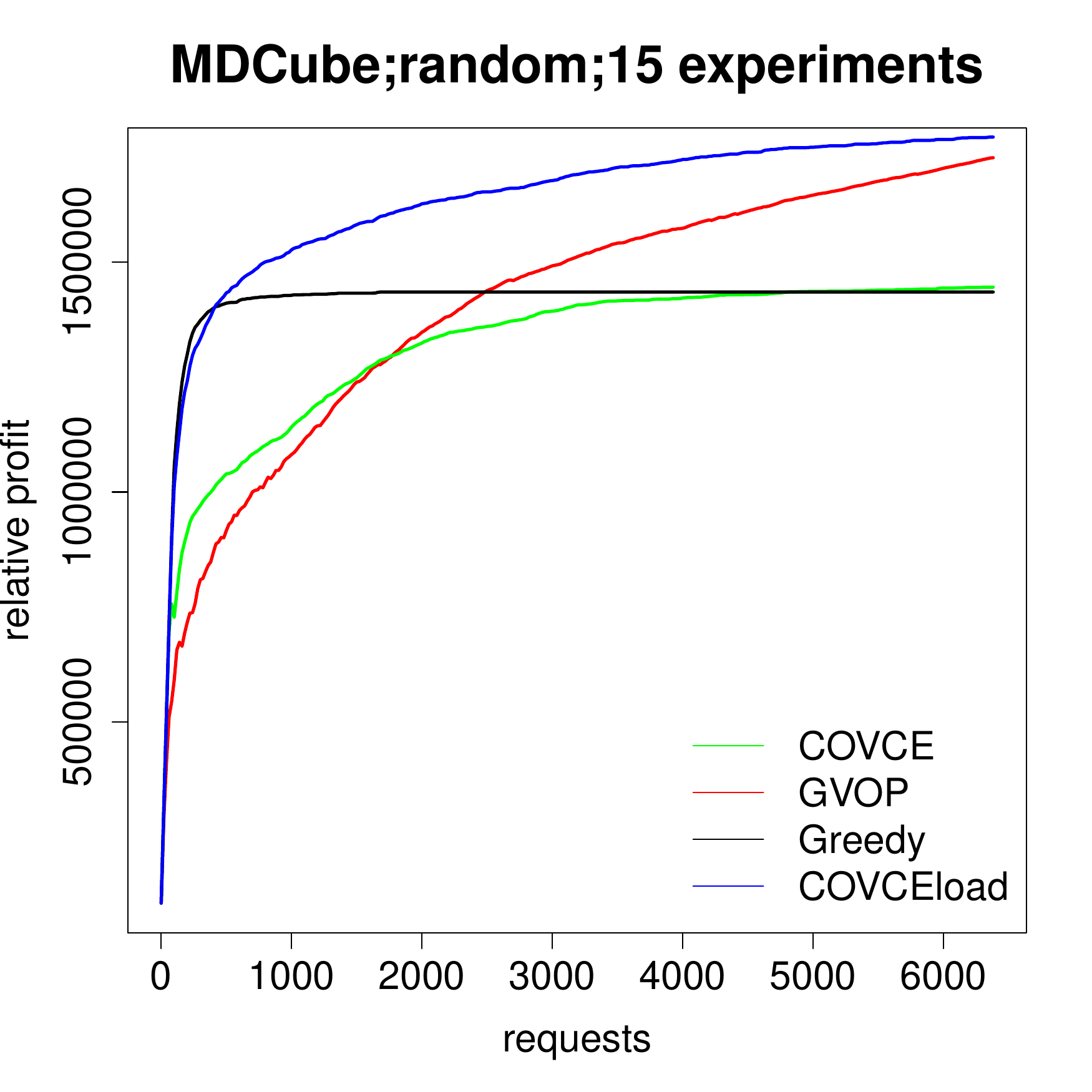}
	\caption{Depicted are the results for the relative profits of the 1000-$\mathtt{peaks}$ BP (left) and the $\mathtt{random}$ BP (right). The former is averaged over 5 and the latter over 15 experiments for the four algorithms on the MDCube topology over 6400 requests.}
	\label{fig46}
\end{figure}
\begin{figure}[h!]
	\includegraphics[width=0.5\textwidth]{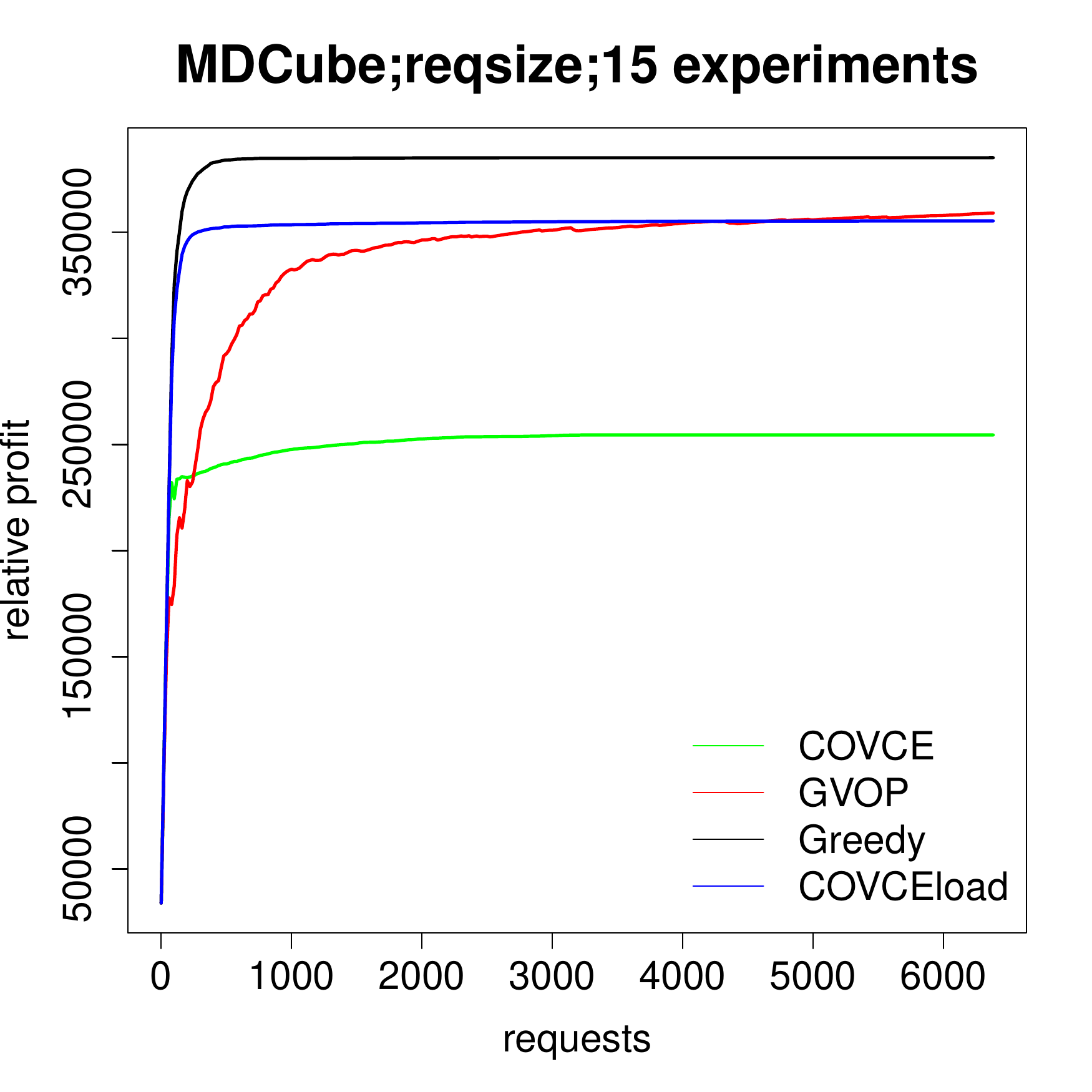}
	\includegraphics[width=0.5\textwidth]{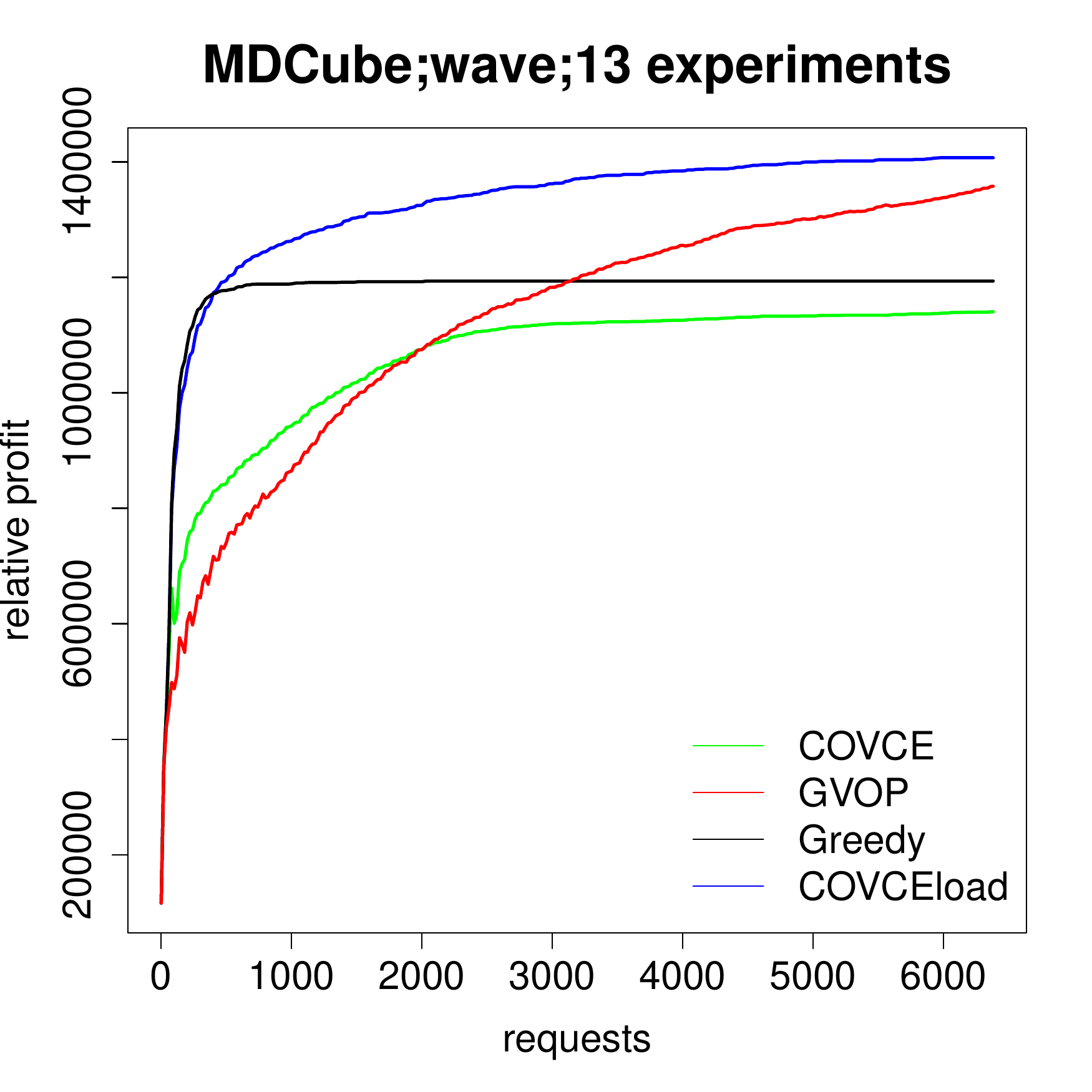}
	\caption{Depicted are the results for the relative profits of the 1000-$\mathtt{reqsize}$ BP (left) and the $\mathtt{wave}$ BP (right). The former is averaged over 15 and the latter over 13 experiments for the four algorithms on the MDCube topology over 6400 requests.}
	\label{fig48}
\end{figure}
\begin{figure}[h!]
	\includegraphics[width=0.5\textwidth]{graphics/acceptance_ratio_window_mdcube_wave.pdf}
	\includegraphics[width=0.5\textwidth]{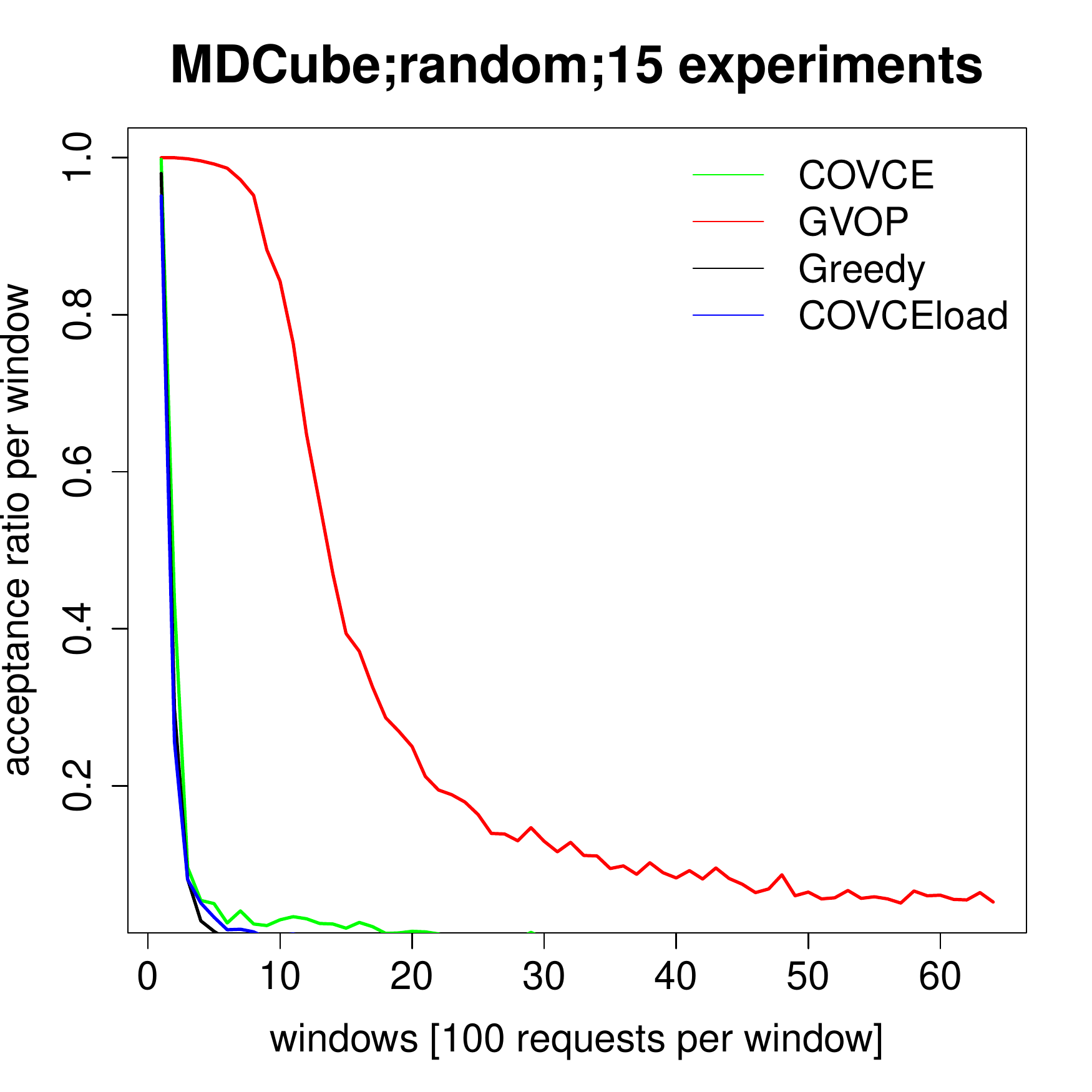}
	\caption{Depicted are the results for the acceptance ratios of the $\mathtt{wave}$ BP (left) and the $\mathtt{random}$ BP (right). The former is averaged over 13 and the latter over 15 experiments for the four algorithms on the MDCube topology over 6400 requests.}
	\label{fig50}
\end{figure}
\begin{figure}[h!]
	\includegraphics[width=0.5\textwidth]{graphics/acceptance_ratio_window_mdcube_reqsize.pdf}
	\includegraphics[width=0.5\textwidth]{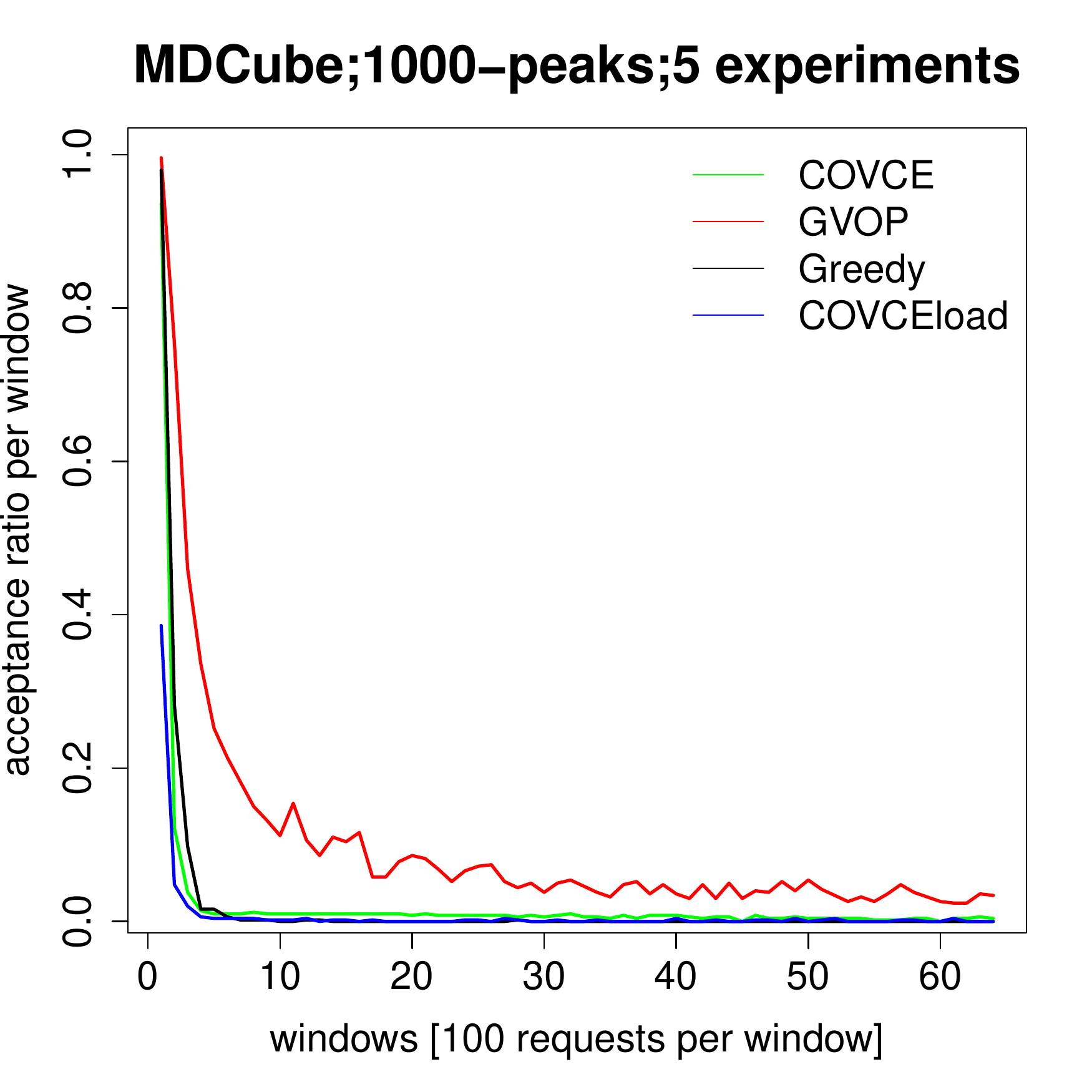}
	\caption{Depicted are the results for the acceptance ratios of the $\mathtt{reqsize}$ BP (left) and the 1000-$\mathtt{peaks}$ BP (right). The former is averaged over 15 and the latter over 5 experiments for the four algorithms  on the MDCube topology over 6400 requests.}
	\label{fig52}
\end{figure}
\begin{figure}[h!]
	\includegraphics[width=0.5\textwidth]{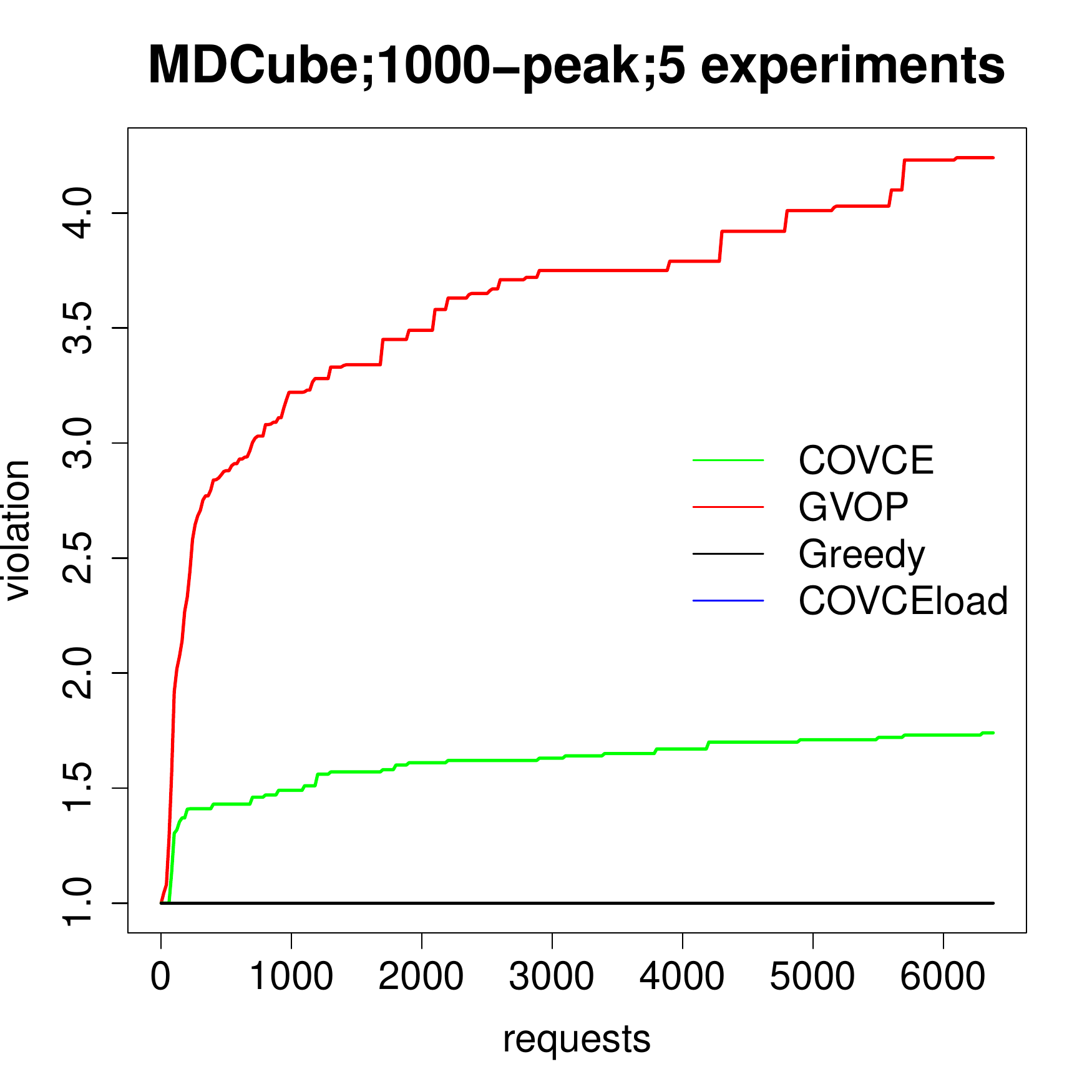}
	\includegraphics[width=0.5\textwidth]{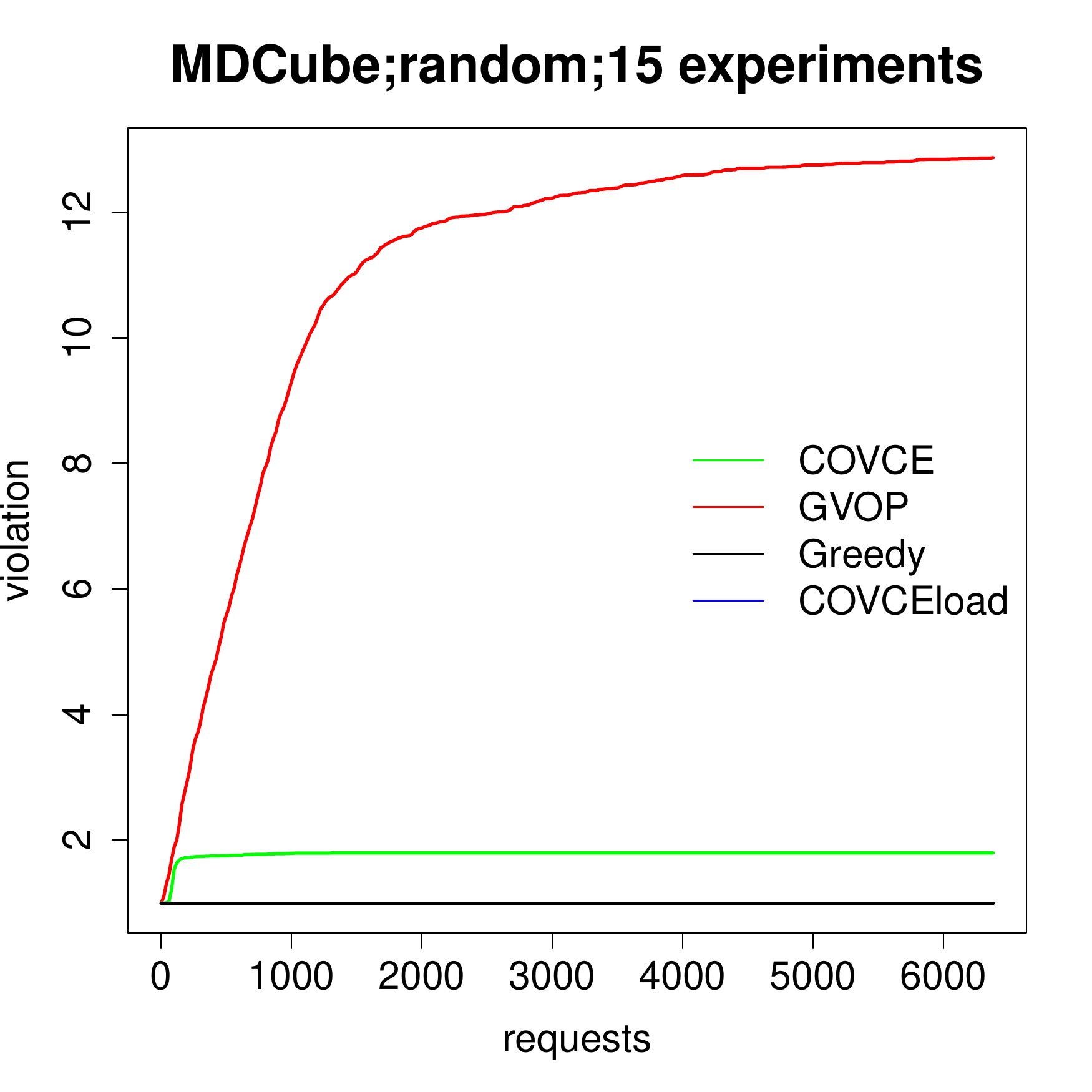}
	\caption{Depicted are the results for the violations of the 1000-$\mathtt{peaks}$ BP (left) and the $\mathtt{random}$ BP (right). The former is averaged over 5 and the latter over 15 experiments for the four algorithms on the MDCube topology over 6400 requests.}
	\label{fig54}
\end{figure}
\begin{figure}[h!]
	\includegraphics[width=0.5\textwidth]{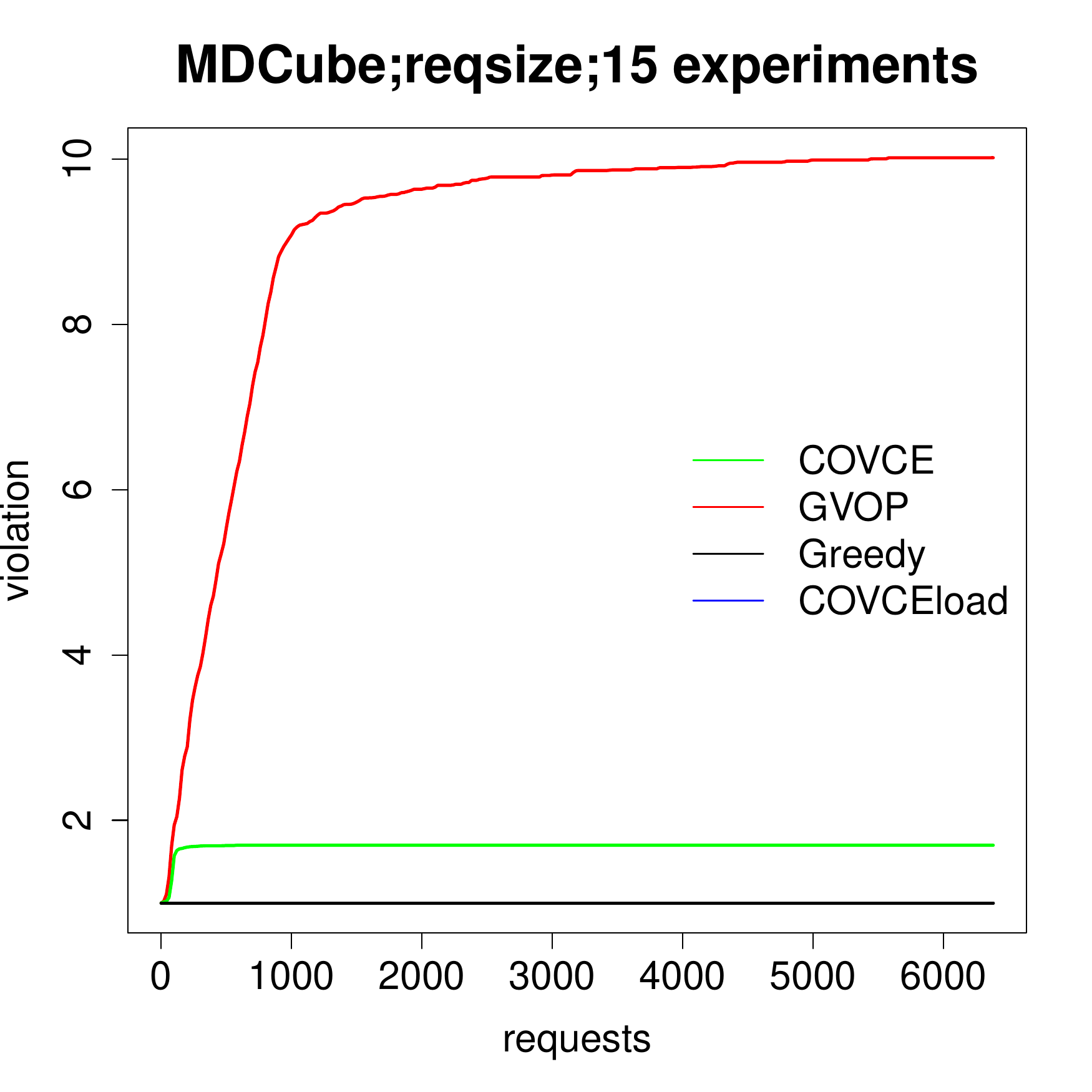}
	\includegraphics[width=0.5\textwidth]{graphics/violation_fattree_wave.pdf}
	\caption{Depicted are the results for the violations of the $\mathtt{reqsize}$ BP (left) and the $\mathtt{wave}$ BP (right). The former is averaged over 15 and the latter over 13 experiments for the four algorithms on the MDCube topology over 6400 requests.}
	\label{fig56}
\end{figure}
\begin{figure}[h!]
	\includegraphics[width=0.5\textwidth]{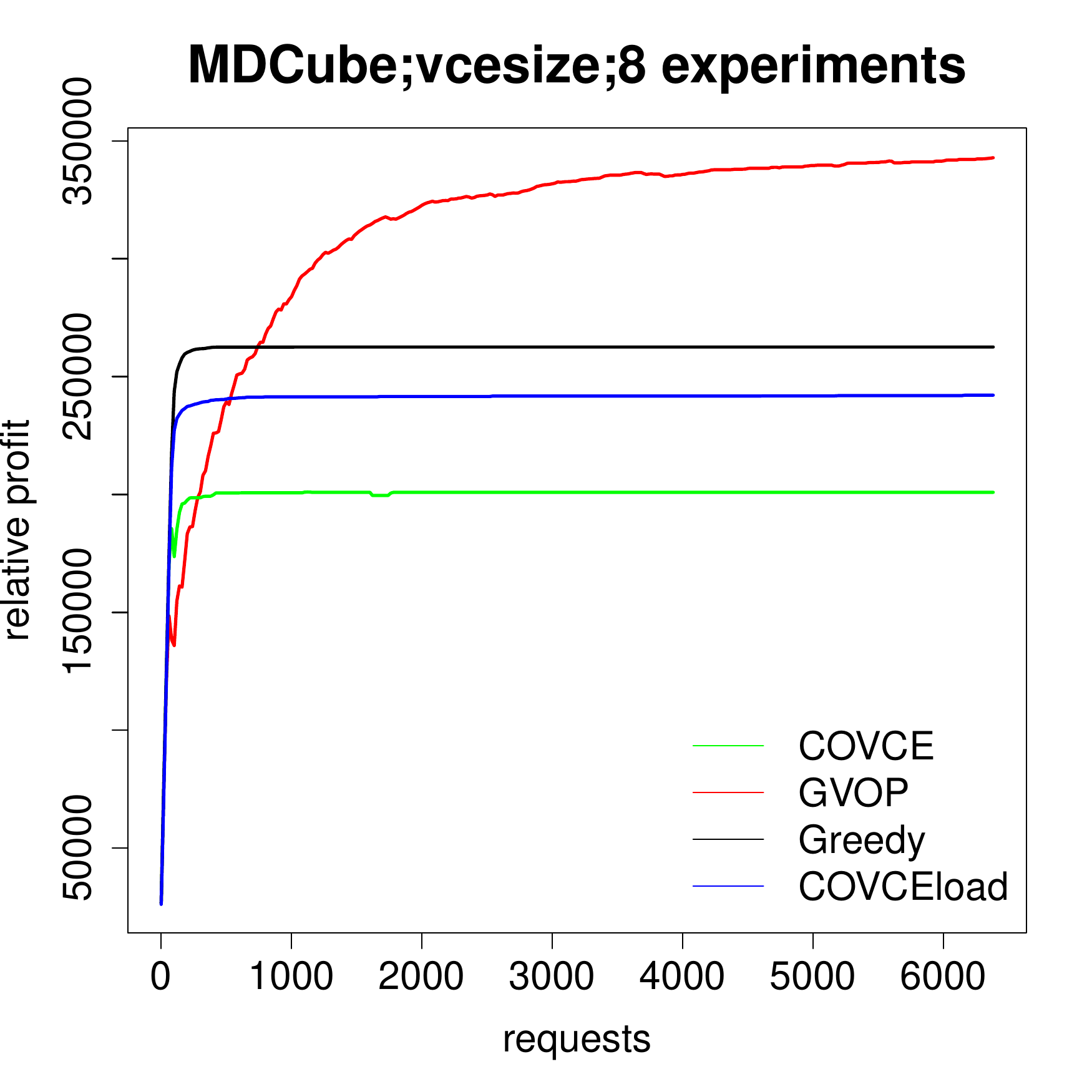}
	\includegraphics[width=0.5\textwidth]{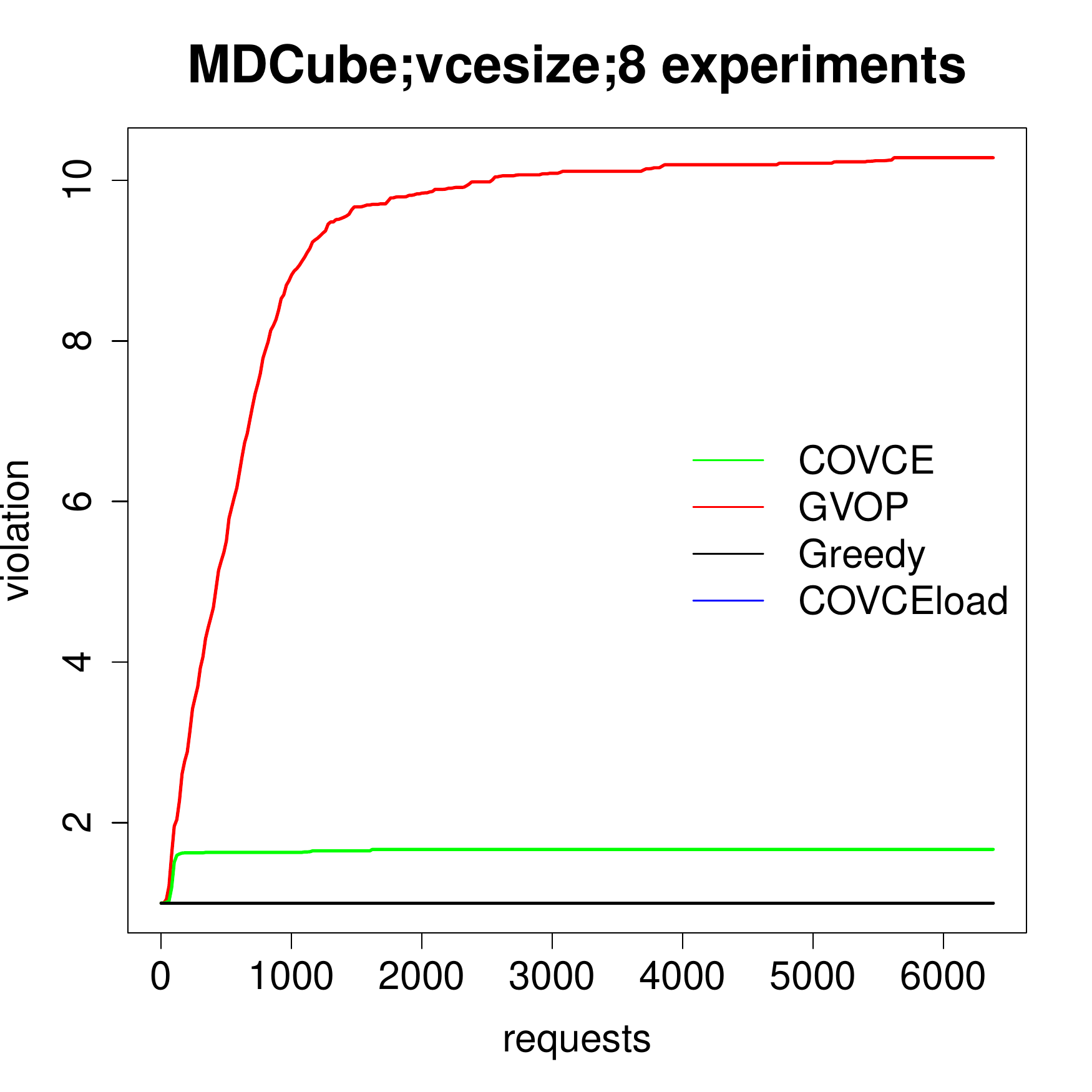}
	\caption{Depicted are the results of the $\mathtt{vcesize}$ BP. Its relative profit is averaged over 8 experiments (left) as well as its violation (right). The four algorithms are run on the MDCube topology over 6400 requests.}
	\label{fig59}
\end{figure}
\begin{figure}[h!]
	\centering
	\includegraphics[width=0.5\textwidth]{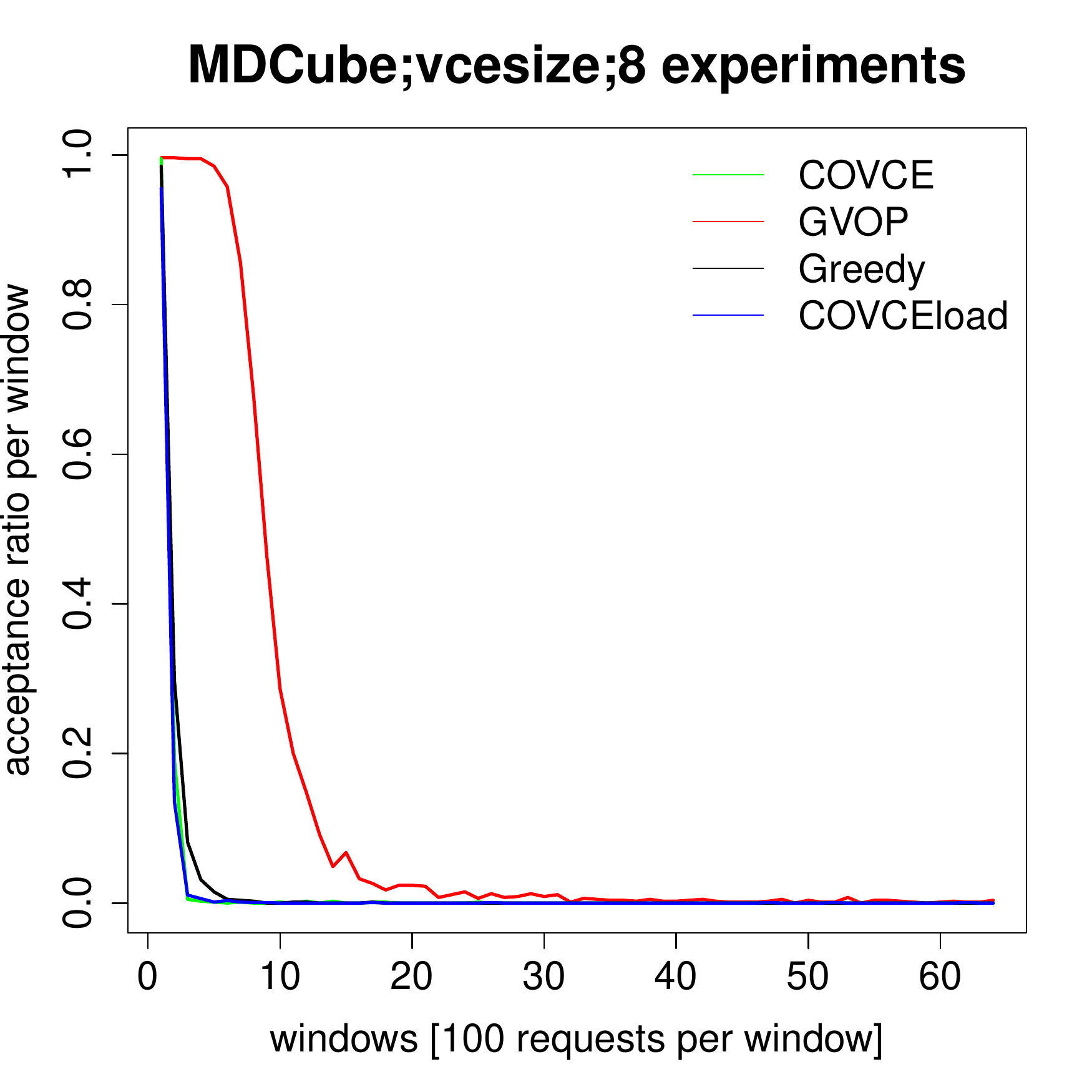}
	\caption{Depicted are the acceptance ratios of the four algorithms for the $\mathtt{vcesize}$ BP. The results of the 8 experiments on the MDCube topology are averaged.}
	\label{fig61}
\end{figure}
To begin with, observe the relative profit of the $\mathtt{random}$ BP experiments (Figure \ref{fig46}) compared to the ones of the fat-tree experiments (Figure \ref{fig6}, page 50). The difference is that the relative profits of the MDCube experiments are for all algorithms by a certain offset higher than the ones of the fat-tree experiments. Another point is that both competitive online algorithms overtakes the Greedy VC-ACE algorithm at later request indexes (in Figure \ref{fig6}). Furthermore, the GVOP algorithm does not outperform the COVCEload algorithm for the shown sequence of requests but probably will do so for a longer sequence. Generally speaking, the intersection points of the curves all happen to be at later request indexes for the MDCube experiments.

When regarding the $\mathtt{reqsize}$ experiments (Figure \ref{fig48}), an offset in the relative profits is again noticeable in comparison to Figure \ref{fig24}. The intersection points between the GVOP algorithm and the COVCEload algorithm also happens to be on a later request index again. In the $\mathtt{wave}$ BP (Figure \ref{fig48}) the same observation as before also take place. The intersection points happen to be on later indexes than for the fat-tree experiments (Figure \ref{fig12}). The COVCE algorithm's relative profit curve does not even pass the one of the greedy algorithm, the same happens for the GVOP algorithm which does not outperform the COVCEload algorithm for the shown sequence. 
The violations of the MDCube experiments (Figure \ref{fig54} and \ref{fig56}) do not show a difference to the fat-tree results.

At last, the acceptance ratios of the GVOP algorithm (Figure \ref{fig50}, \ref{fig52} and \ref{fig61}) are noticeably higher throughout the execution in comparison to the fat-tree results. Consider Figure \ref{fig7} in which the acceptance ratio for the $\mathtt{reqsize}$ BP is shown. The curve of the GVOP algorithm drops earlier than in the corresponding $\mathtt{reqsize}$ BP curve in Figure \ref{fig52}. The same behavior occurs when comparing the $\mathtt{wave}$ (Figure \ref{fig26}) and $\mathtt{random}$ (Figure \ref{fig4}) scenarios. One possible reason for the higher acceptance ratio is that there are additional edges connected to the servers, so that more resources are available to accept requests. Therefore, the higher acceptance ratio additionally explains the higher relative profit. 
The results on the 1000-$\mathtt{peak}$ and the $\mathtt{vcesize}$ BP are not discussed because they do not show any visible differences to the fat-tree results. The 100- and 10-$\mathtt{peak}$ experiments are omitted.  

\subsection{Summary and Discussion}
The main results of the evaluation are that: 
\begin{enumerate}[label=\alph*)]
	\item \label{a}The COVCEload algorithm has the best performance for the $\mathtt{random}$ and $\mathtt{wave}$ BP on both topologies when considering the relative profit, the violation and superlinear increasing capacity multiplication costs.
	\item \label{b}The COVCE algorithm exhibits the best performance for the $\mathtt{peaks}$ BP on both topologies regarding the relative profit, the violation and superlinear increasing capacity multiplication costs.
	\item \label{c}The Greedy VC-ACE algorithm shows the best performance for the $\mathtt{reqsize}$ and $\mathtt{vcesize}$ BPs on both topologies considering the relative profit, the violation and superlinear increasing costs.
	\item \label{d}As a consequence of the above three points, there is a dependency of the competitive algorithms' performance on the benefit to embedding cost ratio.
	\item \label{e}The COVCE algorithm has the best performance for the shown BPs considering the average relative profit, the average violation and superlinear increasing costs.
	\item \label{f}The MDCube results show higher acceptance ratios and consequently higher relative profit values for all algorithms.
\end{enumerate}
The COVCEload algorithm is a heuristic algorithm that involves functionalities of the Greedy VC-ACE and the COVCE algorithm. The combination of these functionalities is advantageous in the mentioned scenarios of \ref{a}. On the one hand, the algorithm utilizes the cost functions of the COVCE algorithm and therefore benefits from the relatively high (compared to $\mathtt{reqsize}$ and $\mathtt{vcesize}$) benefit to embedding cost ratios of the $\mathtt{random}$ and $\mathtt{wave}$ BPs. On the other hand, it does not violate the capacity constraints, so that the profit is not potentially halved in the worst case as for the relative profit of the COVCE algorithm. In addition, maintaining the capacity constraints is advantageous when dealing with superlinear increasing capacity multiplication costs. Thus, the COVCEload is a better choice than the GVOP algorithm.

In the $\mathtt{peaks}$ scenarios, where the most extreme benefit fluctuations of this evaluation are tested, the COVCE algorithm outperforms the other algorithms as stated in \ref{b}. The reason for this is, that the cost function, and therefore the admission control of the algorithm, is dependent on $\alpha$. When the COVCE algorithms experiences a new highest benefit, the costs of the edges and nodes are updated, so that embeddings generally result in higher costs. Thus, subsequent requests with a relatively low benefit are less likely to be embedded. As the very first request in the 1000-$\mathtt{peaks}$ scenario includes the highest benefit of this sequence, namely 1000, the COVCE algorithm is already aware of the highest possible $\alpha$ and therefore rejects many 1-benefit requests. In contrast to this, the GVOP algorithm does not include the benefit into the cost functions. The COVCEload has the same cost functions as the COVCE algorithm but has a worse performance in these scenarios. The capacity violation of factor 2 pays off since at the late stage, meaning after the first capacity violation, of the request sequence, requests with a benefit of 1000 are mostly only accepted while in the early stage many requests with a benefit of 1 are accepted. Hence, the percentage of accepted requests with a benefit of 1 is higher in the early than in the late stage, such that the COVCEload algorithm, which only accepts requests in the early stage, accepted a higher percentage of 1-benefit requests than the COVCE algorithm.

As mentioned in \ref{d}, the benefit to embedding cost ratio is important for the admission control functionality of the competitive algorithms. In the  $\mathtt{reqsize}$ and  $\mathtt{vcesize}$ BPs, the benefit is dependent on the size of the request/embedding, so that the quotient of the benefit and the embedding costs is relatively low in comparison to the  $\mathtt{random}$ and  $\mathtt{wave}$ BPs. This also shows up in the results for the $\mathtt{reqsize}$ and  $\mathtt{vcesize}$ scenarios in which the Greedy VC-ACE algorithm dominates as stated in \ref{c}.

However, the COVCE algorithm has no capacity violation when considering the average violation, such that its performance is comparable to the non-competitive algorithms upon only considering the relative profit. Since the COVCE algorithm has a higher ARP than the non-competitive ones, it outperforms both of them. In general, it has a higher ARP than the non-competitive algorithms. Furthermore, its performance is better than the GVOP algorithm in all BPs when considering the ARP and superlinear increasing costs for the capacity multiplication. Thus, \ref{e} follows.

Lastly, the results on the MDCube experiments show that all algorithms achieved a higher relative profit as stated in \ref{f}. Furthermore, the intersection points between the curves are shifted to later requests indexes, as if the progress of the curves is being delayed. A reason for this is that the utilized (1,12,4,1)-MDCube consists of BCube$_1$ elements which means that each server has $k + 1 = 2$ ports. In contrast to this, each server in a fat-tree only has one port, so that a fat-tree server cannot be included in further embeddings, if its incident edge has no further free bandwidth. Thus, a server in the MDCube in general has more possibilities to be included into an embedding.
\section{Future Work}\label{futurework}
The combination of non-competitive and competitive characteristics (COVCEload) showed to dominate in certain scenarios. Accordingly, designing a ``rejecting''-non-competitive online VCE algorithm is interesting in terms of practical applications since it turned out that extra capacity violations do not pay off for the relative profit but accepting any request, disregarding the benefit to cost ratio, is also not beneficial. Therefore, it is interesting to investigate further heuristic algorithm, such as a ``GVOPload'' algorithm, i.e. passing the residual capacities to the VC-ACE oracle in the GVOP algorithm. 

In addition, testing the mentioned benefit patterns in more details with a higher variety of settings is also important in order to determine more accurate boundaries which show when competitive online algorithms pay off. 

Another interesting point is to add an experiment which runs the optimal offline algorithm besides the online algorithms in order to have a comparison between the optimal solution (the maximum relative profit), the relative profits of the online algorithms and their capacity violations.

From a more theoretical view, it is interesting to investigate whether the 2-competitiveness on the capacity constraints of the COVCE algorithm can be even improved. On the other hand, one can also try to prove that there is no competitive online VCE algorithm that is better than 2-competitive on the capacity constraints. Assume there exists a 1-competitive algorithm on the capacity constraints, it is interesting to investigate how this 1-competitiveness competes with the non-competitive online algorithms in the experiments. As soon as this 1-competitive online algorithm has a higher relative profit, it will already pay off since it does not violate any resources.
\section{Conclusion}\label{conclusion}
In this thesis, competitive online algorithms for virtual cluster embedding problems were studied and discussed. Two competitive online VCE algorithms and two heuristic algorithms were presented. The design of competitive online algorithms based on the primal-dual framework of \cite{survey} was introduced and applied in order to develop the COVCE algorithm which is $(4 \cdot \mathrm{ln}(1+|G_S|\cdot C_{max}\cdot \alpha) + 1, 2)$-competitive. Its competitive ratios were proved. Afterwards, the GVOP algorithm was designed based on \cite{gvop} which is (2, $\mathrm{log}_2 (1+ 3\cdot (\mathrm{max}_{i,k}w(i,k)\cdot \mathrm{max}_{i}b_i))$)-competitive. Both algorithms with their complementary competitive ratio tuples were evaluated together with the Greedy VC-ACE and the COVCEload algorithm. The evaluation consisted of running the four algorithms upon several sequences of requests where the benefit pattern was the most interesting varying parameter. Afterwards the results were interpreted upon several metrics as the relative profit and the average relative profit. 

On the one hand, when applying the relative profit, the Greedy VC-ACE algorithm was found to be the best algorithm in the realistic scenarios ($\mathtt{reqsize}$ and $\mathtt{vcesize}$). On the other hand, the COVCE algorithm had the best performance over all benefit patterns when applying the average relative profit. The GVOP algorithm had the highest (average) relative profit in many scenarios but it was always found to be inapplicable because of its relatively high violation factors. The COVCEload algorithm, which combines the greedy and admission control functionality, has the best performance for the $\mathtt{random}$ and $\mathtt{wave}$ BP when applying the relative profit. Its relative profit was competitive to the Greedy VC-ACE algorithm's one for the realistic scenarios. 

In conclusion, the benefit pattern has a significant influence on the results since in  different scenarios, different algorithms show their best performance with respect to certain metrics.

\end{document}